\documentclass[doublecol]{epl2arxiv} 
\usepackage[utf8]{inputenc}
\usepackage{amsfonts}
\usepackage{amssymb}
\usepackage{mathrsfs}
\usepackage{amsthm}
\usepackage{float}
\usepackage{graphicx}% Include figure files
\usepackage{dcolumn}% Align table columns on decimal point
\usepackage{bm}% bold math
\usepackage{xcolor}
\usepackage{amsmath}
\usepackage[normalem]{ulem}

\usepackage{accents}

\newcommand{\stkout}[1]{\ifmmode\text{\sout{\ensuremath{#1}}}\else\sout{#1}\fi}

\long\def\dm[#1]{\!\operatorname{d\mu}\left(#1\right)}

\newtheorem{rem}{Remark}
\newtheorem{thm}{Theorem}
\newtheorem{lemma}{Lemma}
\newtheorem{prop}{Proposition}
\newtheorem{cor}{Corollary}
\newtheorem{defn}{Definition}

\title{Supervised Hebbian learning}
\shorttitle{Supervised Hebbian learning} %Insert here a short version of the title if it exceeds 70 characters

\author{Francesco Alemanno\inst{1,2} \and Miriam Aquaro\inst{3,4} \and Ido Kanter\inst{5} \and Adriano Barra\inst{1,2} \and Elena Agliari \inst{3,4}}
\shortauthor{F. Alemanno \etal}

\institute{                    
  \inst{1} Dipartimento di Matematica e Fisica, Universit\`a del Salento, Campus Ecotekne, via Monteroni, Lecce 73100, Italy.\\
  \inst{2} Istituto Nazionale di Fisica Nucleare, Sezione di Lecce, Campus Ecotekne, via Monteroni, Lecce 73100, Italy.\\
  \inst{3} Dipartimento di Matematica, Sapienza Universit\`a di Roma, P.le A. Moro 5, 00185, Rome, Italy.\\
  \inst{4} Istituto Nazionale d'Alta Matematica (GNFM), F. Severi, P.le A. Moro 5, 00185, Rome, Italy.\\
  \inst{5} Department of Physics, Bar-Ilan University, Ramat-Gan, 52900, Israel.
}

\abstract{In neural network's Literature, {\it Hebbian learning} traditionally refers to the procedure by which the Hopfield model and its generalizations  {\it store} archetypes (i.e., definite patterns that are experienced just once to form the synaptic matrix). However, the term {\it learning} in Machine Learning refers to the ability of the machine to extract features from the supplied dataset (e.g., made of blurred examples of these archetypes), in order to make its own representation of the unavailable archetypes.
Here, given a sample of examples, we define a supervised learning protocol by which the Hopfield network can infer the archetypes, and we detect the correct control parameters (including size and quality of the dataset) to depict a phase diagram for the system performance.
We also prove that, for structureless datasets, the Hopfield model equipped with this supervised learning rule is equivalent to a restricted Boltzmann machine and this suggests an optimal  and interpretable  training routine. Finally, this approach is generalized to structured datasets: we highlight a quasi-ultrametric organization (reminiscent of replica-symmetry-breaking) in the analyzed datasets and, consequently, we introduce an additional {\em broken replica hidden layer} for its (partial) disentanglement, which is shown to improve MNIST classification from $\sim 75 \%$ to $\sim 95 \%$, and to offer a new perspective on deep architectures.}

\begin{document}

\maketitle

Forty years have elapsed since Hopfield's seminal work, yielding a model for biological information-processing \cite{Hopfield}; meanwhile, we have witnessed a striking development of artificial machine-learning (see e.g., \cite{DL0,Carleo,NoiJPhysA}) and we are finally in a stage where ideas, techniques and results stemming from biological and artificial sides can be fruitfully compared (see e.g., \cite{BarraEquivalenceRBMeAHN,Mezard,Monasson,Kanter1,Ventura}). Here we revise and leverage their analogies to unveil the internal mechanisms of a learning machine, focusing on two paradigmatic models, that is, respectively, the Hopfield neural network (HNN) and the restricted Boltzmann machine (RBM). In order for this comparison to be exhaustive, we first need to profoundly revise the assumptions underlying the theory developed by Amit, Gutfreund and Sompolinsky (AGS) \cite{AGS2}, who, in the eighties, gave a pioneering statistical-mechanical treatment of the HNN based on spin glasses \cite{MPV}. The problem lies in the fact that, in the AGS theory, the HNN actually does not learn, rather it stores definite patterns -- hereafter called {\em archetypes} -- by the so-called Hebb's rule (or countless variations on theme); on the other hand, in standard machine learning the network has to infer these archetypes by solely experiencing (a finite number of) their noisy versions -- hereafter called {\em examples} -- while the original archetypes remain unknown. Hence, in order to match biological and artificial information-processing, we must supply the HNN with examples rather than directly archetypes and therefore turn Hebb's rule into a genuine learning rule. In the following we will reach such a framework, whence we will show that standard machine-learning rules based on contrastive divergence algorithms collapse onto Hebb's learning rule, and we will highlight quantitative control parameters whose tuning determines the learning-machine failure or success. 
These results are obtained analytically by statistical-mechanics tools for random, unstructured datasets, where we can also establish a direct connection between the number of archetypes and the number of hidden neurons in the RBM. As for structured datasets, the robustness of these results is checked numerically for the MNIST and the fashion-MNIST datasets \cite{MNIST,FMNIST} and we also generalize the connection between the size of the hidden-layer(s) and the intrinsic complexity of the dataset, exploiting an iterative rule, reminiscent of the replica-symmetry-breaking (RSB) paradigma \cite{MPV}.

Let us start with the theoretical approach and introduce the information the network has to deal with: we define $K$ archetypes denoted with $\boldsymbol \xi^{\mu}$, $\mu \in \{ 1,...,K \}$, as binary vectors of length $N$ and whose entries are i.i.d. variables drawn from %(like the standard {\em patterns} of the Hopfield model) 
\begin{equation}
\mathcal P (\xi_{i}^{\mu}) = \frac{1}{2}\delta(\xi_{i}^{\mu}-1)+\frac{1}{2}\delta(\xi_{i}^{\mu}+1),
\end{equation}
for any $i \in \{1,\cdots,N\}$ and $\mu \in \{1, ..., K\}$, then, for each of them we generate $M$ examples $\boldsymbol \eta^{\mu a}$, $a \in \{1,...,M\}$, that we obtain by corrupting the archetype flipping its digits randomly as  
\begin{eqnarray} \label{eq:esempi}
\eta_{i}^{\mu a}&=&\xi_{i}^{\mu}\chi_{i}^{\mu a}, \\
\mathcal P (\chi_{i}^{\mu a}) &=&\frac{1+r}{2}\delta(\chi_{i}^{\mu a}-1)+\frac{1-r}{2}\delta(\chi_{i}^{\mu a}+1)
\end{eqnarray}
for any $i, \mu, a$, being $r \in (0,1]$ a parameter tuning the \emph{quality} of the sample.
We now feed the HNN on the dataset $\mathcal S = \{ \boldsymbol \eta^{\mu a}\}_{\mu=1,...,K}^{a=1,...,M}$ and, for this operation to be unambiguous, we also need to specify \emph{how} these examples are presented to the network, mirroring supervised and unsupervised learning.  
In fact, the HNN Hamiltonian reads as $\mathcal H^{\textrm{\tiny{(HNN)}}}(\boldsymbol \sigma | \boldsymbol J)= - \sum_{i<j}^{N,N}J_{ij} \sigma_i \sigma_j$, where $\boldsymbol \sigma = \{\sigma_i \}_{i=1,...,N}  \in \{ -1, +1\}^N$ are $N$ binary neurons and the synaptic connections $J_{ij}$'s incorporate the accessible information: in the original setting, where archetypes are available, the Hebbian (storing) rule reads as $J_{ij} \propto \sum_{\mu} \xi_i^{\mu}\xi_j^{\mu}$, while here $J_{ij} = J_{ij}(\mathcal S)$ and we envisage the following protocols:\\
%\footnote{In the following equations we omit the normalization factors to highlight the different structures.}:\\
%$J_{ij}=J_{ij}(\{\boldsymbol \eta^{\mu,a}\}_{a=1,...,M}^{\mu=1,...,P})$ 
%$J_{ij} = J_{ij}(\mathcal S)$ incorporate the information contained in the dataset. In order to allocate this information we envisage the following %protocols\footnote{In the following equations we omit the normalization factors to highlighten the different structures}:
\newline
{\em Supervised Hebbian learning}. A teacher discloses the example labels and they can therefore be combined as 
\begin{equation}\label{HebbSup}
J_{ij}^{\textrm{sup}} \propto \sum_{\mu=1}^K (\sum_{a=1}^M \eta_i^{\mu a}) (\sum_{b=1}^M \eta_j^{\mu b});   
\end{equation}
{\em Unsupervised Hebbian learning}. Without a teacher that tells how to cluster examples, we mix them up obtaining
\begin{equation}\label{HebbUnsup}
J_{ij}^{\textrm{unsup}} \propto  \sum_{\mu=1}^K \sum_{a=1}^M \eta_i^{\mu a} \eta_j^{\mu a}.
\end{equation}
Clearly, when $r=1$, $M$ becomes a dummy variable because examples coincide with the related archetype and we recover the classical Hebbian rule in both cases. % (storing) rule $J_{ij} \propto \sum_{\mu} \xi_i^{\mu}\xi_j^{\mu}$.
%%%%%%%%%%%%%%%FIGURE%%%%%%%%%%%%%
%
Here we focus on the former \eqref{HebbSup}, while we refer to the Supplementary Material (SM) \cite{SM} for a discussion on the latter \eqref{HebbUnsup}.

A convenient control parameter to assess the information content in $\mathcal S$ is $ \rho := \frac{1-r^2}{M r^2}$. To see this, let us focus on the $\mu$-th pattern and the $i$-th digit, whose related block is $\boldsymbol \eta_i^{\mu}= (\eta_i^{\mu 1},\eta_i^{\mu 2},...,\eta_i^{\mu M})$; the error probability for any single entry is $\mathcal P(\chi_i^{\mu a} = -1) = (1-r)/2$ and, by applying the majority rule on the block, it is reduced to 
$\mathcal P (\textrm{sgn}(\sum_a \chi_i^{\mu a}) = -1 ) \underset{M \gg 1}{\approx} [1 - \textrm{erf}(1 / \sqrt{2 \rho})]$, thus, the
 conditional entropy $H(\xi_i^{\mu} | \boldsymbol \eta_i^{\mu})$, that quantifies the amount of information needed to describe the original message $\xi_i^{\mu}$, given the related $M$-length block $\boldsymbol \eta_i^{\mu}$, is monotonically increasing with $\rho$, saturating to 1 bit.
 %\footnote{\red{With a little abuse, in the following we will refer to $\rho$ as the dataset ``entropy''.}}. 
 Hence, in order for the dataset to retain information on the original archetypes, $\rho$ must be finite, that is, $M r^2$ must be non-vanishing. 

This scaling, arising from an information-theory perspective, is recovered and sharpened in the neural network framework.
We start with the signal-to-noise analysis on the HNN to check for local stability of the archetype-retrieval configurations in the noiseless limit, that is, we study the conditions under which the internal field $h_i(\boldsymbol \sigma)=\sum_{\underset{j\neq i}{j =1 }}^{N}J_{ij} (\mathcal S)\sigma_j$, namely the post-synaptic potential experienced by the neuron $i$, is aligned with the neural activity $\sigma_i$ while $\boldsymbol \sigma = \boldsymbol \xi^{\mu}$, for any arbitrary $\mu$ and $i$. As detailed in the SM \cite{SM}, this approach can be recast into the one-step contrastive divergence scheme \cite{Hinton1MC} and returns 
$\frac{K}{N} \left(1+\frac{1-r^{2}}{Mr^{2}}\right)^{2}+\frac{1-r^{2}}{Mr^{2}}\lesssim1$;
%
%Thus, we find that the conditions for the network to robustly retrieve the archetypes read as 
%\begin{eqnarray}
%\label{supsup}
%\textrm{supervised}: &~& \alpha\left(1+\frac{1-r^{2}}{Mr^{2}}\right)^{2}+\frac{1-r^{2}}{Mr^{2}}\lesssim1\\
%\label{2stability}
%& \Rightarrow& M_c^{\textrm{sup}}(r) \sim \mathcal O(r^{-2}) ~ \textrm {as} ~ r \to 0.
%\end{eqnarray}
%%
%\begin{eqnarray}
%\label{unsupunsup}
%\textrm{unsupervised}: &~&\alpha\left(1+\frac{1-r^{4}}{Mr^{4}}\right)+\frac{1-r^{2}}{Mr^{2}}\lesssim1 \\
%\label{eq:stability_giordano}
%&\Rightarrow& M_c^{\textrm{unsup}}(r) \sim \mathcal O(r^{-4}) ~\textrm {as} ~ r \to 0.
%\end{eqnarray}
%
%\begin{eqnarray}
%\label{supsup}\label{2stability}
%\frac{K}{N} \left(1+\frac{1-r^{2}}{Mr^{2}}\right)^{2}+\frac{1-r^{2}}{Mr^{2}}\lesssim1; %\Rightarrow& M_c(r) \sim \mathcal O(r^{-2}) ~ \textrm {as} ~ r \to 0.
%\end{eqnarray}
%Their agreement with Monte Carlo simulations is shown in Figure \ref{Fig:scalings}; 
this relation advises on the suitable rescaling of the dataset size ($M \gtrsim r^{-2}$), as the dataset quality is impaired ($ r \to 0$), in order to preserve network's abilities: note that power-law scalings were already evidenced in the machine-learning context, see e.g. \cite{Kanter-PowerLaw}. 
%\newline
To achieve a quantitative picture and control of the network behavior, we work out a statistical-mechanics investigation and we start by introducing
the Boltzmann-Gibbs measure for the system: 
\begin{equation}
\mathcal P_{\beta} (\boldsymbol \sigma | \mathcal S ) = \frac{1}{Z^{\textrm{\tiny{(HNN)}}}_{\beta}( \mathcal S)} e^{-\beta \mathcal H^{\textrm{\tiny{(HNN)}}}(\boldsymbol \sigma | \boldsymbol J (\mathcal S))}
\end{equation}
where $Z^{\textrm{\tiny{(HNN)}}}_{\beta}$ is the partition function and $\beta:=1/T \in \mathbb R^+$ tunes the distribution broadness; $\beta$ along with the load $\alpha := \lim_{N \to \infty} K/N$ and the dataset ``entropy'' $\rho=(1-r^2)/Mr^2$, make up the set of control parameters. Further, we introduce the macroscopic observables (order parameters) useful to describe the system behavior, namely 
\begin{eqnarray} \label{Mattis}
m &:=& \frac{1}{N} \sum_{i=1}^N \xi_i^1 \sigma_i,\\ \label{Mattel}
n&:=&\frac{1}{r(1+\rho)}\frac{1}{NM}\sum_{i,a=1}^{N,M}\eta_i^{1a}\sigma_i,\\ 
\label{overlaps_q}
q_{12}&:=&\frac{1}{N}\sum_{i=1}^N \sigma_i^{(1)} \sigma_i^{(2)},\
%\label{overlaps_p}
% p_{12}&:=&\frac{1}{P}\sum_{\mu}^K z_{\mu}^{(1)} z_{\mu}^{(2)}
\end{eqnarray}
where we defined, respectively, the Mattis magnetization of the archetype  (eq. \ref{Mattis}), the typical magnetization of the example (eq. \ref{Mattel}), and the two-replica overlap (eq. \ref{overlaps_q}); for $m$ and $n$ we referred to $\mu=1$ without loss of generality.
%\newline
%The control parameters are the load $\alpha := \lim_{N \to \infty} P/N$, the fast noise $\beta \in \mathbb R^+$, and the  $\rho=(1-r^2)/Mr^2$ ($\beta \in \mathbb{R}^+$  is the {\em inverse temperature} \cite{Amit,CKS} and tunes the stochasticity in the network). 

\begin{figure}
\begin{center}
\includegraphics[scale=0.25]{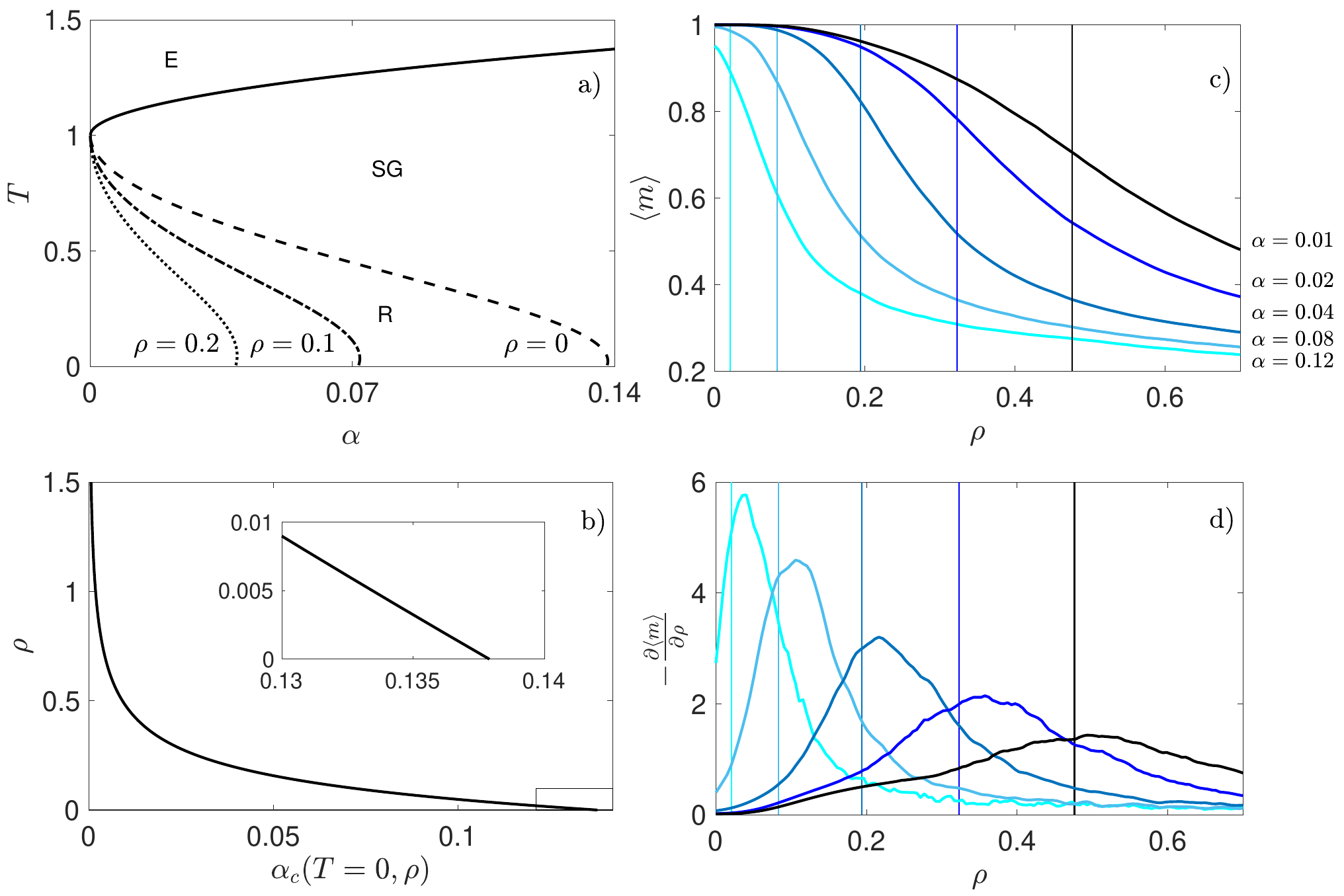}
\caption{Behaviour of the supervised HNN as the control parameter are varied. Panel $a$: phase diagram highlighting the ergodic (E), the spin-glass (SG) and the retrieval (R) phase versus $T$ and $\alpha$; the transition line between the SG phase and the R phase depends on $\rho$ and three cases are shown: $\rho =0$ (dashed line, corresponding to AGS theory), $\rho=0.1$ (dashed-dotted line), and $\rho=0.2$ (dotted line). Panel $b$: critical load $\alpha_c$ obtained for $T=0$ and as a function of $\rho$. Panel $c$: Estimate of the Mattis magnetization versus $\rho$ by MC simulations for systems of size $N=5000$; different loads are considered and plotted in different colors (brighter nuances correspond to larger values of $\alpha$, as reported on the right); the vertical lines represent the transition points as predicted by statistical mechanics. Panel $d$: From data presented in panel $c$ we derive the susceptibility w.r.t. $\rho$ and notice that the peaks approximately match the transition points (by a finite-size scaling we checked that the match gets closer as $N$ is made larger). }
\label{Fig:Uno}
\end{center}
\end{figure}

Under the replica-symmetry (RS) ansatz, all the order parameters do not fluctuate in the thermodynamic limit, i.e., being $\mathcal{P}(x)$ the probability distribution for the observable $x = (m, n,q_{12})$ and $\langle x \rangle$ its expectation, then $\lim_{N \to \infty}\mathcal{P}(x)=\delta(x-\langle {x} \rangle)$. These expectation values can be obtained by extremizing the quenched free-energy of the model w.r.t. the order parameters and, as explained in  the SM \cite{SM}, for $N \to \infty$ and $M \gg 1$, we obtain the following set of self-consistent equations 
%%%%%%%%%%%%%%%FIGURE%%%%%%%%%%%%%
%
%
%
\begingroup\makeatletter\def\f@size{8.5}\check@mathfonts
\def\maketag@@@#1{\hbox{\m@th\small\normalfont#1}}%
\begin{eqnarray}
\label{eq:sce_m}
\langle{m} \rangle &=&\mathbb{E}_{z}\tanh\left\{\beta \langle{n} \rangle+z\beta\sqrt{ \langle{n} \rangle^{2}\rho+\frac{\alpha \langle{q} \rangle}{\left[1-\beta\left(1- \langle{q} \rangle\right)\right]^{2}}}\right\}\\
\label{eq:sce_mu}
\langle{n} \rangle &=&\frac{ \langle{m} \rangle}{(1+\rho)-\rho \beta(1- \langle{q} \rangle)}\\
\label{eq:sce_ultima}
\langle{q} \rangle &=&\mathbb{E}_{z}\tanh^{2}\left\{\beta \langle{n} \rangle+z\beta\sqrt{ \langle{n} \rangle^{2}\rho+\frac{\alpha \langle{q} \rangle}{\left[1-\beta\left(1- \langle{q} \rangle\right)\right]^{2}}}\right\}
\end{eqnarray}
\endgroup
where $\mathbb E_{z}$ denotes the average w.r.t. the standard Gaussian variable $z$.
The inspection of eqs.~(\ref{eq:sce_m})-(\ref{eq:sce_ultima}) provides a quantitative picture of the system behavior in the space of the control parameters as reported in Fig.~\ref{Fig:Uno} $a$-$b$. In particular, like in the classical HNN, we recognize the emergence of an ergodic region corresponding to large values of $T$ and a retrieval region for relatively small values of $\alpha$ and $T$, yet, the Hebbian learning rule \eqref{HebbSup} makes the phenomenolgy much richer: here we have an additional tuneable parameter $\rho$ which controls the width of the retrieval region. Denoting with $\alpha_c(T, \rho)$ the first-order transition line between the spin-glass phase and the retrieval phase, we show that $\alpha_c(T=0, \rho)$ is a decreasing function of $\rho$ and, as expected, $\alpha_c(T=0, \rho=0) \approx 0.138$, consistently with the AGS theory. Signatures of this transition are also found by means of finite-size Monte Carlo (MC) simulations as shown in Fig.~\ref{Fig:Uno} $c$-$d$. 
%Equations (\ref{eq:sce_m})-(\ref{eq:sce_ultima}) can be further treated  to get an explicit expression in the zero fast-noise limit $T \to 0$ that is,
%
%\begin{equation}
%\begin{aligned} \langle{m} \rangle & = \langle{n} \rangle\left[1+\rho(1-\Delta)\right]\\
% \langle{n} \rangle & =\frac{1}{1+\rho(1-\Delta)}\mathrm{erf}\left(\frac{ \langle{n} \rangle}{G}\right)\\
%\Delta & =  \frac{2}{\sqrt{\pi} G} \exp \Big[ - \Big( \frac{\langle n \rangle }{G} \Big)^2 \Big]\\
%G & =\sqrt{2\langle n \rangle^{2}\rho+2\frac{\alpha}{\left(1-\Delta\right)^{2}}}.
%%\Delta & =\frac{1}{G}\partial\mathrm{erf}\left(\frac{ \langle{n} \rangle}{G}\right),\\
%\end{aligned}
%\label{eq:sce_Tzero}
%\end{equation}
%From these equations, requiring a non-vanishing magnetization $\langle m \rangle$, we derive that $\rho$ must be finite and therefore we recover the scaling \eqref{HebbSup}.
%
Further, looking at eqs.~(\ref{eq:sce_m})-(\ref{eq:sce_ultima}) and requiring a non-vanishing magnetization $\langle m \rangle$, we derive that $\rho$ must be finite and therefore we recover the scaling $M \sim r^{-2}$; also, in the zero fast-noise limit $T \to 0$, these equations can be treated to get explicit expressions as achieved in  the SM \cite{SM}.

We now bridge this theory with the machine learning counterpart. We consider a RBM made of two layers, a visible one endowed with $N$ binary neurons $\boldsymbol \sigma = \{ \sigma_i \}_{i=1,...,N} \in \{-1,+1\}^N$, and a  hidden one built of $K$ real-valued neurons $\boldsymbol z = \{z_\mu\}_{\mu=1,...,K} \in \mathbb{R}^K$ with a Gaussian prior, and whose Hamiltonian reads as $\mathcal H^{\textrm{\tiny{(RBM)}}}(\boldsymbol \sigma , \boldsymbol z | \boldsymbol W)=-\sum_{i,\mu}^{N,K} W_{i,\mu}\sigma_i z_{\mu}$.  We choose the length of the hidden layer to match the number of archetypes in such a way that, as we will see, we can assign to each hidden neuron the recognition of a unique archetype. % (see Fig.~\ref{cartoon}). 
The Boltzmann-Gibbs distribution associated to $\mathcal H^{\textrm{\tiny{(RBM)}}}$ is
\begin{equation}
\label{eq:BG_RBM}
\mathcal {P}_{\beta} (\boldsymbol \sigma, \boldsymbol z | \boldsymbol W) =  \frac{1}{Z^{\textrm{\tiny{(RBM)}}}_{\beta}(\boldsymbol W)} e^{- \beta \mathcal H^{\tiny{\textrm{(RBM)}}}(\boldsymbol \sigma , \boldsymbol z | \boldsymbol W) - \beta \frac{\boldsymbol {z}^2}{2} }.
\end{equation}
Now, the goal is to find the weight setting such that this measure mimics the target one, referred to as $\mathcal {Q}$, which generated the examples in $\mathcal S$. 
Focusing on a classification task, we adopt the so-called grandmother-cell scheme: during training, the generic input-output pair is $(\boldsymbol \eta^{\nu a} , \boldsymbol z^{(\nu)})$, where  $\boldsymbol z^{(\nu)}$ is the one-hot vector whose $\nu$-th entry is the single non-null entry \cite{Leonelli,Giordano}.
%\begin{eqnarray}
%\mathcal Q^{\text{sup}}(\boldsymbol \sigma , \boldsymbol z) &=& \sum_{\mu,a} \delta (\boldsymbol \eta^{\mu a} - \boldsymbol \sigma) \delta (\boldsymbol z^{(\mu)} - \boldsymbol z),\\
%\mathcal Q^{\text{unsup}}(\boldsymbol \sigma) &=& \sum_{\mu,a} \delta (\boldsymbol \eta^{\mu a} - \boldsymbol \sigma), 
%\end{eqnarray}
{Thus, the target distribution reads as}
\begin{eqnarray} \label{Q_sup}
\mathcal Q(\boldsymbol \sigma , \boldsymbol z) &=& \sum_{\mu,a} \delta (\boldsymbol \eta^{\mu a} - \boldsymbol \sigma) \delta (\boldsymbol z^{(\mu)} - \boldsymbol z),
\end{eqnarray}
and, if training is successful, we expect that, initializing the visible layer as a test example $\tilde{\boldsymbol \eta}^{\nu}$ of the $\nu$-th archetype and letting the neurons evolve freely up to thermalization, the hidden layer will provide the estimated class as $\textrm{argmax} [ \langle \boldsymbol z(\tilde{\boldsymbol \eta}^{\nu}) \rangle]$.
\newline
The learning rule can be derived by a gradient descent on the Kullback-Leibler (KL) cross entropy $D_{\textrm{KL}}(\mathcal Q \Vert \mathcal P)$ between the distributions  $\mathcal Q$ and $\mathcal P$, that is, $W_{i,\mu}^{n+1} = W_{i,\mu}^{n} -\epsilon \frac{dD_{\textrm{KL}}(\mathcal Q \Vert \mathcal P)}{dW_{i,\mu}}$, where $n$ accounts for training iterations and $\epsilon$ is the learning rate; recalling \eqref{Q_sup} this yields
\begin{equation}\label{RBMsup}
W_{i,\mu}^{n+1} = W_{i,\mu}^{n} + \epsilon \left( \langle \sigma_i z_{\mu} \rangle_{\boldsymbol \sigma \& \boldsymbol z} -  \langle \sigma_i z_{\mu} \rangle \right),
\end{equation}
where the brackets denote the expectation under the Boltzmann-Gibbs measure \eqref{eq:BG_RBM} and the bracket subscript specifies the clamped variables. 
\newline
In the case of orthogonal patterns, the configuration where weight entries are set as the empirical average of example entries, i.e., $W_{i \mu} = \bar{\eta}_{i \mu}:= \frac{1}{M} \sum_{a=1}^M \eta_i^{\mu a}$, is a fixed point for the contrastive divergence and therefore compatible with a trained machine \cite{Leonelli,SM}.
%\red{To see this, let us}
%%
%look at the conditional probabilities: \red{$\mathcal P_{\beta} (\boldsymbol z | \boldsymbol \sigma = \boldsymbol \eta^{\nu a}, \boldsymbol W) \propto \prod_{\mu} e^{\red{-}\beta (z_{\mu} - (\boldsymbol W \cdot \boldsymbol \eta^{\nu a})_{\mu})^2/\red{2}}$ should be peaked at $ \boldsymbol z^{(\nu)}$, while  $\mathcal P_{\beta} (\boldsymbol \sigma | \boldsymbol z = \boldsymbol z^{(\nu)}, \boldsymbol W) \propto \prod_i e^{\beta \sigma_i W_{i \nu}}$ should be peaked at $\boldsymbol \xi^{\nu}$, therefore $\boldsymbol W_{\mu}$ should be simultaneously orthogonal to $\boldsymbol \eta^{\nu a}$ for any $\nu \neq \mu$ and parallel to $\boldsymbol \xi^{\mu}$ for any $\mu$} \red{, so the optimal solution is given by $\boldsymbol W = \bar{\boldsymbol \eta}$.} 
Further, with this choice we can prove that the RBM is equivalent, in distribution, to the HNN with supervised Hebbian rule; in fact, by a Gaussian integration
\begin{eqnarray} %\nonumber
Z^{\textrm{\tiny{(RBM)}}}_{\beta}(\boldsymbol W = \bar{\boldsymbol \eta})&=& \sum_{\sigma}^{2^N}\int e^{\frac{\beta}{\sqrt{N}} \sum_{\mu} (\sum_i \sigma_i \bar{\eta}_{i \mu} )z_{\mu} ) - \frac{\beta z_{\mu}^2}{2} } \\
\nonumber
&=& \sum_{\sigma}^{2^N} e^{\frac{\beta}{2N} \sum_{\mu} \sum_{ij} \sigma_i \bar{\eta}_{i \mu} \bar{\eta}_{j \mu} \sigma_j}=Z^{\textrm{\tiny{(HNN)}}}_{\beta}(\mathcal S).
\end{eqnarray}
%as long as we identify, once learning is over (whatever the machine), the weights in the former as the patterns in the latter, i.e $W_{i,\mu}=\xi_i^{\mu}$ (for a broader treatment accounting for RBM in general see \cite{Peter1,Peter2}). 
%\newline
%Here we focus on the supervised setting where we can get theoretical insights thanks to the formal equivalence $w_{i\mu} = \bar {\eta}$, while for the unsupervised setting we refer to \cite{Giordano}.\\
%If we had access to the target distribution $Q = \sum_{\mu} =\delta(\boldsymbol \sigma - \boldsymbol \xi^{\mu})$ then, the probability distribution $P$ that minimizes the distance, given the prior $z \sim \mathcal N(0,1)$ and $\beta \to \infty$ is just the Boltzmann Gibbs distribution with $w = \xi$. If we approximate the target distribution with the empirical distribution $\tilde{Q} = \sum_{\mu} =\delta(\boldsymbol \sigma - \bar{\boldsymbol \eta^{\mu}})$ then $w_{i\mu} = \bar {\eta}$.
%\newline
%

\begin{figure}
\begin{center}
\includegraphics[scale=0.24]{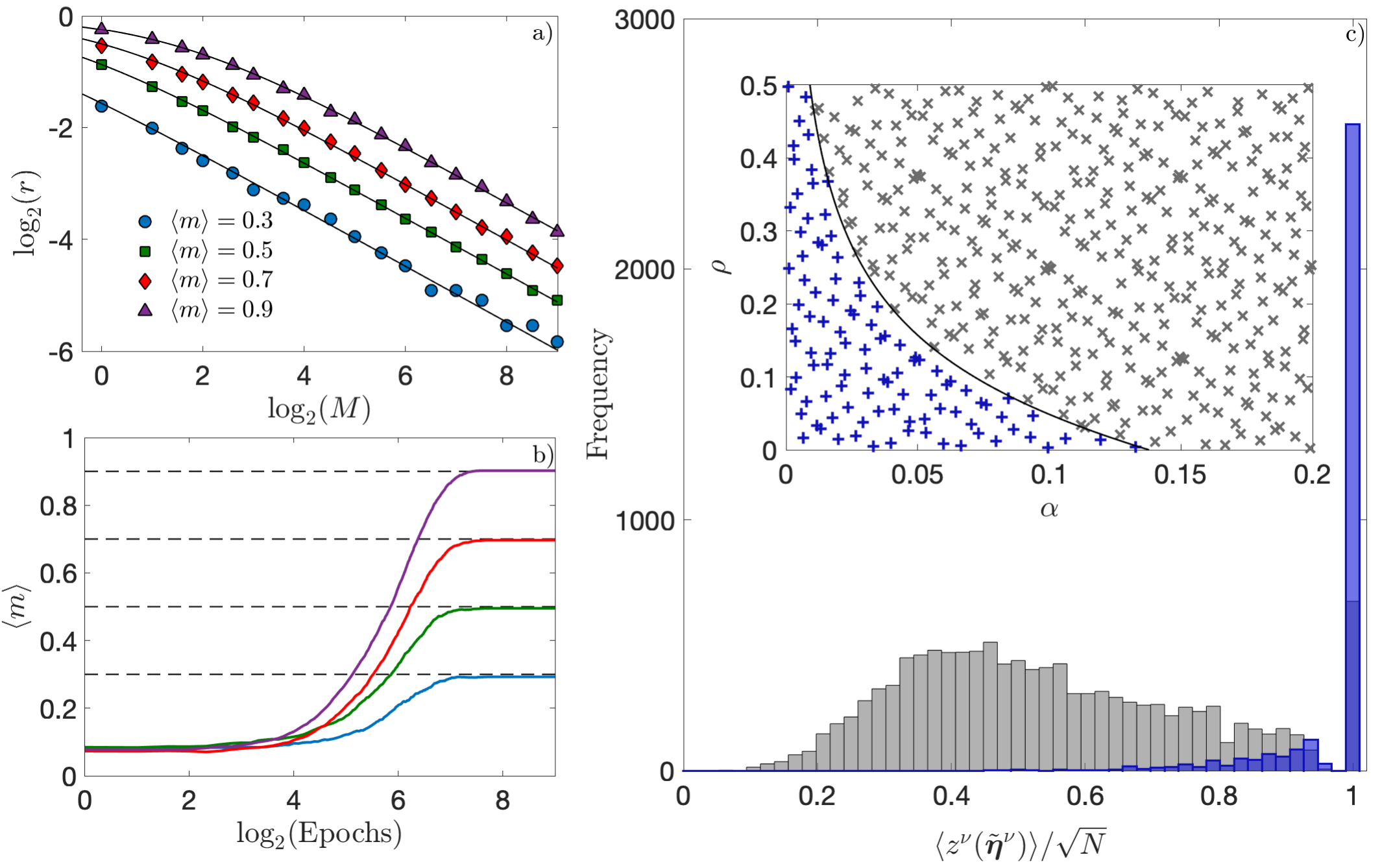}
\caption{Comparison between HNN and RBM performances. Panel $a$: we fix a certain value for the expected magnetization $\langle m \rangle$ and we derive from Eq.~\eqref{eq:sce_m}, obtained theoretically for the HNN, how $r$ and $M$ should be tuned in order to retain this value constant (solid line);  an analogous analysis is repeated numerically for the RBM where now $m$ is evaluated as the overlap between the visible layer and a given archetype (symbols); different values of magnetization are considered and represented with different symbols. Panel $b$: Expected value of the RBM magnetization versus the training time and for given values of $r$ and $M$, under on-line contrastive divergence (CD-1) \cite{Hinton1MC}; the long-time value corresponds to the theoretical estimate obtained for the HNN for the same choice of $r$ and $M$ (horizontal lines). Panel $c$: we sampled $1.5\times 10^4$ couples $(\alpha, \rho) \in (0,0.2) \times (0, 0.5)$ by Sobol's low-discrepancy sequence; for each extraction (represented by a cross in the inset) we build a RBM of size $N=5000$ and $K = \alpha N$, we generate a set $\mathcal S$ of examples and we set the machine weights as $\boldsymbol W = \bar{\boldsymbol \eta}$. Then, we initialize the visible layer as a test example $\tilde{\boldsymbol \eta}^{\nu}$, we run MC simulations and we evaluate $\langle z_{\nu} \rangle$, whose histogram is depicted in the main plot, distinguishing between cases inside (blue) and outside (grey) the retrieval region.}
\label{comparison}
\end{center}
\end{figure}
This equivalence implies that the phase diagram outlined for the HNN (see Fig.~\ref{Fig:Uno} $a$-$b$) also applies to the RBM, as confirmed in Fig.~\ref{comparison}. In particular, the retrieval region corresponds to a parameter setting where the trained RBM relaxes to configurations such that the overlap between the visible layer and the archetype are close to one. Remarkably, this is consistent with the usual  performance and score values \cite{SompoLearning} or error-based measures as in the Vapnik-Chervonenkis learning theory \cite{VCL} where one aims to minimize the distance between the output and the instances of a test set. In fact, $-(\boldsymbol \sigma - \boldsymbol \xi^{1})^2 \propto \boldsymbol \sigma \cdot \boldsymbol \xi^{1} = m$ and $-\overline{(\boldsymbol \sigma - \boldsymbol \eta^{1})^2} \propto n$: whenever the network is in the retrieval region, for some archetype $\mu$ it is minimizing one of these Loss functions $L^{\mu}_{\pm}=(1/2N) ||\bold{\xi^{\mu}} \pm \bold{\sigma}||^2 = 1 \pm m_{\mu}$ as the Hopfied Hamiltonian can be written as $\mathcal H^{\textrm{\tiny{(HNN)}}}(\boldsymbol \sigma | \boldsymbol J)= -N\sum_{\mu}^K (1 - L^{\mu}_+L^{\mu}_-)$ (in fact it learns both the pattern $\xi^{\mu}$ and its gauge symmetric copy $-\xi^{\mu}$).

In order to appreciate further the equivalence between HNN and RBM, we show that it can be reached from a different perspective, namely using the maximum-entropy principle, according to Jaynes' inferential interpretation \cite{CKS,Jaynes}. Let us look for the least structured probability distribution $\mathcal P(\boldsymbol \sigma, \boldsymbol z)$ that is compatible with the set of data $\{(\boldsymbol \eta^{\mu a}, \boldsymbol z^{(\mu)})\}_{\mu=1,...,K}^{a=1,...,M}$  to inspect which kind of correlations  the machine detects in the dataset. While extensive calculations are provided in  the SM \cite{SM}, here we report the main findings: the minimal constraints needed to recover the HNN and RBM's Boltzmann-Gibbs distribution concern the variance of hidden units and the correlations between visible and hidden units -- set equal to their empirical estimates $C_{z_{\mu}^2}$ and $C_{\sigma, z}^{i,\mu}$ for $i=1,...,N$ and $\mu =1,...,K$, respectively -- beyond those for $\mathcal{P}$ to be well defined. The constrained optimization problem therefore reads: $\max_{\{\lambda_0, \lambda_1, \Lambda_{i,\mu} \}_{i,\mu} }S[\mathcal{P}]$ with
%
%\begin{eqnarray}
%%&&\red{\max_{\{\lambda_0, \lambda_1, \Lambda_{i,\mu} \}_{i,\mu} }S[\mathcal{P}]}\\
%S[\mathcal{P}] &=&-\textrm{Tr}\left( \mathcal{P} \ln \mathcal{P} \right) + \lambda_0 \left( \textrm{Tr} \mathcal{P} - 1 \right) +  \\ \nonumber
%+ & \lambda_1 & ( \textrm{Tr}(\sum_{\mu=1}^K z_{\mu}^2 )  -   C_{z_{\mu}^2}) + \sum_{i,\mu}^{N,K}\Lambda_{i,\mu}\left[ \textrm{Tr}(\sigma_i z_{\mu}) - C_{\sigma, z}^{i,\mu} \right]
%\end{eqnarray}
%
\begin{eqnarray}
%&&\red{\max_{\{\lambda_0, \lambda_1, \Lambda_{i,\mu} \}_{i,\mu} }S[\mathcal{P}]}\\
S[\mathcal{P}] &=&- \langle  \mathcal{P} \ln \mathcal{P} \rangle_{\mathcal{P}} + \lambda_0 \left(  \langle \mathcal{P} \rangle_{\mathcal{P}}  - 1 \right) +  \\ \nonumber
+ & \lambda_1 & (  \langle \sum_{\mu=1}^K z_{\mu}^2  \rangle_{\mathcal{P}}  -   C_{z_{\mu}^2}) + \sum_{i,\mu}^{N,K}\Lambda_{i,\mu}\left[ \langle \sigma_i z_{\mu}  \rangle_{\mathcal{P}} - C_{\sigma, z}^{i,\mu} \right]
\end{eqnarray}
where $\langle \cdot  \rangle_{\mathcal{P}}$ denotes the expectation over $ \mathcal{P}$.
%we introduced the trace operator $\text{Tr}[O(\boldsymbol{\sigma},\boldsymbol{z})$ to denote the expectation over $ \mathcal{P}$.
%$\text{Tr}[O(\boldsymbol{\sigma},\boldsymbol{z})]:=\sum_{\boldsymbol \sigma}\int\left[\prod_{\mu=1}^{K}\frac{dz_{\mu}}{\sqrt{2\pi}}\right]O(\boldsymbol{\sigma},\boldsymbol{z})$.
%$\langle O \rangle_{\mathcal{P}}= \textrm{Tr}[O(\boldsymbol \sigma, \boldsymbol z) \mathcal{P}(\boldsymbol \sigma, \boldsymbol z)]$ 
%to lighten the notation. 
The solution yields the following Lagrange multipliers:
\begingroup\makeatletter\def\f@size{8.5}\check@mathfonts
\def\maketag@@@#1{\hbox{\m@th\small\normalfont#1}}%
\begin{equation}\nonumber
 e^{\lambda_0 -1} =\sum_{\boldsymbol \sigma,~ \boldsymbol z} \mathcal P({\boldsymbol \sigma, \boldsymbol z}),\  \lambda_1 = 1,~
\Lambda_{i,\mu} = \sqrt{\frac{\beta}{N r^2(1+\rho) }} \frac{1}{M}\sum_{a=1}^M \eta_i^{\mu a}.
\end{equation}
\endgroup
Therefore, this machine captures correlations between the two classes of neurons and, under the supervised learning protocol chosen here, these are recast into empirical averages over examples. 
%%%%%
In particular, the hidden-layer size can be interpreted as a measure of the model flexibility: a larger $K$ allows for a larger number of degrees of freedom and for a finer inference, yet too large a flexibility can imply overfitting phenomena which in our framework are naturally recast as the emergence of a pure spin-glass phase. According to the phase diagram in Fig.~\ref{Fig:Uno}$a$, the maximum flexibility allowed is $K_c = \alpha_c(\rho) N$; this estimate is successfully checked in Fig.~\ref{comparison}$c$.

\begin{figure}
\begin{center}
\includegraphics[scale=0.43]{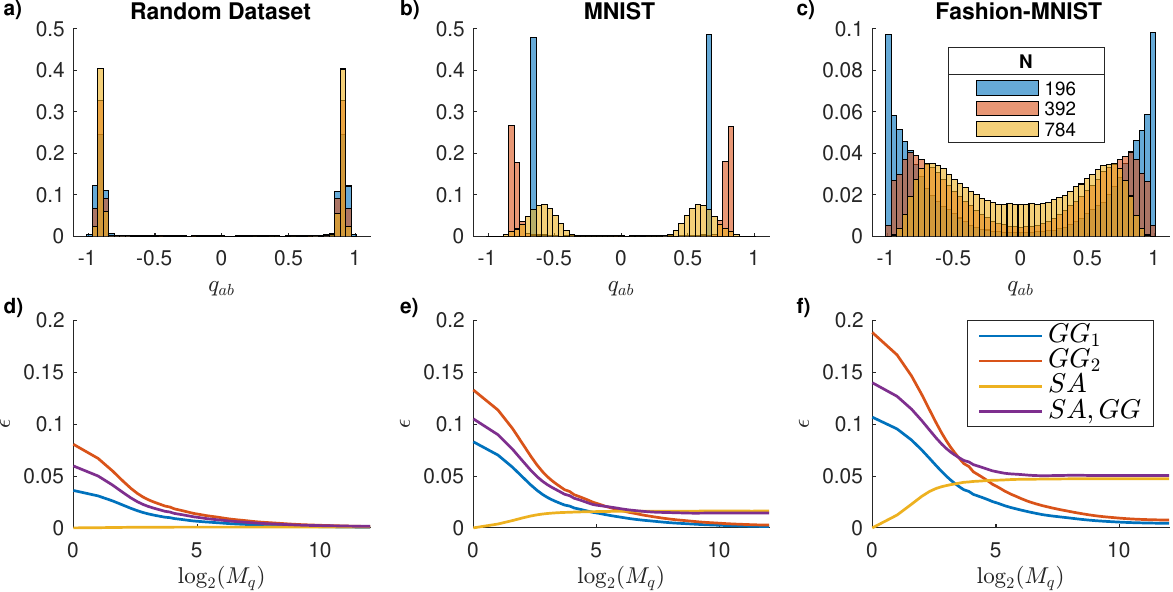}
\caption{Evidence of RSB in structured datasets. Upper Plots: we compare the empirical overlap distribution $\mathcal P(q)$ obtained for the random (panel $a$), the MNIST (panel $b$), and the fashion-MNIST (panel $c$) datasets; three different item sizes are also considered, see the legend. From left to right, we move from a RS scenario where $\mathcal P(q)$ exhibits two peaks that get sharper as the item size increases, to a RSB scenario where $\mathcal P(q)$ is bimodal but with increasing broadness as the item size increases. Lower plots:  we report the violation of the Ghirlanda-Guerra identities ($GG_1$, $GG_2$) and the violation of self-averaging $SA$ as obtained for the random (panel $d$), the MNIST (panel $e$) and the fashion-MNIST (panel $f$) datasets. Again from left to right we move from a replica-symmetric scenario where the self-averaging relations hold and the Ghirlanda-Guerra relations (corresponding to trivial  identities) are fast vanishing, to a picture resembling broken replica symmetry, where  self-averaging does not hold any longer but the Ghirlanda-Guerra relations are still preserved (this time in a not trivial manner). See the SM \cite{SM} for further explanation. 
}\label{fig:MNIST}
\end{center}
\end{figure}

Up to now, we proved that, when dealing with a random, structureless dataset, the HNN with supervised Hebbian rule and the RBM trained under a grandmother-cell scheme are equivalent, and that parameters that emerge naturally in a statistical mechanics framework can be related to standard quantifiers in a machine learning context. More challenging datasets can also be treated as long as the intrinsic structure is properly encoded in the system as we are going to explain. 
Let us denote with $\mathcal S = \{ \boldsymbol \zeta^{\mu a} \}^{\mu=1,...,K}_{a=1,...,M}$ the sample of examples, where the change of notation underlines that now, in general, there is no archetype available hence $\zeta^{\mu a}_i$ can not be obtained by flipping some pixels in the related archetype as in Eq.~\ref{eq:esempi}. 
Moreover, in the structureless case, scrolling though the various examples belonging to the same class, pixels are all homogeneously subject to a flipping probability, while in the structured case some pixels turn out to be more persistent than others. This recalls the difference between ergodic and glassy configurations in spin models. In particular, glassy configurations are characterized by peculiar statistical properties (e.g., lack of self-averaging) which are in turn related to a ultrametric organization. The existence of an analogous organization for dataset items may suggest effective strategies for their processing of a learning machine. In Fig.~\ref{fig:MNIST} we show some evidence in this sense: the distribution of item overlaps -- mirrorring replica overlaps in spin systems -- resembles the Parisi distribution \cite{MPV}, further (as deepened in the SM \cite{SM}) Ghirlanda-Guerra identities \cite{GGide} are numerically shown to hold.  
\newline
In the light of this result we extend the previous grandmother-cell scheme: We pre-treat each sub-sample $\mathcal S_{\mu} = \{ \boldsymbol \zeta^{\mu a} \}_{a=1,...,M}$ to assess its intrinsic structure (e.g., by principal component analysis), whence we determine $K_{\mu}$ disjoint and exhaustive sub-groups $\{ \mathcal S_{\mu}^{\ell}\}_{\ell =1, ..., K_{\mu}}$ and we allocate as many hidden neurons for each class, the overall size of the hidden layer therefore reads as $\hat{K} = \sum_{\mu=1}^K K_{\mu}$. The weight matrix $\boldsymbol W \in \mathbb{R}^{\hat{K} \times N}$ is determined by averaging over instances assigned to each sub-group $\mathcal S_{\mu}^{\ell}$ for $\ell = 1, ..., K$. Classification is finally performed over this hidden layer by an additional softmax layer $\boldsymbol \pi = \textrm{softmax}[\boldsymbol \Gamma \cdot (\boldsymbol W \cdot \boldsymbol  \sigma)^2] \in [0,1]^K$, where $\boldsymbol \Gamma$ can again be determined by simple, algebraic operations over the training set, see Fig. \ref{fig:RSB_RBM} and the SM \cite{SM}. 
\begin{figure}
\begin{center}
\includegraphics[scale=0.15]{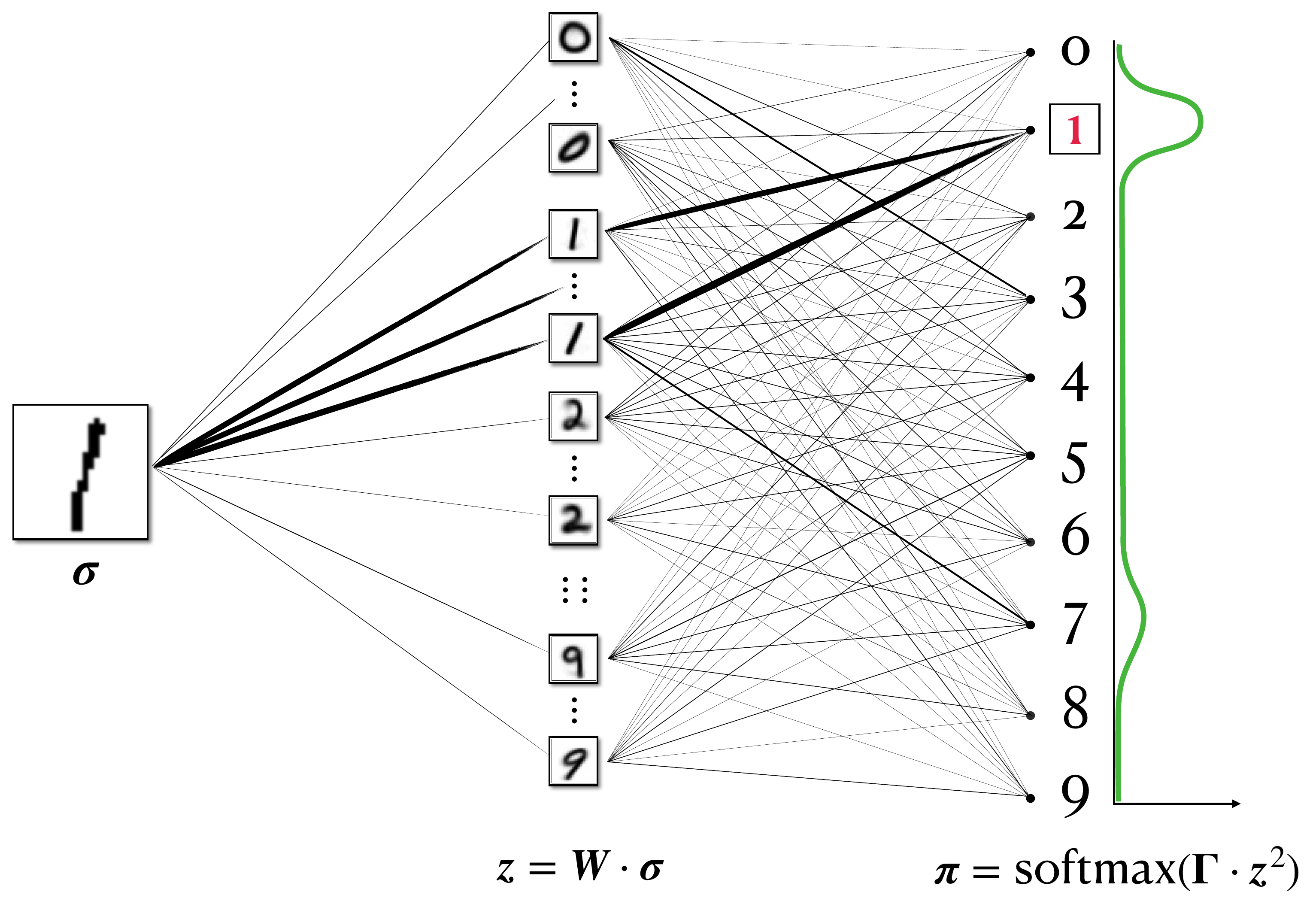}
\caption{Schematic representation of a three-layer RBM for the MNIST dataset based on RSB hierarchy. From left to right: visible layer $\bold{\sigma} \in [-1,+1]^N$ receiving digits to be classified (raw data); hidden layer $\bold{z}\in \mathcal{R}^K$ where each node corresponds to a pseudo archetype as sketched (1-RSB effective representation); softmax layer $\bold{\pi} \in [0,+1]^K$ for classification (RS effective representation).}\label{fig:RSB_RBM}
\end{center}
\end{figure}
The rationale underlying this scheme is that we want to achieve a ``simplified'' representation of data that can be supplied to the classifier: each sub-sample in the structureless case displays a RS representation that allows for an identification between the class and the archetype and therefore to a direct classification; conversely, in the MNIST and in the fashion-MNIST datasets each sub-sample exhibits an intrinsic organization, much as like there were several (pseudo) archetypes for each class in such a way that we need (at least) one extra layer to lift them before classifying them. This procedure can be iterated so to establish a connection between more and more abstract representations in deep learning layers and more and more general representations in RSB steps, hence moving from the leafs (items) toward the common ancestor (archetype).
\newline
The machine obtained in this way has been tested over the two benchmark datasets obtaining an accuracy of about $95 \%$ for MNIST and $84 \%$ for fashion-MNIST, to be compared with, respectively,  $75 \%$ and $63 \%$ obtained for the simple (RS) machine, see Fig.~S6-S7 in  the SM \cite{SM}.

\acknowledgments
The authors are grateful to MOST (Ministry of Science, Technology and Space in Israel) and MAECI  (Ministero degli Affari Esteri e della Cooperazione Internazionale in Italy) for the shared grant ``BULBUL'' (F85F21006230001).
%(IL identification number: ***, IT identification number: Project n. F85F21006230001) for the scientific and tecnological cooperation between Israel and Italy. 
\newline
EA acknowledges financial support from Sapienza University of Rome (RM120172B8066CB0).
\newline
AB is also grateful to the MUR via the (PRIN:2017JFFHS).
%{\em stochastic methods for complex systems}, IT identification number: Project n. 2017JFFHS.
\newline
FA acknowledge partial fundings by PON R\&I (ARS01-00876).

\cleardoublepage
\onecolumn
\gdef\epl@noabstrtrue{1}
\title{Supervised Hebbian learning: Supplementary Material}
\abstract{}
\maketitle
\section{Introduction}
In this Supplementary Material we discuss in details the analytical and computational techniques underlying the results presented in the main text.
\newline
In particular, in Sec.~\ref{sec:trade-off} we inspect the trade-off between the dataset quality and the dataset size from a probabilistic perspective. Then, Secs.~\ref{Signal2Noise} and \ref{sec:SM} are dedicated to the analysis of the Hopfield network (HN) equipped with the Hebbian learning rule (bio-inspired computing): in the former we revise the signal-to-noise approach, while in the latter we detail the statistical mechanics approach. Next, Secs.~\ref{sec:RBM} and \ref{sec:MAE} deal with Restricted Boltzmann machines (RBMs, artificial computing):  in the former we prove that the scaling laws for a successful learning by the HN hold also for the RBM, while in the latter we discuss why and how these two (apparently distant) learning schemes can be related, tackling their equivalence from Jayne's maximum-entropy perspective \cite{Jaynes}. The last Sec.~\ref{sec:structured} extends the treatment to include structured datasets like the MNIST and the fashion-MNIST \cite{MNIST,FMNIST}.
\newline
Before starting, it is useful to recall the basic definitions concerning the random dataset and the HN cost functions implementing the unsupervised and the supervised Hebbian learning.
\begin{defn} \label{def:dataset}
The \emph{archetype} dataset is made of $K$ binary vectors of length $N$, denoted as $\boldsymbol \xi^{\mu} = (\xi_1^{\mu},...,\xi_N^{\mu})$, for $\mu =1,...,K$, whose entries are i.i.d. Rademacher random variables
\begin{equation} \label{eq:Rade}
\mathcal P(\xi_{i}^{\mu}=+1)=\mathcal P(\xi_{i}^{\mu}=-1)=\frac{1}{2},
\end{equation}
for any $i =1, ..., N$ and $\mu=1,...,K$.\\
The dataset available to the machine, denoted as $\mathcal S:= \{ \boldsymbol{\eta}^{\mu a}\}_{\mu=1,...,K}^{a=1,...,M}$, is made of $M$ corrupted \emph{examples} of each archetype: $\boldsymbol{\eta}^{\mu a}$ is the $a$-th example of the $\mu$-th archetype and its entries are defined as 
\begin{equation} \label{eq:bern0}
\eta_{i}^{\mu a}=\chi_{i}^{\mu a}\xi_{i}^{\mu}
\end{equation}
where $\chi_{i}^{\mu a}$ is a binary random variable drawn as
\begin{equation} \label{eq:bern}
\mathcal P(\chi_{i}^{\mu a}=\pm1)=\frac{1\pm r}{2},r\in(0,1]
\end{equation}
for any $i =1, ..., N$, $\mu=1,...,K$, and $a=1,...,M$.
\end{defn}

\begin{defn}\label{def:op_media} The averaging with respect to $\boldsymbol{\xi}^{\mu}$ and $\boldsymbol{\chi}^{\mu a}$ are
\begin{eqnarray}
\mathbb{E}_{\xi^{\mu}}f(\boldsymbol{\xi}^\mu) &:=& \int\prod_{i=1}^{N}\left\{ \frac{d\xi_{i}^{\mu}}{2}\left[\delta(\xi_{i}^{\mu}-1)+\delta(\xi_{i}^{\mu}+1)\right]\right\} f(\boldsymbol{\xi}^\mu), \\
\mathbb{E}_{\chi^{\mu}}f(\boldsymbol{\chi}^\mu) &:=& \int\prod_{i,a=1}^{N,M}\left\{ \frac{d\chi_{i}^{\mu a}}{2}\left[(1+r)\delta(\chi_{i}^{\mu a}-1)+(1-r)\delta(\chi_{i}^{\mu a}+1)\right]\right\} f(\boldsymbol{\chi}^\mu),\\
\mathbb{E}_{\xi} &:=& \prod_{\mu=1}^{K} \mathbb{E}_{\xi^{\mu}},\qquad\mathbb{E}_{\chi} := \prod_{\mu=1}^{K} \mathbb{E}_{\chi^{\mu}},
\end{eqnarray}
	where $f$ is a generic function.
\end{defn}

\begin{defn}
\label{def:H}The Hamiltonian of the \emph{unsupervised} Hopfield model is defined as
\begin{equation} 
\label{model_unsup}
\mathcal{H}^{\textrm{\tiny{HN},unsup}}(\boldsymbol{\sigma}|\mathcal{S}):=-\frac{N}{2M\mathcal{R}}\sum_{\mu=1}^{K}\sum_{a=1}^{M}\left(\frac{1}{N}\sum_{i=1}^{N}\eta_{i}^{\mu a}\sigma_{i}\right)^{2},
\end{equation}
while the Hamiltonian of the \emph{supervised} Hopfield model is defined as
\begin{equation}\label{model}
\mathcal{H}^{\textrm{\tiny{HN},sup}}(\boldsymbol{\sigma}|\mathcal{S}):=-\frac{N}{2\mathcal{R}}\sum_{\mu=1}^{K}\left(\frac{1}{NM}\sum_{i,a=1}^{N,M}\eta_{i}^{\mu a}\sigma_{i}\right)^{2}
\end{equation}
where $\sigma_{i} \in \{-1, +1\}$ for $i=1,\cdots,N$ is a binary neuron (i.e., an Ising spin) and $\mathcal{R}:=r^{2}+\frac{1-r^{2}}{M}$ is a normalization factor.
\end{defn}
\begin{rem}
In the physical model under investigation self-interaction terms are excluded, that is, a neuron $\sigma_i$ interacts with any other neuron $\sigma_k$ with $k \neq i$. This should be accounted for in \eqref{model_unsup} and \eqref{model} by inserting a corrective contribution that neutralizes diagonal terms, however, since this contribution is constant, here it is neglected in order to retain the notation simple.
\end{rem}

%%%%%%%%%%%%%%%%%%%%%%%%%%%%%%

\section{Dataset quality and quantity trade-off}\label{sec:trade-off}

In this section we analyse the supervised and the unsupervised settings, looking for a trade-off between the dataset size $M$ and the dataset quality $r$ by simple probabilitistic arguments based on the 
\begin{thm} (Hoeffding's inequality for bounded random variables) Let $X_1 , . . . , X_N$ be independent r.v. such that $X_i \in [m_i , M_i ]$, with $- \infty < m_i \leq  M_i < +\infty$, $\forall i = 1, . . . , N$. Then, $\forall t \geq 0$
\begin{equation}
\mathcal P \left (  \sum_{i=1}^N (X_i - \mathbb E X_i) \geq t \right) \leq \exp \left( -\frac{2t^2}{\sum_{i=1}^N (M_i - m_i)^2} \right). 
\end{equation}
\end{thm}

Let us for first focus on supervised learning, where we have an a priori knowledge of the map $\mathcal M:\{1,\dots,K \times M \} \to \{1,\dots,K\}$ that assigns to each example
$\boldsymbol{\eta}^{\mu a},\forall a=1,\dots,K \times M$ the related archetype $\boldsymbol{\xi}^{\mu},\mu=1,\dots,K$.
For simplicity, let us assume that each sample $\{\boldsymbol \eta^{\mu a}\}_{a=1,...,M}$ has the
same cardinality $M$ independent of $\mu$.

We recall that the examples are generated according to \eqref{eq:bern0}  where $\chi_{i}^{\mu a} \in \{-1, +1\}$ is a binary random variable, whose value determines whether the $i$-th pixel of the $a$-th example of the $\mu$-th archetype shall be flipped or not, and it is drawn from \eqref{eq:bern}.
Now, the data available can be combined by a majority rule and the probability for a correct prediction on $\xi_i^{\mu}$ is 
\begin{equation}
\mathcal P \left( \sum_{a=1}^M \chi_i^{\mu a} >0 \right ) \geq 1 - \epsilon ~ \Rightarrow \mathcal P \left( \sum_{a=1}^M \chi_i^{\mu a} \leq 0 \right ) \leq  \epsilon,
\end{equation}
for some $\epsilon \in (0,1)$, where in the right-hand side we moved to the complimentary event which allows for a direct application of Hoeffding's inequality:
\begin{eqnarray}
\mathcal P \left( \sum_{a=1}^M \chi_i^{\mu a} \leq 0 \right ) &=& \mathcal P \left( - \sum_{a=1}^M \chi_i^{\mu a} \geq 0 \right ) = \mathcal P \left[  - \left(\sum_{a=1}^M\chi_i^{\mu a} -r \right ) > Mr \right ] \\
 &\leq& \exp \left( - \frac{2M^2 r^2}{4 M} \right) = \exp \left( - \frac{M r^2}{2} \right).
\end{eqnarray}
For a correct generalization we therefore ask that
\begin{equation} \label{eq:scalscal}
\exp \left( - \frac{Mr^2}{ 2}\right) \leq \epsilon \Rightarrow M \geq  \frac{2}{r^2} \log \frac{1}{\epsilon}
\end{equation}
whence we derive the scaling $M r^2 \sim \mathcal O(1)$.

For a higher-accuracy estimate of the success probability, in the case $M \gg 1$, we can exploit the relation
\begin{equation}
\sum_{a=1}^M \chi_i^{\mu a} \sim M r + \lambda \sqrt{ M(1 - r^2 )}, ~\textrm{with} ~  \lambda \sim \mathcal N(0, 1).
\end{equation}
Then, recalling $\rho = \frac{1-r^2}{M r^2}$,
\begin{equation}
\mathcal P \left( M r + \lambda \sqrt{M (1 - r^2)}  \geq  0 \right) =  \frac{1}{\sqrt{2 \pi M (1 - r^2)}} \int_0^{\infty} \exp \left ( -\frac{(\lambda - Mr)^2}{2M(1-r^2)} \right) d \lambda  = \frac{1}{2} \left[ 1 + \textrm{erf} \left( \frac{1}{\sqrt{2 \rho}} \right)  \right]. 
\end{equation}
Therefore, the probability of correctly reconstructing the pixel of the archetype given the set of examples is
\begin{equation}
\mathcal P ( \boldsymbol \xi^{\mu}| \{ \boldsymbol \eta^{\mu a} \}_{a=1}^M ) = \frac{1}{2} \left[ 1 + \textrm{erf} \left(  \frac{1}{\sqrt{2 \rho}} \right )  \right]
\end{equation}
and in order for the sample to play a role, i.e., $\mathcal P ( \boldsymbol \xi^{\mu}| \{ \boldsymbol \eta^{\mu a} \}_{a=1}^M )>1/2$ to guess better than purely random, $\rho$ must be finite, and so we recover $M r^2 \sim \mathcal O(1)$.

As for the unsupervised setting, here the examples are provided without disclosing
the class label and a possible way to clusterize them is by means of the $k$-means algorithm which allows us to group the $K \times M$ examples into $K$ classes and to realize a posteriori a map $\tilde{\mathcal M} :\{1,\dots,K \times M \} \to \{1,\dots, K\}$ under some margin of error.
\newline
For simplicity, also in this case, we introduce a constraint on the cardinality
of the classes that must be class-independent and equal to $M$, so at the end we will
obtain $K$ classes including $M$ examples each. 
\newline
The error margin can be quantified in terms of the fraction mismatched examples.

Once examples have been clusterized, and therefore once
each example has been assigned a label $\mu$, we estimate the archetypes as the centroids of the various classes that are denoted as $\hat{\boldsymbol \xi}^{\mu}$ for $\mu=1,...,K$
and obtained
by averaging over the examples assigned to the $\mu$-th class.
%\begin{equation}
%\bar{\xi}_{i}^{\mu}=\frac{1}{M}\sum_{a=1}^{M}\chi_{i}^{\mu,a}\xi_{i}^{1}.
%\end{equation}
%
Let us assume that the algorithm has placed in the first cluster $\mu=1$, $c \times M$ examples that actually belong to the first cluster, and $(1-c)\times M$
examples that would actually belong to other clusters (without loss of generality these mismatched examples may correspond to the second archetype); the parameter $c\in[0,1]$ regulates the quantity of examples correctly assigned to the class $\mu =1$. 

The centroid pixels are determined as 
\begin{equation}
\hat{\xi}_{i}^{1}=\frac{1}{M}\left(\sum_{a=1}^{cM}\chi_{i}^{1a}\xi_{i}^{1}+\sum_{a=1}^{(1-c)\cdot M}\chi_{i}^{2a}\xi_{i}^{2}\right).
\end{equation}
In order to faithfully reproduce the archetype, we would like to have
that $\xi_{i}^{1}\hat{\xi}_{i}^{1}$ positive
with a high probability and therefore that, being $\epsilon\in (0,1)$, 
\begin{equation}
\sum_{a=1}^{cM}\chi_{i}^{1 a}+\sum_{a=1}^{(1-c)\cdot M}\xi_{i}^{1}\chi_{i}^{2 a}\xi_{i}^{2}\geq 1 - \epsilon
\end{equation}
with high probability.
We notice that $\xi_{i}^{1}\xi_{i}^{2}$ is still a
Rademacher variable and to lighten the notation we replace it with $\lambda_{i}:=\xi_{i}^{1}\xi_{i}^{2}$, where $\lambda_{i}$ is Rademacher's. Now, we use Hoeffding's inequality to estimate the probability $\mathcal P\left(\sum_{a=1}^{cM}\chi_{i}^{1a}+\sum_{a=1}^{(1-c)\cdot M}\chi_{i}^{2 a}\lambda_{i}\geq0\right)$
\begin{equation}
\begin{split}  \mathcal P \left(\sum_{a=1}^{c M}\chi_{i}^{1a}+\sum_{a=1}^{(1-c) M}\chi_{i}^{2a}\lambda_{i}\leq0\right)&= \mathcal P\left(-\sum_{a=1}^{cM}\left(\chi_{i}^{1a}-cr\right)-\sum_{a=1}^{(1-c) M}\lambda_{i}\chi_{i}^{2 a}\geq crM\right)\\
 & \leq\exp\left(-\frac{2c^{2}r^{2}M^{2}}{4cM+4(1-c)M}\right)\\
 &=\exp\left(-\frac{c^{2}r^{2}M}{2}\right)
\end{split}
\end{equation}
then,
\begin{equation}
\mathcal P\left(\sum_{a=1}^{cM}\chi_{i}^{1 a}+\sum_{a=1}^{(1-c)\cdot M}\xi_{i}^{1}\chi_{i}^{2 a}\xi_{i}^{2}>0\right)\geq1-\exp\left(\frac{c^{2}r^{2}M}{2}\right).
\end{equation}
Therefore, we need to ask
\begin{equation}
1-\exp\left(-\frac{c^{2}r^{2}M}{2}\right)\geq1-\epsilon\implies M\geq\frac{2}{c^{2}r^{2}}\log\frac{1}{\epsilon}.
\end{equation}
Now, by setting $c=r$ we obtain a scaling of the type $Mr^{4}\sim \mathcal O(1)$,
on the other hand, if $c=1$ we return to the supervised case in which the rate
of cluster contamination $(1-c)M$ takes on a null value.

\section{Signal-to-noise Analysis}\label{Signal2Noise}

By the signal-to-noise technique we find the conditions for the HN to successfully generalize from examples, namely to retrieve one of the archetypes although it was never presented to them, in the noiseless $\beta \to \infty$ limit. The idea is to assume that the network is in a retrieval state of an archetype, say $\boldsymbol \xi^1$, evaluate the internal fields $h_i(\boldsymbol \xi^1)$ acting on each neuron in such a configuration, and check that the constraints for stability $h_i(\boldsymbol \xi^1) \sigma_i > 0$ are all satisfied. The analysis is led in the high-storage regime, where the load $\alpha:= \lim_{N\to\infty} K/N$ is finite. Clearly, we expect an interplay between $\alpha$, $M$ and $r$ as, by raising $\alpha$ we need a large and good-quality dataset (i.e., large $M$ and $r$) to disentangle examples, and, for a given load, by reducing $r$ we need a larger $M$ to retain enough information in the dataset.
In fact, as proved in the following subsections \ref{ssec:unsup} and \ref{ssec:sup}, respectively,  the optimal tradeoff in the unsupervised regime reads as 
\begin{equation}
\alpha\left(1+\frac{1-r^{4}}{Mr^{4}}\right)+\frac{1-r^{2}}{Mr^{2}}\lesssim1,\label{astronzi}
\end{equation}
while in the supervised regime it reads as 
\begin{equation}
\alpha\left(1+\frac{1-r^{2}}{Mr^{2}}\right)^{2}+\frac{1-r^{2}}{Mr^{2}}\lesssim1.\label{2stability}
\end{equation}
Before proceeding, we anticipate that these scalings are in excellent agreement with Monte Carlo (MC) simulations (see Fig.~\ref{S2N}) and obviously in full accordance with the statistical mechanical predictions reported in the Sec.~\ref{sec:SM}. Moreover, 
in the present treatment, the standard signal-to-noise approach is recast into the one-step Hinton's prescription for fast MC sampling \cite{Hinton1MC}: the reward in this reformulation of the signal-to-noise technique is that it is predictive also for those numerical algorithms that, in turn, we implement in the simulations presented in Sec.~\ref{sec:RBM}, see Eq.~\eqref{eq:1-one}.

\begin{figure}[tb]
\noindent \begin{centering}
\includegraphics[width=0.65\textwidth]{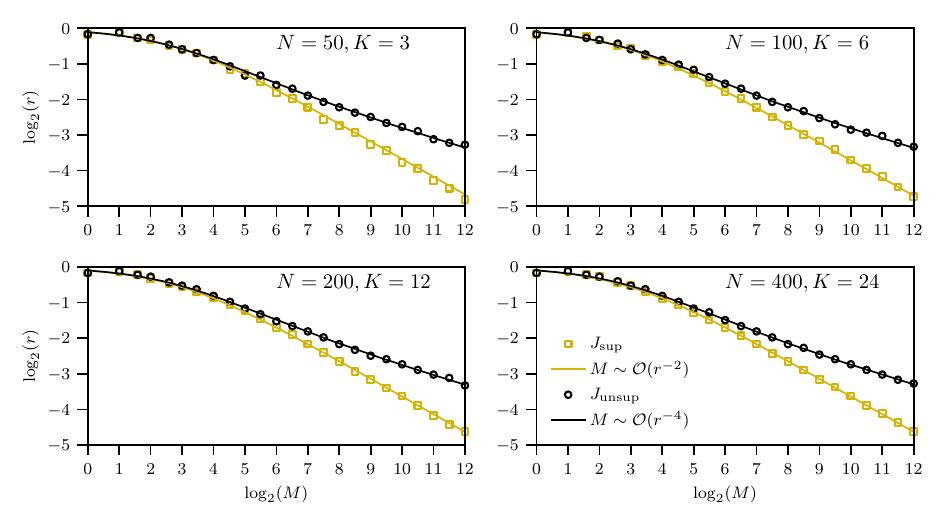}
\par\end{centering}
\caption{Comparison between signal-to-noise predictions (solid lines) and MC simulations (symbols) for the unsupervised (dark color, $\circ$) and the supervised (bright color, $\square$) Hebbian learning in the noiseless $\beta \to \infty$ limit. More precisely,  the theoretical estimates are obtained by eq. \eqref{astronzi} and eq. \eqref{2stability}, while the numerical estimates are obtained by fixing a certain number $M$ of examples and determining the minimal value of the dataset quality $r$ such that the mean overlap between the neural configuration $\boldsymbol \sigma$ and a chosen archetype, say $\boldsymbol \xi^1$, i.e., the magnetization $\langle m \rangle$, is (approximately) unitary, the operation is then repeated for several values of $M$ ranging exponentially from $2$ to $2^{12}$. Note that the theoretical estimates perfectly interpolate the numerical outcomes.}\label{S2N}
\end{figure}

%%%%%%%%%%%%%%%%%%%%%%%%
%
\subsection{Signal vs Noise Analysis: Unsupervised Hebbian Learning} \label{ssec:unsup}
 
The unsupervised case has been previously addressed by signal-to-noise technique in \cite{Giordano} whose results are here recovered and refined, further, the previous investigation is now reformulated in analogy with Hinton's one-step MC recipe \cite{Hinton1MC}.

\begin{prop}
In the unsupervised setting, the condition determining if the HN \eqref{model_unsup} successfully retrieves one of the archetypes is
\begin{equation}
\frac{1}{\sqrt{2\left[\alpha\left(1+\frac{1-r^{4}}{r^{4}M}\right)+\frac{1-r^{2}}{Mr^{2}}\right]}}>\theta\label{eq:stability-1}
\end{equation}
where $\theta \in \mathbb{R}^+$ is a tolerance level, $\textrm{erf}(\theta)$ representing a lower-bound for the overlap between the neural configuration and the retrieved archetype
%the signal-to-noise ratio; in particular, under this condition $m > erf(\theta)$ 
(we omit plots in this case as the results are qualitatively the same as those already reported in \cite{Giordano}).
\end{prop}

\begin{proof}
First we recall the unnormalized log-density of the model:
\begin{equation}
-\beta\mathcal{H}^{\textrm{\tiny{HN}, unsup}}(\boldsymbol{\sigma}|\boldsymbol{\eta})=\frac{\beta N}{2M}\sum_{\mu,a=1}^{K,M}\left(\frac{1}{N}\sum_{i=1}^{N}\eta_{i}^{\mu a}\sigma_{i}\right)^{2}=\frac{\beta}{2NM}\sum_{\mu=1}^{K}\sum_{i,j=1}^{N,N}\sum_{a=1}^{M}\eta_{i}^{\mu a}\eta_{j}^{\mu a}\sigma_{i}\sigma_{j}=\frac{1}{2}\sum_{k=1}^{N}h_{k}(\boldsymbol{\sigma})\sigma_{k}
\end{equation}
where
\begin{equation}
h_{k}(\boldsymbol{\sigma}):=\frac{\beta}{NM}\sum_{\mu=1}^{K}\sum_{i\neq k}^{N}\sum_{a=1}^{M}\eta_{i}^{\mu a}\eta_{k}^{\mu a}\sigma_{i}
\end{equation}
is the internal field acting on $\sigma_{k}$.
\newline
The discrete time MC dynamics is given by the following update
rule
\begin{equation}
\sigma_{k}^{(n+1)}=\sigma_{k}^{(n)}\mathrm{sign}\left\{ \tanh\left[h_{k}^{(n)}(\boldsymbol{\sigma}^{(n)})\sigma_{k}^{(n)}\right]+\zeta_{k}^{(n)}\right\} ,\quad\zeta_{k}^{(n)}\sim \mathcal U(-1,+1).\label{lim-1}
\end{equation}
where $n\in\mathbb{N}$,  $k=1,\dots,N$ and $\mathcal U(-1,1)$ is the uniform probability density in the interval $(-1,+1)$.
\newline
By performing the zero fast-noise limit $\beta\to\infty$, equation
(\ref{lim-1}) becomes
\begin{equation}
\sigma_{k}^{(n+1)}=\sigma_{k}^{(n)}\mathrm{sign}\left(h_{k}^{(n)}(\boldsymbol{\sigma}^{(n)})\sigma_{k}^{(n)}\right).
\end{equation}
Since we want to study the stability of the archetype retrieval, we set as the initial configuration for the MC dynamics
$\sigma_{k}^{(1)}=\xi_{k}^{1}$ for $k=1,...,N$.
The one-step MC approximation for the magnetization is then
\begin{equation}
m_{1}^{(2)}=\frac{1}{N}\sum_{k=1}^N\xi_{k}^{1}\sigma_{k}^{(2)}=\frac{1}{N}\sum_{k=1}^{N}\mathrm{sign}\left(h_{k}^{(1)}(\boldsymbol{\xi}^{1})\xi_{k}^{1}\right) \label{magn-1}
\end{equation}
and, if $N\gg1$, by the central limit theorem, equation (\ref{magn-1})
can be approximated as
\begin{equation}
m_{1}^{(2)}\underset{N\gg1}{\approx}\int\frac{dz}{\sqrt{2\pi}}\exp\left(-\frac{z^{2}}{2}\right)\mathrm{sign}\left(\mu_1+z\sqrt{\mu_{2}-\mu_{1}^{2}}\right)=\mathrm{erf}\left(\frac{\mu_{1}}{\sqrt{2(\mu_{2}-\mu_{1}^{2})}}\right),\label{magnet-1}
\end{equation}
where we posed
\begin{eqnarray}
\mu_{1} &:=& \mathbb{E}_{\xi}\mathbb{E}_{\chi}\left[  h_{k}^{(n)}(\boldsymbol{\xi}^{1})\xi_{k}^{1} \right],\\
\mu_{2} &:=& \mathbb{E}_{\xi}\mathbb{E}_{\chi}\left[  h_{k}^{(n)}(\boldsymbol{\xi}^{1})^{2}\right].
\end{eqnarray}
%with $\langle\cdot\rangle$ the average over the noise affecting the patterns.
%\newline
First, we calculate $\mu_{1}$: 
\begin{equation}
\mu_{1} = \frac{\beta}{NM}\mathbb{E}_{\xi}\mathbb{E}_{\chi}\left[ \sum_{\mu=1}^{K}\sum_{a=1}^{M} \sum_{i\neq k}^{N} \xi_{k}^{\mu}\chi_{k}^{\mu a}\xi_{i}^{\mu}\chi_{i}^{\mu a}\xi_{i}^{1}\xi_{k}^{1}\right] = \frac{\beta}{NM}\mathbb{E}_{\xi}\mathbb{E}_{\chi}\left[ \sum_{\mu>1}^{K}\sum_{a=1}^{M}\sum_{i\neq k}^{N}\xi_{k}^{\mu}\xi_{i}^{\mu}\chi_{k}^{\mu a}\chi_{i}^{ \mu a}+\sum_{a=1}^{M}\sum_{i\neq k}^{N}\chi_{k}^{1 a}\chi_{i}^{1 a}\right].
\end{equation}
Exploiting the fact that, by construction, archetypes are orthogonal in the average, that is $\mathbb{E}_{\xi}[\xi^{\mu}_i \xi^{\mu}_k]=0$ for any $i \neq k$
\begin{equation}
\mathbb{E}_{\xi}\mathbb{E}_{\chi}\left[ \sum_{\mu>1}^{K}\sum_{a=1}^{M}\sum_{i\neq k}^{N}\xi_{k}^{\mu}\xi_{i}^{\mu}\chi_{k}^{\mu a}\chi_{i}^{\mu a}\right] =0,
\end{equation}
and, recalling \eqref{eq:bern},
\begin{equation}
\mathbb{E}_{\xi}\mathbb{E}_{\chi}\left[ \sum_{a=1}^{M}\sum_{i\neq k}^{N}\chi_{k}^{1a}\chi_{i}^{1a}\right] =(N-1)Mr^{2},
\end{equation}
we get
\begin{equation}
\label{eq:two-1-1}
\mu_{1}=\frac{\beta}{N}(N-1)r^{2} \underset{N\gg1}{\approx} \mu_{1}=\beta r^{2}.
\end{equation}
%and, if $N\gg1$, the last  expression can be approximated as  
%\begin{equation}
%\mu_{1}=\beta r^{2}.\label{eq:two-1-1}
%\end{equation}
%
Then, we calculate $\mu_{2}$: 
\begin{eqnarray}
\mu_{2} &=& \mathbb{E}_{\xi}\mathbb{E}_{\chi}\left[ h_{k}(\boldsymbol{\boldsymbol{\xi}^{1}})^{2}\right] =\frac{\beta^{2}}{N^{2}M^{2}}\mathbb{E}_{\xi}\mathbb{E}_{\chi}\left[ \left(\sum_{\mu>1}^{K}\sum_{a=1}^{M}\sum_{i\neq k}^{N}\xi_{k}^{\mu}\xi_{i}^{\mu}\chi_{k}^{\mu a}\chi_{i}^{\mu a}+\sum_{a=1}^{M}\sum_{i\neq k}^{N}\chi_{k}^{1a}\chi_{i}^{1a}\right)^{2}\right] =\\
 &=& \frac{\beta^{2}}{\mathcal{R}^{2}N^{2}M^{4}}\mathbb{E}_{\xi}\mathbb{E}_{\chi}\left[ \left(\sum_{\mu>1}^{K}\sum_{a=1}^{M}\sum_{i\neq k}^{N}\xi_{k}^{\mu}\xi_{i}^{\mu}\chi_{k}^{\mu a}\chi_{i}^{\mu a}\right)^{2}+\left(\sum_{a=1}^{M}\sum_{i\neq k}^{N}\chi_{k}^{1a}\chi_{i}^{1a}\right)^{2}\right]. 
\end{eqnarray}
For the sake of simplicity we pose $A:=\left(\sum_{\mu>1}^{K}\sum_{a=1}^{M}\sum_{i\neq k}^{N}\xi_{k}^{\mu}\xi_{i}^{\mu}\chi_{k}^{\mu a}\chi_{i}^{\mu a}\right)^{2}$
and $B:=\left(\sum_{a=1}^{M}\sum_{i\neq k}^{N}\chi_{k}^{1a}\chi_{i}^{1a}\right)^{2}$such
that 
\begin{equation}
\mu_{2}=\frac{\beta^{2}}{N^{2}M^{2}}\left(\mathbb{E}_{\xi}\mathbb{E}_{\chi}\left[ A\right] +\mathbb{E}_{\xi}\mathbb{E}_{\chi}\left[ B\right] \right).\label{mu2-1}
\end{equation}
After minimal manipulations we can write these terms as
\begin{eqnarray}
\mathbb{E}_{\xi}\mathbb{E}_{\chi}\left[ A\right] &=& \mathbb{E}_{\xi}\mathbb{E}_{\chi}\left[ \sum_{\mu>1}\sum_{a,b}\sum_{i\neq k}\chi_{k}^{\mu a}\chi_{i}^{\mu a}\chi_{k}^{ \mu b}\chi_{i}^{\mu b}\right], \label{eq:m1-1}\\
\mathbb{E}_{\xi}\mathbb{E}_{\chi}\left[ B\right] &=& \mathbb{E}_{\xi}\mathbb{E}_{\chi}\left[ \sum_{a,b}^{M}\sum_{i\neq k}^{N}\chi_{k}^{1a}\chi_{i}^{1a}\chi_{k}^{1b}\chi_{i}^{1b}\right] +\mathbb{E}_{\xi}\mathbb{E}_{\chi}\left[ \sum_{a,b}^{M}\sum_{i,j\neq k}^{N}\lambda_{ij}\chi_{k}^{1 a}\chi_{i}^{1a}\chi_{k}^{1b}\chi_{j}^{1b}\right],\label{eq:m2-1}
\end{eqnarray}
where $\lambda_{ij}:=1-\delta_{ij}$ . By merging the results in (\ref{eq:m1-1}) and (\ref{eq:m2-1}) we
get 
\begin{eqnarray}\label{eq:one-1}
\mathbb{E}_{\xi}\mathbb{E}_{\chi}\left[ A\right] +\mathbb{E}_{\xi}\mathbb{E}_{\chi}\left[ B\right] &=&\mathbb{E}_{\xi}\mathbb{E}_{\chi}\left[ \sum_{\mu=1}\sum_{a,b}\sum_{i\neq k}\chi_{k}^{\mu a}\chi_{i}^{ \mu a}\chi_{k}^{\mu b}\chi_{i}^{ \mu b}\right] +\mathbb{E}_{\xi}\mathbb{E}_{\chi}\left[ \sum_{a,b}^{M}\sum_{i,j\neq k}^{N}\lambda_{ij}\chi_{k}^{1 a}\chi_{i}^{1a}\chi_{k}^{1b}\chi_{j}^{1b}\right] =\\
 &=& K(N-1)\left[M+r^{4}M(M-1)\right]+(N-1)(N-2)\left[r^2 M+r^4M(M-1)\right]\\
 \label{finap-1}
&\underset{N \gg 1}{\approx}& KNr^{4}M^{2}\left(1+\frac{1-r^{4}}{r^{4}M}\right)+N^{2} r^4 M^{2}\left(1+\frac{1-r^{2}}{r^2 M}\right)
\end{eqnarray}
%Keeping $N\gg1$, the last expression can be approximated as  
%\begin{equation}
%\left\langle A\right\rangle +\left\langle B\right\rangle \approx KNr^{4}M^{2}\left(1+\frac{1-r^{4}}{r^{4}M}\right)+N^{2}M^{2}r^{4}\left(1+\frac{1-r^{2}}{Mr^{2}}\right)\label{finap-1}
%\end{equation}
and, by direct substitution of \eqref{finap-1} into \eqref{mu2-1}, we get 
\begin{equation}
\mu_{2}\underset{N \gg 1}{\approx}\beta^{2}r^{4}\left[\alpha\left(1+\frac{1-r^{4}}{r^{4}M}\right)+1+\frac{1-r^{2}}{r^2 M}\right]\label{mudue-1}.
\end{equation}
Finally, by plugging equations \eqref{eq:two-1-1} and \eqref{mudue-1} into the expression of the magnetization \eqref{magnet-1},
we get
\begin{equation}
m_{1}^{(2)}\underset{N \gg 1}{\approx}\mathrm{erf}\left(\frac{1}{\sqrt{2\left[\alpha\left(1+\frac{1-r^{4}}{r^{4}M}\right)+\frac{1-r^{2}}{Mr^{2}}\right]}}\right)
\end{equation}
and, by requiring that this one-step MC magnetization is larger than $\textrm{erf}(\theta)$ we recover Eq.~\eqref{eq:stability-1}. 
\end{proof}

%%%%%%%%%%%%%%%%%%%%%%%%%%%%%%%
\subsection{Signal vs Noise Analysis: Supervised Hebbian learning} \label{ssec:sup}

We now repeat analogous passages for the supervised setting that lead to
\begin{prop}
In the supervised setting, the condition determining if the HN \eqref{model} successfully retrieves one of the archetypes is
\begin{equation}
\frac{1}{\sqrt{2\alpha\left(1+\rho\right)^{2}+2\rho}}>\theta\label{eq:stability_sup}
\end{equation}
where $\theta \in \mathbb{R}^+$ is a tolerance level, $\textrm{erf}(\theta)$ representing a lower-bound for the overlap between the neural configuration and the retrieved archetype (for instance, the case $\theta=1/\sqrt{2}$ is shown in Fig.~\ref{figSignal-vs-Noise}).

\begin{figure}[tb]
\begin{centering}
\includegraphics[width=0.9\textwidth]{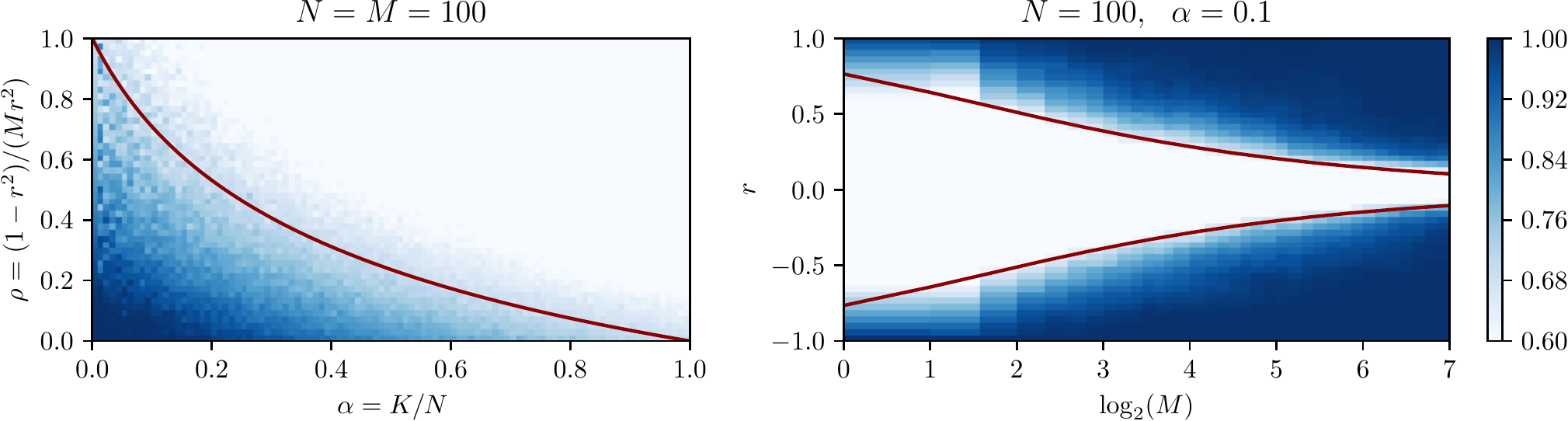}
\par\end{centering}
\caption{Signal-to-Noise analysis\label{figSignal-vs-Noise}. Left panel: 
Contour plot for the one-step-MC magnetization $m^{(2)}$ defined in \eqref{magnet} and evaluated by MC simulations, run at different values of $\alpha$ and $\rho$; the solid line corresponds to the curve $\alpha\left(1+\rho\right)^{2}+\rho=1$ obtained by eq.~ \eqref{eq:stability_sup} by setting as tolerance level $\theta=\frac{1}{\sqrt{2}}$.
Right panel: Contour plot for the one-step-MC magnetization $m^{(2)}$ defined in \eqref{magnet} and evaluated by MC simulations, run at different values of $r$ and $M$ (notice the log-scale); the solid lines are again obtained by eq.~\eqref{eq:stability_sup} but now plotted as a function of $(r,\log_2M)$ and we see that, as the quality $|r|$
reaches zero, more and more examples are needed to ensure the retrieval of the archetype. The colormap on the right is shared by the two panels. }
\end{figure}
\end{prop}

\begin{proof}
First we recall the unnormalized log-density of the model:
\begin{equation}
-\beta\mathcal{H}^{\textrm{\tiny{HN}, sup}}(\boldsymbol{\sigma}|\boldsymbol{\eta})=\frac{\beta N}{2\mathcal{R}}\sum_{\mu=1}^{K}\left(\frac{1}{NM}\sum_{i,a=1}^{N,M}\eta_{i}^{\mu a}\sigma_{i}\right)^{2}=\frac{\beta}{2\mathcal{R}NM^{2}}\sum_{\mu=1}^{K}\sum_{i,j=1}^{N}\sum_{a,b=1}^{M}\eta_{i}^{\mu a}\eta_{j}^{\mu b}\sigma_{i}\sigma_{j}=\frac{1}{2}\sum_{k=1}^{N}h_{k}(\boldsymbol{\sigma})\sigma_{k}
\end{equation}
where
\begin{equation}
h_{k}(\boldsymbol{\sigma}):=\frac{\beta}{\mathcal{R}NM^{2}}\sum_{\mu=1}^{K}\sum_{i\neq k}^{N}\sum_{a,b=1}^{M}\eta_{i}^{\mu a}\eta_{k}^{\mu b}\sigma_{i}
\end{equation} 
is the internal field acting on $\sigma_{k}$. 
\newline
The discrete-time MC dynamics is given by the following update rule
\begin{equation}
\sigma_{k}^{(n+1)}=\sigma_{k}^{(n)}\mathrm{sign}\left\{ \tanh\left[h_{k}^{(n)}(\boldsymbol{\sigma}^{(n)})\sigma_{k}^{(n)}\right]+\zeta_{k}^{(n)}\right\} ,\quad\zeta_{k}^{(n)}\sim \mathcal U(-1,+1),\label{lim}
\end{equation}
where $n\in\mathbb{N}$ and $k=1,\dots,N$. By performing the low
temperature limit $\beta\to\infty$ equation (\ref{lim}) becomes
\begin{equation}
\sigma_{k}^{(n+1)}=\sigma_{k}^{(n)}\mathrm{sign}\left(h_{k}^{(n)}(\boldsymbol{\sigma}^{(n)})\sigma_{k}^{(n)}\right).
\end{equation}
We want to study the stability of the archetype retrieval thus
we set as an initial configuration for the MC dynamics $\sigma_{k}^{(1)}=\xi_{k}^{1}$; 
the one-step MC approximation for the magnetization is then
\begin{equation}
m_{1}^{(2)}=\frac{1}{N}\sum_{k=1}^N\xi_{k}^{1}\sigma_{k}^{(2)}=\frac{1}{N}\sum_{k=1}^{N}\mathrm{sign}\left(h_{k}^{(1)}(\boldsymbol{\xi}^{1})\xi_{k}^{1}\right).\label{magn}
\end{equation}
If $N\gg1$, by the central limit theorem, equation (\ref{magn})
can be approximated as follows
\begin{equation}
m_{1}^{(2)}\underset{N\gg1}{\approx}\int\frac{dz}{\sqrt{2\pi}}\exp\left(-\frac{z^{2}}{2}\right)\mathrm{sign}\left(\mu_1+z\sqrt{\mu_{2}-\mu_{1}^{2}}\right)=\mathrm{erf}\left(\frac{\mu_{1}}{\sqrt{2(\mu_{2}-\mu_{1}^{2})}}\right)\label{magnet}
\end{equation}
where we posed
\begin{eqnarray}
\mu_{1} &:=& \mathbb{E}_{\xi}\mathbb{E}_{\chi}\left[ h_{k}^{(n)}(\boldsymbol{\xi}^{1})\xi_{k}^{1}\right], \\
\mu_{2} &:=& \mathbb{E}_{\xi}\mathbb{E}_{\chi}\left[ h_{k}^{(n)}(\boldsymbol{\xi}^{1})^{2}\right].
\end{eqnarray}
First, we calculate $\mu_{1}$:
\begin{equation}
\mu_{1} =\frac{\beta}{\mathcal{R}NM^{2}}\mathbb{E}_{\xi}\mathbb{E}_{\chi}\left[ \sum_{\mu=1}^{K}\sum_{i\neq k}^{N}\sum_{a,b=1}^{M}\xi_{k}^{\mu}\chi_{k}^{\mu a}\xi_{i}^{\mu}\chi_{i}^{\mu b}\xi_{i}^{1}\xi_{k}^{1}\right] =\frac{\beta}{\mathcal{R}NM^{2}}\mathbb{E}_{\xi}\mathbb{E}_{\chi}\left[ \sum_{\mu>1}^{K}\sum_{a,b=1}^{M}\sum_{i\neq k}^{N}\xi_{k}^{\mu}\xi_{i}^{\mu}\chi_{k}^{ \mu a}\chi_{i}^{\mu b}+\sum_{a,b}^{M}\sum_{i\neq k}^{N}\chi_{k}^{1a}\chi_{i}^{1b}\right].
\end{equation}
Exploiting the fact that, by construction, archetypes are orthogonal in the average, that is $\mathbb{E}_{\xi}[\xi^{\mu}_i \xi^{\mu}_k]=0$ for any $i \neq k$,
\begin{equation}
\mathbb{E}_{\xi}\mathbb{E}_{\chi}\left[ \sum_{\mu>1}^{K}\sum_{a,b=1}^{M}\sum_{i\neq k}^{N}\xi_{k}^{\mu}\xi_{i}^{\mu}\chi_{k}^{\mu a}\chi_{i}^{ \mu b}\right] =0,
\end{equation}
and recalling \eqref{eq:bern}
\begin{equation}
\mathbb{E}_{\xi}\mathbb{E}_{\chi}\left[ \sum_{a,b}\sum_{i\neq k}\chi_{k}^{1a}\chi_{i}^{1b}\right] =(N-1)M^{2}r^{2},
\end{equation}
we get
\begin{equation}\label{eq:two-1}
\mu_{1}=\frac{\beta}{(1+\rho)N}(N-1)  \underset{N\gg1}{\approx}\mu_{1}=\frac{\beta}{(1+\rho)}.
\end{equation}
%and, keeping $N\gg1$, the last expression can be approximated as
%\begin{equation}
%\mu_{1}=\frac{\beta}{(1+\rho)}.
%\end{equation}
Then, we calculate $\mu_{2}$:
\begin{eqnarray}
  \mu_{2}&=&\mathbb{E}_{\xi}\mathbb{E}_{\chi}\left[ h_{k}(\boldsymbol{\boldsymbol{\xi}}^{1})^{2}\right] =\frac{\beta^{2}}{\mathcal{R}^{2}N^{2}M^{4}}\mathbb{E}_{\xi}\mathbb{E}_{\chi}\left[ \left(\sum_{\mu>1}^{K}\sum_{a,b=1}^{M}\sum_{i\neq k}^{N}\xi_{k}^{\mu}\xi_{i}^{\mu}\chi_{k}^{\mu a}\chi_{i}^{ \mu b}+\sum_{a,b}^{M}\sum_{i\neq k}^{N}\chi_{k}^{1a}\chi_{i}^{1b}\right)^{2}\right] =\\
  &=&\frac{\beta^{2}}{\mathcal{R}^{2}N^{2}M^{4}}\mathbb{E}_{\xi}\mathbb{E}_{\chi}\left[ \left(\sum_{\mu>1}^{K}\sum_{a,b=1}^{M}\sum_{i\neq k}^{N}\xi_{k}^{\mu}\xi_{i}^{\mu}\chi_{k}^{\mu a}\chi_{i}^{ \mu b}\right)^{2}+\left(\sum_{a,b}^{M}\sum_{i\neq k}^{N}\chi_{k}^{1 a}\chi_{i}^{1b}\right)^{2}\right]. 
\end{eqnarray}
For the sake of simplicity we pose $A:=\left(\sum_{\mu>1}^{K}\sum_{a,b=1}^{M}\sum_{i\neq k}^{N}\xi_{k}^{\mu}\xi_{i}^{\mu}\chi_{k}^{\mu a}\chi_{i}^{ \mu b}\right)^{2}$
and $B:=\left(\sum_{a,b}^{M}\sum_{i\neq k}^{N}\chi_{k}^{ 1a}\chi_{i}^{1b}\right)^{2}$such
that
\begin{equation}
\mu_{2}=\frac{\beta^{2}}{\mathcal{R}^{2}N^{2}M^{4}}\left(\mathbb{E}_{\xi}\mathbb{E}_{\chi}\left[ A\right] +\mathbb{E}_{\xi}\mathbb{E}_{\chi}\left[ B\right] \right).\label{mu2}
\end{equation}
After minimal manipulations we can write these terms as
\begin{eqnarray}
\mathbb{E}_{\xi}\mathbb{E}_{\chi}\left[ A\right] &=& \mathbb{E}_{\xi}\mathbb{E}_{\chi}\left[ \sum_{\mu>1}\sum_{a,b,c,d}\sum_{i\neq k}\chi_{k}^{\mu a}\chi_{i}^{\mu b}\chi_{k}^{\mu c}\chi_{i}^{ \mu d}\right], \label{eq:m1}\\
\mathbb{E}_{\xi}\mathbb{E}_{\chi}\left[ B\right] &=& \mathbb{E}_{\xi}\mathbb{E}_{\chi}\left[\sum_{a,b,c,d}^{M}\sum_{i\neq k}^{N}\chi_{k}^{1 a}\chi_{i}^{1b}\chi_{k}^{1c}\chi_{i}^{1d}\right] +\mathbb{E}_{\xi}\mathbb{E}_{\chi}\left[ \sum_{a,b,c,d}^{M}\sum_{i,j\neq k}^{N}\lambda_{ij}\chi_{k}^{1a}\chi_{i}^{1b}\chi_{k}^{1c}\chi_{j}^{1d}\right], \label{eq:m2}
\end{eqnarray}
where $\lambda_{ij}:=1-\delta_{ij}$ . By merging the results in (\ref{eq:m1})
and (\ref{eq:m2}) we get
\begin{eqnarray}
\label{eq:one}
\mathbb{E}_{\xi}\mathbb{E}_{\chi}\left[ A\right] +\mathbb{E}_{\xi}\mathbb{E}_{\chi}\left[ B\right] &=&\mathbb{E}_{\xi}\mathbb{E}_{\chi}\left[ \sum_{\mu=1}\sum_{a,b,c,d,}\sum_{i\neq k}\chi_{k}^{\mu a}\chi_{i}^{ \mu b}\chi_{k}^{ \mu c}\chi_{i}^{\mu d}\right] +\mathbb{E}_{\xi}\mathbb{E}_{\chi}\left[ \sum_{a,b,c,d}^{M}\sum_{i,j\neq k}^{N}\lambda_{ij}\chi_{k}^{1a}\chi_{i}^{1b}\chi_{k}^{1c}\chi_{j}^{1d}\right] =\\ 
\nonumber
&=& K(N-1)M^{4}r^{4}\left(1+\rho\right)^{2}+(N-1)(N-2)M^{4}r^{4}\left(1+\rho\right)\\
\label{finap}
&\underset{N \gg 1}{\approx}& KNM^{4}r^{4}\left(1+\rho\right)^{2}+N^{2}M^{4}r^{4}\left(1+\rho\right).
\end{eqnarray}
By direct substitution of \eqref{finap} into \eqref{mu2} we obtain
\begin{equation}
\mu_{2}\underset{N \gg 1}{\approx}\beta^{2}\left(\alpha+\frac{1}{1+\rho}\right).\label{mudue}
\end{equation}
Finally, by plugging equations \eqref{eq:two-1} and \eqref{mudue} into the expression of the magnetization \eqref{magnet}
we get
\begin{equation}
m_{1}^{(2)}\underset{N \gg 1}{\approx}\mathrm{erf}\left(\frac{1}{\sqrt{2\alpha\left(1+\rho\right)^{2}+2\rho}}\right),
\end{equation}
and, by requiring that this one-step MC magnetization is larger than $\textrm{erf}(\theta)$ we recover Eq.~\eqref{eq:stability_sup}. 
\end{proof}

\begin{rem}
Recalling Eq.~\eqref{magnet}, one can see that setting $\theta = \frac{1}{\sqrt{2}}$ in \eqref{eq:stability_sup} corresponds to the condition $\mathbb{E}_{\xi}\mathbb{E}_{\chi}\left[ h_{k}^{(n)}(\boldsymbol{\xi}^{1})\xi_{k}^{1}\right]>\sqrt{\textrm{Var}(h_{k}^{(n)}(\boldsymbol{\xi}^{1})\xi_{k}^{1})}$, that is the standard condition used in signal-to-noise analysis; this is also depicted in Fig.~\ref{figSignal-vs-Noise}
\end{rem}

%%%%%%%%%%%%%%%%%%%%

\section{Statistical mechanics of Supervised Hebbian Learning} \label{sec:SM}
Here we follow a standard route in glassy statistical-mechanics \cite{Parisi}: first, we outline the control parameters of the system (i.e., $\alpha, \beta, \rho$), introduce the partition function $\mathcal{Z}^{\textrm{\tiny{HN}}}_{\beta}(\mathcal{S})$ and quenched free-energy $f_{\alpha,\beta,\rho}$ related to the cost function (\ref{model}), as well as the order parameters necessary to describe the macroscopic behaviour of the system. Then, we achieve an explicit expression of the quenched free-energy in terms of these parameters: this will be accomplished in the replica symmetric (RS) regime, namely  under the assumption that the order parameters do not fluctuate in the thermodynamic limit. Finally, we extremize the free energy w.r.t. the order parameters to force the thermodynamic requirements of minimum energy and maximum entropy. This extremization returns a set of self-consistent equations for the evolution of the order parameters in the space of the control parameters whose inspection allows us to draw a phase diagram for the network and thus to obtain an exhaustive characterization of its emergent information-processing skills.
\newline
We will address this investigation in a mathematical rigorous way, by relying on the generalized Guerra's interpolation scheme \cite{Guerra1,Guerra2,AABF-NN2020}, rather than pseudo-heuristic tools like replica-trick: the basic idea of this approach is to introduce a suitable functional that interpolates between the original model and a series of simpler (and solvable) models and then to solve for these simpler models and propagate the solution back to the other extremum of the interpolation, namely the model under study.
\newline
The statistical mechanical treatment of the unsupervised HN \eqref{model_unsup} has been already addressed in \cite{Giordano} hence here we will just focus on the supervised case \eqref{model}  and we can drop the superscript ``sup'' without ambiguity. 
\begin{defn} The control parameters that tune the performance of the HN implementing a supervised Hebbian learning  \eqref{model} are
\begin{itemize}
\item the network \emph{storage} $\alpha := \lim_{N \to \infty}\frac{K}{N}$, that is, the ratio between the number $K$ of archetypes that we aim to retrieve and the number $N$ of neurons employed for this task,
\item the inverse \emph{noise} $\beta:=T^{-1}$ tuning the stochasticity in the network dynamics,
\item the dataset \emph{entropy}\footnote{Strictly speaking, $\rho$ is not an entropy, yet here we allow ourselves for this slight abuse of language because, as discussed in the main text, the conditional entropy $H( \xi_i^{\mu} | \boldsymbol \eta_i^{\mu})$ is a monotonically increasing function of $\rho$.} $\rho:=\frac{1-r^{2}}{Mr^{2}}$, that is the amount of information needed to describe the archetype set $\{ \boldsymbol \xi^{\mu} \}_{\mu=1,...,K}$ given the sample $\mathcal S$.     
\end{itemize}
\end{defn}
\begin{defn}
The partition function of the supervised HN \eqref{model} is defined as
\begin{equation}
\mathcal{Z}^{\textrm{\tiny{HN}}}_{\beta}(\mathcal{S}):=\sum_{\{\sigma\}}^{2^N}\exp\left[-\beta\mathcal{H}^{\textrm{\tiny{HN}}}(\boldsymbol{\sigma}| \mathcal{S})\right]=\sum_{\{\sigma\}}^{2^N}\exp\left[\frac{\beta N}{2\mathcal{R}}\sum_{\mu=1}^{K}\left(\frac{1}{NM}\sum_{i,a=1}^{N,M}\eta_{i}^{\mu a}\sigma_{i}\right)^{2}\right] \label{eq:Z_raw}
\end{equation}
and the related quenched free-energy is defined as
\begin{equation}
\label{eq:f_raw}
f_{\alpha, \beta, \rho}:= - \frac{1}{\beta N} \mathbb{E}_{\xi}\mathbb{E}_{\chi}\left[ \log \mathcal{Z}^{\textrm{\tiny{HN}}}_{\beta}(\mathcal{S})\right]
\end{equation}
where the average $\mathbb{E}_{\xi}\mathbb{E}_{\chi}$ is specified in Definition \eqref{def:op_media}. 
\end{defn}
\begin{defn} \label{def:q}
\label{def:The-order-parameters}The order parameters required to describe the performance of the HN implementing a supervised Hebbian learning  \eqref{model} are
%\begin{equation}
%\begin{aligned}
\begin{itemize}
\item the \emph{archetype magnetization} $m :=\frac{1}{N}\sum_{i=1}^{N}\xi_{i}^{1}\sigma_{i}$,
\item the \emph{mean example magnetization} $n  :=\frac{r}{\mathcal{R}}\frac{1}{NM}\sum_{i,a=1}^{N,M}\eta_{i}^{1a}\sigma_{i}$, 
\item the two-replica \emph{overlap} for the Ising neurons $q_{ab} :=\frac{1}{N}\sum_{i=1}^{N}\sigma_{i}^{(a)}\sigma_{i}^{(b)},$ 
\item the two-replica \emph{overlap} for the Gaussian neurons (vide infra) $p_{ab} :=\frac{1}{K-1}\sum_{\mu=2}^{K}z_{\mu}^{(a)}z_{\mu}^{(b)}$.
\end{itemize}
%\end{aligned}
%\end{equation}
%Regarding the latter, the two-replica overlap for the Gaussian neurons $p_{ab}$, by now we note that the real-valued neurons are intrinsically related to the Hebbian learning as they naturally appear in the following integral representation of the Hopfield partition function (that, in turn, already suggests deep links with Restricted Boltzmann Machine-like architectures as a two layer network naturally emerges via this representation of $\mathcal{Z}(\beta|\boldsymbol{\eta})$: we will deepen this duality in the last Section).
Notice that $m$ and $n$ are referred to the first archetype without loss of generality.\\
Further, being $x$ the generic order parameter, we denote with $\langle x \rangle$ its expectation under the Boltzmann-Gibbs measure $\mathcal P_{\beta}(\boldsymbol \sigma | \mathcal S)= [\mathcal{Z}^{\textrm{\tiny{HN}}}_{\beta}(\mathcal{S})]^{-1} \exp\left[-\beta\mathcal{H}^{\textrm{\tiny{HN}}}(\boldsymbol{\sigma}| \mathcal{S})\right]$.
\end{defn}
\begin{prop} \label{prop:int_rep}
The integral representation of the partition function of the HN implementing supervised Hebbian learning \eqref{model} reads as
\begin{equation}\label{eq:Z_lin_prop}
\mathcal{Z}^{\textrm{\tiny{HN}}}_{\beta}(\boldsymbol{\lambda},\boldsymbol{\eta}^{1})=\sum_{\{\sigma\}}\int\prod_{\mu=2}^{K}\left(\frac{dz_{\mu}}{\sqrt{2\pi}}\right)\exp\left[-\sum_{\mu=2}^{K}\frac{z_{\mu}^{2}}{2}+\sqrt{\frac{\beta}{N}}\sum_{\mu=2}^{K}\sum_{i=1}^{N}\lambda_{i}^{\mu}z_{\mu}\sigma_{i}+\frac{\beta N}{2}{n}^{2}(1+\rho)\right],
\end{equation}
where $\boldsymbol \lambda = \{ \lambda_{i}^{\mu} \}_{i=1,...,N}^{\mu=1,...,K}$ represent i.i.d standard-Gaussian random fields.
\end{prop}
\begin{proof}
First, we split the Hamiltonian $\mathcal{H}^{\textrm{\tiny{HN}}}(\boldsymbol \sigma | \mathcal{S})$ into two contributions: one containing the terms related to examples with label $\mu=1$ (playing as a ``signal'') and the other containing the terms related to examples with labels $\mu \neq 1$ (playing as slow noise). Then, exploiting the relation $\int dz\exp(-z^{2}/2+Az)=\sqrt{2\pi}\exp(A^{2}/2)$, we linearise the noise terms inside \eqref{eq:Z_raw} obtaining
\begin{equation}
\mathcal{Z}^{\textrm{\tiny{HN}}}_{\beta}(\mathcal{S})=\sum_{\{\sigma\}}\int\prod_{\mu=2}^{K}\left(\frac{dz_{\mu}}{\sqrt{2\pi}}\right)\exp\left[-\sum_{\mu=2}^{K}\frac{z_{\mu}^{2}}{2}+\sqrt{\frac{\beta}{N\mathcal{R}}}\frac{1}{M}\sum_{\mu=2}^{K}\sum_{i,a=1}^{N,M}\eta_{i}^{\mu a}\sigma_{i}z_{\mu}+\frac{\beta N}{2\mathcal{R}}\left(\frac{1}{NM}\sum_{i,a=1}^{N,M}\eta_{i}^{1a}\sigma_{i}\right)^{2}\right],
\end{equation}
and, recalling Definition \eqref{def:The-order-parameters}, we rewrite the signal contribution, namely
\begin{equation}
\mathcal{Z}^{\textrm{\tiny{HN}}}_{\beta}(\mathcal{S})=\sum_{\{\sigma\}}\int\prod_{\mu=2}^{K}\left(\frac{dz_{\mu}}{\sqrt{2\pi}}\right)\exp\left[-\sum_{\mu=2}^{K}\frac{z_{\mu}^{2}}{2}+\sqrt{\frac{\beta}{N\mathcal{R}}}\frac{1}{M}\sum_{\mu=2}^{K}\sum_{i,a=1}^{N,M}\eta_{i}^{\mu a}\sigma_{i}z_{\mu}+\frac{\beta N}{2}(1+\rho){n}^{2}\right].\label{eq:Z_lin}
\end{equation}
We now handle the noise contribution trying to make it analytically more treatable: from \eqref{eq:Z_lin} we extract the noise contribution, namely 
\begin{equation}
\frac{1}{M}\sqrt{\frac{\beta}{N\mathcal{R}}}\sum_{\mu=2}^{K}\sum_{i,a=1}^{N,M}\eta_{i}^{\mu a}z_{\mu}\sigma_{i}=\sqrt{\frac{\beta}{N\mathcal{R}}}\sum_{\mu=2}^{K}\sum_{i=1}^{N}\left(\frac{1}{M}\sum_{a=1}^{M}\eta_{i}^{\mu a}\right)z_{\mu}\sigma_{i}
\end{equation}
and we see that the random field acting on the pairs $\sigma_{i}z_{\mu}$
is $\frac{1}{M}\sum_{a=1}^{M}\eta_{i}^{\mu a}$, hence, in the large dataset scenario $M \gg 1$, this random field
can be replaced by a Gaussian-distributed random field with the same mean and variance that turn out to be, respectively,
\begin{equation}
\mathbb{E}_{\chi}\mathbb{E}_{\xi}\frac{1}{M}\sum_{a=1}^{M}\eta_{i}^{\mu a}=\frac{1}{M}\sum_{a=1}^{M}\mathbb{E}_{\xi}\xi_{i}^{\mu}\,\mathbb{E}_{\chi}\chi_{i}^{\mu a}=0
\end{equation}
\begin{equation}
\begin{split}\mathbb{E}_{\chi}\mathbb{E}_{\xi}\left(\frac{1}{M}\sum_{a=1}^{M}\eta_{i}^{\mu a}\right)^{2}= & \mathbb{E}_{\chi}\mathbb{E}_{\xi}\left(\frac{1}{M}\sum_{a=1}^{M}\chi_{i}^{\mu a}\xi_{i}^{\mu}\right)^{2}=\\
= & \mathbb{E}_{\chi}\left(\frac{1}{M}\sum_{a=1}^{M}\chi_{i}^{\mu a}\right)^{2}=\\
= & \frac{1}{M^{2}}\mathbb{E}_{\chi}\left[\sum_{a=1}^{M}\left(\chi_{i}^{\mu a}\right)^{2}+\sum_{a\neq b=1}^{M,M}\chi_{i}^{\mu a}\chi_{i}^{\mu b}\right]=\\
= & \frac{1}{M^{2}}\left[M+M(M-1)r^{2}\right]=r^{2}+\frac{1-r^{2}}{M}=\mathcal{R}.
\end{split}
\end{equation}
%
%where $\mathbb{E}_{\xi}$ and $\mathbb{E}_{\chi}$ stand for averages
%over the distribution of the random tensors $\xi$ and $\chi$ whose
%probability density was specified in Definition \eqref{def:H}.
%
Therefore, the noise term in the free energy can be replaced by
\begin{equation}
\sqrt{\frac{\beta}{N}}\sum_{\mu=2}^{K}\sum_{i=1}^{N}\lambda_{i}^{\mu}z_{\mu}\sigma_{i}~~\mathrm{with}~~\lambda_{i}^{\mu}\sim\mathcal{N}(0,1)
\end{equation}
recovering eq.~\eqref{eq:Z_lin_prop}.
\end{proof}
\begin{rem}
The partition function $\mathcal{Z}^{\textrm{\tiny{HN}}}_{\beta}(\mathcal{S})$ given in \eqref{eq:Z_lin} can be seen as the partition function of a two-species spin-glass model, where one species is made of binary spins $\boldsymbol \sigma \in \{-1, +1\}^N$ and the other species is made of Gaussian spins $\boldsymbol z \sim \mathcal N(0,1)^{K-1}$; spins of different nature interact pairwisely by a random coupling given by $\sqrt{\frac{\beta}{ N } }\lambda_i^{\mu}$, binary spins interact  pairwisely by a Hebbian-like coupling $\frac{\beta}{2 \mathcal R} \bar{\eta}_i^1 \bar{\eta}_j^1$; interactions between real spins are absent.
\end{rem}
\begin{rem}
The previous proposition is consistent with the universality of the quenched noise proved for the Sherrington-Kirkpatrick spin-glasses \cite{Carmona} and extended to bipartite spin-glasses and neural networks \cite{Giuseppe,Barattolo}. Remarkably, it allows us to substitute the digital i.i.d. entries that contribute to the slow noise with Gaussian random variables and this, in turn, allows us to apply the Wick-Isserlis theorem to analytically treat the expression of the quenched free-energy.
\end{rem}
As a result of the previous proposition we can recast the quenched free-energy \eqref{eq:f_raw} as
%
%\begin{defn}
%The quenched free energy of the Hopfield Hamiltonian accounting for supervised Hebbian learning is defined as
\begin{equation}
f_{\alpha,\beta,\rho}:=-\frac{1}{\beta N}\mathbb{E}\log\left\{ \sum_{\{\sigma\}}\int\prod_{\mu=2}^{K}\left(\frac{dz_{\mu}}{\sqrt{2\pi}}\right)\exp\left[-\sum_{\mu=2}^{K}\frac{z_{\mu}^{2}}{2}+\sqrt{\frac{\beta}{N}}\sum_{\mu=2}^{K}\sum_{i=1}^{N}\lambda_{i}^{\mu}z_{\mu}\sigma_{i}+\frac{\beta N}{2}{n}^{2}(1+\rho)\right]\right\} \label{eq:f_lin}
\end{equation}
where the operator $\mathbb{E}:=\mathbb{E}_{\chi^{1}}\mathbb{E}_{\xi^{1}}\mathbb{E}_{\lambda}$ is defined by Definition \eqref{def:op_media} and by
\begin{eqnarray}
\mathbb{E}_{\lambda}f(\boldsymbol{\lambda}) &:=& \int\prod_{a=1}^{M}\prod_{i=1}^{N}\left\{ \frac{d\lambda_{i}^{\mu}}{\sqrt{2\pi}}\exp\left[-\frac{\left(\lambda_{i}^{\mu}\right)^{2}}{2}\right]\right\} f(\boldsymbol{\lambda}),
\end{eqnarray}
where $f$ is a generic function.
%\end{defn}

We are now ready to introduce the Guerra functional that will guide us toward the solution of the model. In a nutshell, the idea is that the Guerra functional interpolates between two extrema: one corresponds to the original model (and the functional coincides with the model free-energy), the other corresponds to a solvable model (and the functional coincides with the free-energy of a one-body model); therefore, one can solve for the Guerra functional in this second extremum and then propagate the solution back to the original model via the fundamental theorem of calculus. We start this journey by giving the following 
\begin{defn}
The Guerra interpolating functional $G_{\alpha,\beta,\rho}(t,J_{m})$ is defined as: 

\begin{equation}
\begin{split}G_{\alpha,\beta,\rho}(t,J_{m})=-\frac{1}{\beta N}\mathbb{E}\log\Big\{\sum_{\{\sigma\}}\int\prod_{\mu=2}^{K}\left(\frac{dz_{\mu}}{\sqrt{2\pi}}\right)\exp\Big[-\frac{1-(1-t)\beta(1-\langle q\rangle)}{2}\sum_{\mu=2}^{K}z_{\mu}^{2}+\sqrt{t}\sqrt{\frac{\beta}{N}}\sum_{\mu=2}^{K}\sum_{i=1}^{N}\lambda_{i}^{\mu}z_{\mu}\sigma_{i}+\\
+t\frac{\beta N}{2}{n}^{2}(1+\rho)+\beta\langle n\rangle(1+\rho)\left(1-t\right)N{n}+\sqrt{\alpha\beta\langle p\rangle\left(1-t\right)}\sum_{i=1}^{N}\theta_{i}\sigma_{i}+\sqrt{\beta\langle q\rangle\left(1-t\right)}\sum_{\mu=2}^{K}\psi_{\mu}z_{\mu}-J_{m}\beta Nm\Big]\Big\}
\end{split}
\label{eq:guerra_f}
\end{equation}
where $t \in [0,1]$ is the interpolation parameter, $J_m $ is an auxiliary field coupled to the Mattis magnetization of the archetype, the operator $\mathbb{E}$ is defined as 
\begin{equation}
\mathbb{E}:=\mathbb{E}_{\chi^{1}}\mathbb{E}_{\xi^{1}}\mathbb{E}_{\lambda}\mathbb{E}_{\theta}\mathbb{E}_{\psi}
\end{equation}
with
\begin{eqnarray}
\mathbb{E_{\theta}}f(\boldsymbol{\theta}) &:=& \int\prod_{\mu=2}^{K}\left\{ \frac{d\theta_{\mu}}{\sqrt{2\pi}}\exp\left[-\theta_{\mu}^{2}/2\right]\right\} f(\boldsymbol{\theta}),\\
\mathbb{E_{\psi}}f(\boldsymbol{\psi}) &:=& \int\prod_{i=1}^{N}\left\{ \frac{d\psi_{i}}{\sqrt{2\pi}}\exp\left[-\psi_{i}^{2}/2\right]\right\} f(\boldsymbol{\psi}),
\end{eqnarray}
and the brackets $\langle \cdot \rangle$ represent the Boltzmann-Gibbs measure, as further specified in Definitions \eqref{def:aves}-\eqref{def:rsansatz}.
\end{defn}
\begin{rem}
By setting $t=1$ and $J_{m}=0$ in Guerra's functional, we recover immediately the correct expression for the free energy of the model \eqref{eq:f_lin}: $G_{\alpha,\beta,\rho}(t=1,J_{m}=0) = f_{\alpha,\beta,\rho}$. On the other hand, by setting $t=0$ we end up with a one-body system that is exactly solvable. The fictitious field yielded by $J_m$ will allow us to determine the expectation of the magnetization by deriving the free energy, namely, $\langle m\rangle = \left. \frac{dG_{\alpha,\beta,\rho}(t=1,J_{m})}{d J_m} \right |_{J_m=0}$. 
\end{rem}
\begin{defn}\label{def:aves} 
The functional \eqref{eq:guerra_f} yields an interpolating Boltzmann-Gibbs measure under which the expectation of the generic function $f(\boldsymbol{\sigma},\boldsymbol{z})$
reads as
\begin{equation}
\left\langle f(\boldsymbol{\sigma},\boldsymbol{z})\right\rangle _{t}:=\frac{\sum_{\{\sigma\}}\int\prod_{\mu=2}^{K}\left(\frac{dz_{\mu}}{\sqrt{2\pi}}\right)f(\boldsymbol{\sigma},\boldsymbol{z})\cdot\exp\left[\mathcal{B}(\boldsymbol{\sigma},\boldsymbol{z})\right]}{\sum_{\{\sigma\}}\int\prod_{\mu=2}^{K}\left(\frac{dz_{\mu}}{\sqrt{2\pi}}\right)\exp\left[\mathcal{B}(\boldsymbol{\sigma},\boldsymbol{z})\right]}
\end{equation}
where $\mathcal{B}(\boldsymbol{\sigma},\boldsymbol{z})$ is the generalized Boltzmann-Gibbs weight given by
\begin{equation}
\begin{aligned}\mathcal{B}(\boldsymbol{\sigma},\boldsymbol{z}):= & -\frac{1-(1-t)\beta(1-\langle q\rangle)}{2}\sum_{\mu=2}^{K}z_{\mu}^{2}+\sqrt{t}\sqrt{\frac{\beta}{N}}\sum_{\mu=2}^{K}\sum_{i=1}^{N}\lambda_{i}^{\mu}z_{\mu}\sigma_{i}+\\  
\nonumber
 & +t\frac{\beta N}{2}{n}^{2}(1+\rho)+\beta\langle n\rangle(1+\rho)\left(1-t\right)N{n}+\sqrt{\alpha\beta\langle p\rangle\left(1-t\right)}\sum_{i=1}^{N}\theta_{i}\sigma_{i}+\sqrt{\beta\langle q\rangle\left(1-t\right)}\sum_{\mu=2}^{K}\psi_{\mu}z_{\mu}-J_{m}\beta Nm.
\end{aligned}
\end{equation}
\end{defn}
%
%
%%%%%. RS %%%%%%%%  
\begin{defn}\label{def:rsansatz} 
\textbf{Theoretical setting (Replica Symmetric Ansatz):} In the replica-symmetric framework we assume that, in the thermodynamic limit $N\to\infty$, the following variance and covariance 
\begin{eqnarray}
\Delta\left[n_{1}^{2}\right] &:=& \mathbb{E}\left\langle \left({n}-\langle n\rangle\right)^{2}\right\rangle _{t}   =\mathbb{E}\left\langle {n}^{2}\right\rangle _{t}+2\langle n\rangle\mathbb{E}\left\langle n\right\rangle _{t}+\langle n\rangle^{2}\\
\Delta\left[q_{12}p_{12}\right] &:=& \mathbb{E}\left\langle \left(q_{12}-\langle q\rangle\right)\left(p_{12}-\langle p\rangle\right)\right\rangle _{t}  =\mathbb{E}\left\langle q_{12}p_{12}\right\rangle _{t}-\langle q\rangle\mathbb{E}\left\langle p_{12}\right\rangle _{t}-\langle p\rangle\mathbb{E}\left\langle q_{12}\right\rangle _{t}+\langle p\rangle\langle q\rangle
\label{eq:deltas}
\end{eqnarray}
go to $0$, namely $\lim_{N \to \infty}\Delta\left[n_{1}^{2}\right] =0$ and $\lim_{N \to \infty}\Delta\left[q_{12}p_{12} \right] =0$. 
Which implies
\begin{eqnarray}
\lim_{N\to\infty} \langle n \rangle_t &=& \langle n \rangle,\\
\lim_{N\to\infty} \langle q_{12} \rangle_t &=& \langle q \rangle,\\
\lim_{N\to\infty} \langle p_{12} \rangle_t &=& \langle p \rangle.
\label{eq:deltas}
\end{eqnarray}
\end{defn}
%%%%%%%%%%%%%%%

%
As we want to calculate the free energy (i.e., Guerra's functional at $t=1$) by evaluating $G_{\alpha,\beta,\rho}(t=0,J_{m})$ and then propagating it to $t=1$, a technical aspect we need to preliminary address is the $t-$derivative, or the streaming, of Guerra's interpolating functional and this is achieved in the next
\begin{lemma}  \label{lemma1}
The streaming equation for Guerra's functional $G_{\alpha,\beta,\rho}(t,J_{m})$ is
\begin{equation}
\frac{d}{dt}G_{\alpha,\beta,\rho}(t,J_{m})=\frac{\langle n\rangle^{2}}{2}(1+\rho)+\frac{\alpha}{2}\langle p\rangle(1-\langle q\rangle)+\Delta\left[q_{12}p_{12}\right]-\frac{1}{2}\Delta\left[{n}^{2}\right].\label{eq:streaming-1}
\end{equation}
\end{lemma}
 
\begin{proof}
By a direct evaluation 
\begin{eqnarray}
\frac{d}{dt}G_{\alpha,\beta,\rho}(t,J_{m})&=&\frac{1}{\beta N}\mathbb{E}\Biggr(\frac{\beta\left(1-\langle q\rangle\right)}{2}\sum_{\mu=2}^{K} \langle z_{\mu}^{2} \rangle_{t}-\frac{1}{2\sqrt{t}}\sqrt{\frac{\beta}{N}}\sum_{\mu=2}^{K}\sum_{i=1}^{N}\lambda_{i}^{\mu} \langle z_{\mu}\sigma_{i} \rangle_{t}+\\ \nonumber
&-&\frac{\beta N(1+\rho)}{2} \langle{n}^{2} \rangle_{t}+\beta\langle n\rangle(1+\rho)N \langle n \rangle_{t}+\frac{\sqrt{\beta\alpha\langle p\rangle}}{2\sqrt{1-t}}\sum_{i=1}^{N}\theta_{i} \langle\sigma_{i} \rangle_{t}+\frac{\sqrt{\beta\langle q\rangle}}{2\sqrt{1-t}}\sum_{\mu=2}^{K}\psi_{\mu}\langle z_{\mu} \rangle_{t}\Biggr),
\label{eq:streaming_step_der}
\end{eqnarray}
Now, we can introduce the order parameters in the above expression,  obtaining
\begin{eqnarray}
\frac{d}{dt}G_{\alpha,\beta,\rho}(t,J_{m})&=& \mathbb{E}\Biggr(\frac{\beta\alpha\left(1-\langle q\rangle\right)}{2}\left\langle p_{11}\right\rangle _{t}-\frac{\alpha\beta}{2}\left(\left\langle p_{11}\right\rangle _{t}-\left\langle p_{12}q_{12}\right\rangle _{t}\right)+\\ \nonumber
&-&\frac{1+\rho}{2}\left\langle {n}^{2}\right\rangle _{t}+(1+\rho)\langle n\rangle\left\langle n\right\rangle _{t}+\frac{\beta\alpha\langle p\rangle}{2}\left(1-\left\langle q_{12}\right\rangle _{t}\right)+\frac{\beta\alpha\langle q\rangle}{2}\sum_{\mu=2}^{K}\left(\left\langle p_{11}\right\rangle _{t}-\left\langle p_{12}\right\rangle _{t}\right)\Biggr).
\label{eq:streaming_step_wickeval-1}
\end{eqnarray}
By expressing in terms of mean values, variances and covariances  all the correlation functions contained in the above expression and by forcing all the $\Delta$'s to go to $0$ in the thermodynamic limit (i.e., by  assuming the RS ansatz), we obtain (\ref{eq:streaming-1}). \end{proof}
We are now ready to state the main
\begin{thm}
The replica symmetric free-energy of the HN implementing the supervised Hebbian learning (\ref{model}), in the infinite volume limit $N \to \infty$ and large dataset scenario $M \gg 1$, can be expressed in terms of control and order parameters as
\begin{equation}
\begin{aligned}f_{\alpha,\beta,\rho}= & \frac{\alpha}{2\beta}\log\left[1-\beta\left(1-\langle q\rangle\right)\right]-\frac{\alpha}{2}\frac{\langle q\rangle}{1-\beta\left(1-\langle q\rangle\right)}+\frac{\langle n\rangle^{2}}{2}(1+\rho)+\frac{\alpha}{2}\langle p\rangle(1-\langle q\rangle)\\
 & -\frac{1}{\beta}\log2-\frac{1}{\beta}\mathbb{E}_{\chi}\mathbb{E}_{\theta}\log\cosh \left.\Big(\beta \langle n\rangle \frac{1}{Mr}\sum_{a=1}^{M}\chi^{a}-J_{m}\beta+\sqrt{\alpha\beta\langle p\rangle}\theta\Big)  \right \vert_{J_m=0}.
\end{aligned}
\label{eq:rs_freeenergy}
\end{equation} 
\end{thm}
\begin{proof}
Via the fundamental theorem of calculus we write
\begin{equation}
f_{\alpha,\beta,\rho}(J_m) = G_{\alpha,\beta,\rho}(t=1,J_{m})=G_{\alpha,\beta,\rho}(t=0,J_{m})+\int_{0}^{1}dt^{\prime}\frac{dG_{\alpha,\beta,\rho}}{dt}(t^{\prime},J_{m})\label{eq:calculus}
\end{equation}
so that we are left with the evaluation of two terms, the Cauchy datum $G_{\alpha,\beta,\rho}(t=0,J_{m})$, that can be faced by direct calculation, and the integral of the streaming of Guerra's functional, that was provided in Lemma \ref{lemma1}. It will be in this last computation that the RS assumption will allow us to obtain an explicit expression for the free energy in terms of the order parameters and thus solve for the model, as shown hereafter.
\newline
We start by evaluating $G_{\alpha,\beta,\rho}(t=0,J_{m})$:
\begin{eqnarray}\nonumber
G_{\alpha,\beta,\rho}(t=0,J_{m}) &=& -\frac{1}{\beta N}\mathbb{E}\log\Big\{\sum_{\{\sigma\}}\int\prod_{\mu=2}^{K}\left(\frac{dz_{\mu}}{\sqrt{2\pi}}\right)\exp\Big[-\frac{1-\beta\left(1-\langle q\rangle\right)}{2}\sum_{\mu=2}^{K}z_{\mu}^{2}+\sqrt{\beta\langle q\rangle}\sum_{\mu=2}^{K}\psi_{\mu}z_{\mu}\\
&+&\frac{\beta\langle n\rangle}{rM}\sum_{i,a=1}^{N,M}\xi_{i}^{1}\chi_{i}^{1a}\sigma_{i}+\sqrt{\alpha\beta\langle p\rangle}\sum_{i=1}^{N}\theta_{i}\sigma_{i}-J_{m}\beta\sum_{i=1}^{N}\xi_{i}^{1}\sigma_{i}\Big]\Big\},
\label{eq:one_body_1}
\end{eqnarray}
as this is a trivial one-body problem (i.e., its probability structure is completely factorized), we immediately can write
\begin{eqnarray}\nonumber
G_{\alpha,\beta,\rho}(t=0,J_{m}) &=& -\frac{1}{\beta N}\sum_{\mu=2}^{K}\mathbb{E}\log\Big\{\int\left(\frac{dz_{\mu}}{\sqrt{2\pi}}\right)\exp\Big[-\frac{1-\beta\left(1-\langle q\rangle\right)}{2}z_{\mu}^{2}+\sqrt{\beta\langle q\rangle}\psi_{\mu}z_{\mu}\Big]\Big\}\\
&-& \frac{1}{\beta N}\sum_{i=1}^{N}\mathbb{E}\log\Big\{\sum_{\sigma_{i}=\pm1}\exp\Big[\beta\langle n\rangle\frac{1}{Mr}\sum_{a=1}^{M}\xi_{i}^{1}\chi_{i}^{1a}\sigma_{i}+\sqrt{\alpha\beta\langle p\rangle}\theta_{i}\sigma_{i}-J_{m}\beta\xi_{i}^{1}\sigma_{i}\Big]\Big\},
\label{eq:one_body_2}
\end{eqnarray}
hence, by direct evaluation of the Gaussian integral in $z$ and sum in $\sigma$, we reach
\begin{eqnarray}\nonumber
G_{\alpha,\beta,\rho}(t=0,J_{m}) &=& -\frac{\alpha}{\beta}\mathbb{E}_{\psi}\log\Big\{\frac{1}{\sqrt{1-\beta\left(1-\langle q\rangle\right)}}\exp\left[\frac{\beta\langle q\rangle\psi^{2}}{2\left[1-\beta\left(1-\langle q\rangle\right)\right]}\right]\Big\}\\
&-&\frac{1}{\beta}\log2-\frac{1}{\beta}\mathbb{E}_{\chi}\mathbb{E}_{\theta}\log\Big\{\cosh\Big[\beta\langle n\rangle\frac{1}{Mr}\sum_{a=1}^{M}\chi^{a}-J_{m}\beta+\sqrt{\alpha\beta\langle p\rangle}\theta\Big]\Big\}
\label{eq:one_body_3}
\end{eqnarray}
where the first average can be carried out directly returning
\begin{eqnarray}\nonumber
G_{\alpha,\beta,\rho}(t=0,J_{m}) &=& \frac{\alpha}{2\beta}\log\left[1-\beta\left(1-\langle q\rangle\right)\right]-\frac{\alpha}{2\beta}\frac{\beta\langle q\rangle}{1-\beta\left(1-\langle q\rangle\right)}\\
&-& \frac{1}{\beta}\log2-\frac{1}{\beta}\mathbb{E}_{\chi}\mathbb{E}_{\theta}\log\cosh\left(\beta\langle n\rangle\frac{1}{Mr}\sum_{a=1}^{M}\chi^{a}-J_{m}\beta+\sqrt{\alpha\beta\langle p\rangle}\theta\right).
\label{eq:one_body_4}
\end{eqnarray}
Now, by using Eq.~\eqref{eq:streaming-1} into the initial scheme \eqref{eq:calculus} we reach
\begin{equation}
\begin{aligned}f_{\alpha,\beta,\rho}(J_{m})= & \frac{\alpha}{2\beta}\log\left[1-\beta\left(1-\langle q\rangle\right)\right]-\frac{\alpha}{2\beta}\frac{\beta\langle q\rangle}{1-\beta\left(1-\langle q\rangle\right)}+\frac{\langle n\rangle^{2}}{2}(1+\rho)+\frac{\alpha}{2}\langle p\rangle(1-\langle q\rangle)+\Delta\left[q_{12}p_{12}\right]-\frac{1}{2}\Delta\left[n_{1}^{2}\right]\\
 & -\frac{1}{\beta}\log2-\frac{1}{\beta}\mathbb{E}_{\chi}\mathbb{E}_{\theta}\log\cosh\left(\beta\langle n\rangle\frac{1}{M}\sum_{a=1}^{M}\chi^{a}-J_{m}\beta+\sqrt{\alpha\beta\langle p\rangle}\theta\right),
\end{aligned}
\end{equation}
such that, finally, sending $N\to\infty$, under the RS assumption, at $J_m=0$ we obtain Eq.~\eqref{eq:rs_freeenergy}.
\end{proof}
\begin{cor}
In the infinite volume limit $N \to \infty$ and large dataset scenario $M \gg 1$, the replica-symmetric self-consistent equations for the evolution of the order parameters of the HN implementing the supervised Hebbian learning \eqref{model} in the space of the tuneable parameters are
\begin{equation}
\begin{aligned}\langle p\rangle & =\frac{\beta\langle q\rangle}{\left[1-\beta\left(1-\langle q\rangle\right)\right]^{2}},\\
\langle n\rangle & =\frac{1}{1+\rho}\mathbb{E}_{\chi}\mathbb{E}_{\theta}\frac{1}{Mr}\sum_{a=1}^{M}\chi^{a}\tanh\left(\beta\langle n\rangle\frac{1}{Mr}\sum_{a=1}^{M}\chi^{a}+\sqrt{\alpha\beta\langle p\rangle}\theta\right),\\
\langle q\rangle & =\mathbb{E}_{\chi}\mathbb{E}_{\theta}\tanh^{2}\left(\beta\langle n\rangle\frac{1}{Mr}\sum_{a=1}^{M}\chi^{a}+\sqrt{\alpha\beta\langle p\rangle}\theta\right),\\
\langle m\rangle & =\mathbb{E}_{\chi}\mathbb{E}_{\theta}\tanh\left(\beta\langle n\rangle\frac{1}{Mr}\sum_{a=1}^{M}\chi^{a}+\sqrt{\alpha\beta\langle p\rangle}\theta\right).
\end{aligned}
\label{eq:sce_Mfinite}
\end{equation}
\end{cor}
\begin{proof}
Let us consider the explicit expression for the quenched free-energy given in \eqref{eq:rs_freeenergy} and extremize w.r.t. the order parameters: 
\begin{equation}
\frac{\partial f_{\alpha, \beta, \rho}(J_{m})}{\partial\langle n_{1}\rangle}=\frac{\partial f_{\alpha, \beta, \rho}(J_{m})}{\partial\langle q\rangle}=\frac{\partial f_{\alpha, \beta, \rho}(J_{m})}{\partial\langle p\rangle}=0,
\end{equation}
we also pose, by construction,
\begin{equation}
\frac{\partial f_{\alpha, \beta, \rho}(J_{m})}{\partial J_{m}}=\langle m\rangle.
\end{equation}
With some algebra and by sending $J_{m}\to0$, we obtain eqs. \eqref{eq:sce_Mfinite}.
\end{proof}
\begin{cor}
In the large dataset scenario $M \gg 1$, the replica-symmetric self-consistent equations \eqref{eq:sce_Mfinite} for the HN \eqref{model} implementing the supervised Hebbian learning can be expressed as
\begin{equation}
\begin{aligned}\langle n\rangle & =\frac{1}{1+\rho}\frac{\langle m\rangle}{1-\frac{\rho}{1+\rho}\beta(1-\langle q\rangle)},\\
\langle q\rangle & =\mathbb{E}_{z}\tanh^{2}\left[\beta\langle n\rangle+z\beta\sqrt{\langle n\rangle^{2}\rho+\frac{\alpha\langle q\rangle}{\left[1-\beta\left(1-\langle q\rangle\right)\right]^{2}}}\right],\\
\langle m\rangle & =\mathbb{E}_{z}\tanh\left[\beta\langle n\rangle+z\beta\sqrt{\langle n\rangle^{2}\rho+\frac{\alpha\langle q\rangle}{\left[1-\beta\left(1-\langle q\rangle\right)\right]^{2}}}\right].
\end{aligned}
\label{eq:sce}
\end{equation}
\end{cor}
\begin{proof}
Note that, for $M \gg 1$, we can write
\begin{equation}
\frac{1}{M}\sum_{a=1}^{M}\chi^{a}\sim r+\lambda\sqrt{\frac{1-r^{2}}{M}}~\mathrm{with}~\lambda\sim\mathcal{N}(0,1).
\end{equation}
First, let us tackle $\langle n_{1}\rangle$ as follows
\begin{eqnarray}
\nonumber
(1+\rho)\langle n\rangle &=& \mathbb{E}_{\lambda}\mathbb{E}_{\theta}\left(1+\sqrt{\rho}\lambda\right)\tanh\left(\beta\langle n\rangle\left(1+\sqrt{\rho}\lambda\right)+\sqrt{\alpha\beta\langle p\rangle}\theta\right)=\\
\nonumber
&=& \langle m\rangle+\sqrt{\rho}\mathbb{E}_{\lambda}\mathbb{E}_{\theta}\partial_{\lambda}\tanh\left[\beta\langle n\rangle\left(1+\sqrt{\rho}\lambda\right)+\sqrt{\alpha\beta\langle p\rangle}\theta\right]=\\
&=& \langle m\rangle+\beta\langle n\rangle\rho(1-\langle q\rangle)
\label{eq:sce_n_Minf}
\end{eqnarray}
then we move on to $\langle m\rangle$
\begin{eqnarray}
\nonumber
\langle m\rangle &=& \mathbb{E}_{\lambda}\mathbb{E}_{\theta}\tanh\left[\beta\langle n\rangle\left(1+\sqrt{\rho}\lambda\right)+\sqrt{\alpha\beta\langle p\rangle}\theta\right]=\\
\nonumber
&=& \mathbb{E}_{\lambda}\mathbb{E}_{\theta}\tanh\left[\beta\langle n\rangle+\sqrt{\left(\beta\langle n\rangle\right)^{2}\rho}\lambda+\sqrt{\alpha\beta\langle p\rangle}\theta\right]=\\
&=& \mathbb{E}_{z}\tanh\left[\beta\langle n\rangle+z\sqrt{\left(\beta\langle n\rangle\right)^{2}\rho+\alpha\beta\langle p\rangle}\right], 
\label{eq:sce_m_Minf}
\end{eqnarray}
and, analogously for $\langle q\rangle$, we obtain \eqref{eq:sce}.
\end{proof}
As a last manipulation we work out explicitly the self-consistencies in the zero fast-noise limit $\beta \to \infty$ as this allows a characterization of the ground state of the network too; results are provided in the next
\begin{cor}
The zero-temperature $\beta \to \infty$, large dataset scenario $M\gg 1$ and infinite volume limit $N \to \infty$ of the self-consistencies for the order parameters of the Hopfield model implementing the supervised Hebbian learning \eqref{model} read as
\begin{equation}
\begin{aligned}\langle m\rangle & =\langle n\rangle\left[1+\rho(1-\Delta)\right],\\
G & =\sqrt{2\langle n\rangle^{2}\rho+2\frac{\alpha}{\left(1-\Delta\right)^{2}}},\\
\Delta & =\frac{1}{G}\partial\mathrm{erf}\left(\frac{\langle n\rangle}{G}\right),\\
\langle n\rangle & =\frac{1}{1+\rho(1-\Delta)}\mathrm{erf}\left(\frac{\langle n\rangle}{G}\right),
\end{aligned}
\label{eq:sce_Tzero}
\end{equation}
where 
\begin{equation}
\begin{aligned}\mathrm{erf}(x):=\int_{0}^{x}\frac{2}{\sqrt{\pi}}\exp(-t^{2})\,dt,\\
\mathrm{\partial erf}(x):=\frac{2}{\sqrt{\pi}}\exp(-x^{2}).
\end{aligned}
\end{equation}
\end{cor}
\begin{proof}
Following the strategy already used by AGS (e.g., see \cite{CKS}) we pose $\Delta:=\beta(1-\langle q\rangle)$ and we use this definition in Eq.~\eqref{eq:sce} obtaining
\begin{eqnarray}
\langle n\rangle &=& \frac{1}{1+\rho}\frac{\langle m\rangle}{1-\frac{\rho}{1+\rho}\Delta},\\
\Delta &=& \beta\left\{ 1-\mathbb{E}_{z}\tanh^{2}\left[\beta\langle n\rangle+z\beta\sqrt{\langle n\rangle^{2}\rho+\frac{\alpha\left(1-\frac{\Delta}{\beta}\right)}{\left[1-\Delta\right]^{2}}}\right]\right\}, \\
\langle m\rangle &=& \mathbb{E}_{z}\tanh\left[\beta\langle n\rangle+z\beta\sqrt{\langle n\rangle^{2}\rho+\frac{\alpha\left(1-\frac{\Delta}{\beta}\right)}{\left[1-\Delta\right]^{2}}}\right].
\label{eq:sce-1}
\end{eqnarray}
Then, we add a field $\beta x$ inside the hyperbolic tangents,
as this allows to rewrite $\Delta$ as $\frac{\partial m(x)}{\partial x}$ and thus
\begin{eqnarray}
\langle n\rangle &=& \frac{1}{1+\rho}\frac{\langle m\rangle}{1-\frac{\rho}{1+\rho}\Delta},\\
\Delta &=& \frac{\partial m(x)}{\partial x},\\
\langle m\rangle(x) &=& \mathbb{E}_{z}\tanh\left[\beta x+\beta\langle n\rangle+z\beta\sqrt{\langle n\rangle^{2}\rho+\frac{\alpha\left(1-\frac{\Delta}{\beta}\right)}{\left[1-\Delta\right]^{2}}}\right],
\label{eq:sce_Tzero_2}
\end{eqnarray}
such that, by taking the limit $\beta\to\infty$ and afterwards $x\to0$,
we get \eqref{eq:sce_Tzero}.
\end{proof}
The numerical solution of the self-consistent equations \eqref{eq:sce_Mfinite} were used to draw the phase diagram presented in the main text, while here we explicitly show the numerical solution for the magnetization in Fig.~\ref{Bruttone} $a$ and results obtained by MC simulations in Fig.~\ref{Bruttone} $b$-$c$, and discuss their consistency. In particular, in Fig.~\ref{Bruttone} $b$-$c$ we compare the outcomes for the estimates of the magnetization $\langle m \rangle$ and its derivative $\partial \langle m \rangle / \partial \rho$ w.r.t. $\rho$ that plays as a susceptibility of system performance versus the dataset entropy. Notice that $\langle m \rangle$ exhibits a flex at a point $\rho_c$ that corresponds to the peak in the susceptibility and this point gets closer and closer to the transition point derived analytically. Further, by increasing the network size $N$, peaks in the susceptibility get sharper. Therefore, these sizes can nicely spot the existence and the location of the phase transition, on the other hand, the first-order nature of the transition is well evidenced by the numerical solution of self-consistencies, as shown in Fig.~\ref{Bruttone} $a$.

\begin{figure}[tb]
\noindent \begin{centering}
\includegraphics[width=0.9\textwidth]{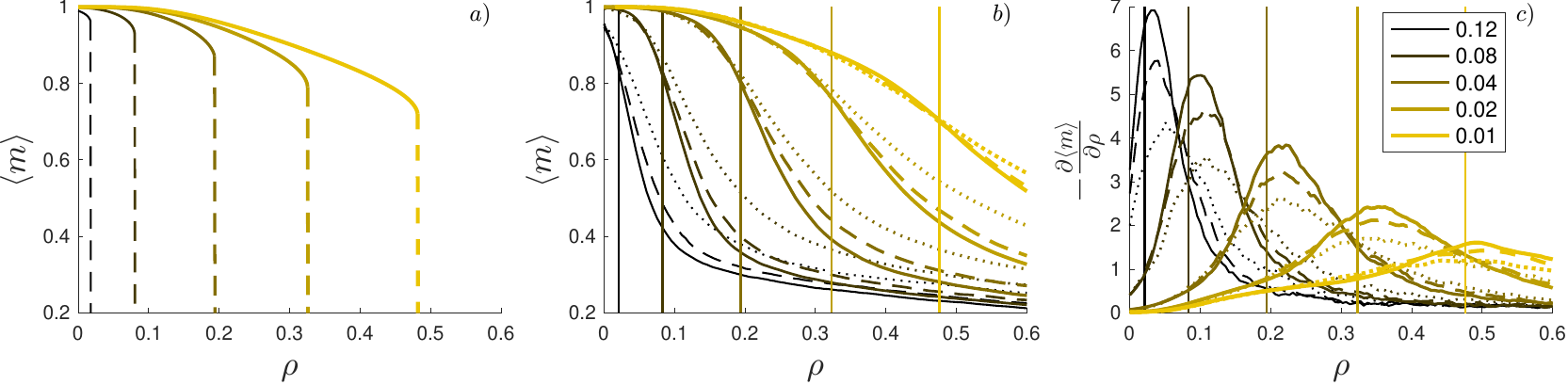}
\par\end{centering}
\caption{Theoretical solution and finite-size scaling by MC simulations. Panel $a$: numerical solution of the self-consistent equation \eqref{eq:sce} for the archetype magnetization versus $\rho$. Notice that, as expected for first-order transitions, the order parameter exhibits a first-order singularity and this happens at a value $\rho_c(\alpha,\beta)$ that is the one highlighted by vertical lines in panels $b$ and $c$. Panel $b$: expected magnetization for the archetype versus the entropy of the dataset resulting from MC simulations. Panel $c$: Susceptibility w.r.t. $\rho$ versus the dataset entropy resulting from MC simulations. For different choices of $\alpha$, corresponding to curves of different colors as explained by the legend, we varied the network size $N$ and the network load $K$ from $N \times K = 3000$ (solid line) to $N \times K = 14000$ (dotted line) and to $N \times K = 43000$ (dashed line); vertical lines represent the transition points $\rho_c(\alpha,\beta)$ as predicted theoretically. Note that, in panel $b$, the curves exhibit a flex approximately corresponding to the intersection with the vertical lines and in panel $c$ the curves peak approximately at the intersection with the vertical lines, and these matches get sharper as the network size is increased. Simulations and numerical solutions are obtained for $\beta=10$. }\label{Bruttone}
\end{figure}

\section{RBM training} \label{sec:RBM}
In this section we briefly address machine learning by one of its classical architecture, that is the RBM, in order to derive its supervised learning rules and compare their learning skills with those shown by the HN described in the previous sections, see Fig.~\ref{fig:cartoon}. For the sake of simplicity here, as standard in machine learning investigations, we fix $\beta=1$, with no loss of generality as the fast noise can always be re-introduced in the network by rescaling the couplings. In a statistical-mechanics framework RBMs are nothing but bipartite spin-glasses \cite{Bipartite} and, from this perspective we give the following
\begin{defn}
We consider a RBM built of two-layers made of, respectively, $N$ binary neurons $\sigma_i$, $i \in (1,...,N)$ and $K$ real-valued neurons $z_{\mu}$, $\mu \in \{1,...,K\}$ equipped with a standard Gaussian prior and whose Hamiltonian reads as
\begin{equation}
\label{eq:RBM_model}
\mathcal{H}^{\textrm{\tiny{RBM}}}(\boldsymbol{\sigma},\boldsymbol{z}|\boldsymbol{W})=-\frac{1}{\sqrt{N}}\sum_{i,\mu=1}^{N,K}W_{i\mu}\sigma_{i}z_{\mu}
\end{equation}
where $\boldsymbol W \in \mathbb{R}^{N\times K}$ is the matrix of weights among the two layers, while there are no interactions within the same layer (whence the {\em restriction}). 
\end{defn}
\begin{defn}
The partition function of the RBM \eqref{eq:RBM_model} is
\begin{equation}
\mathcal{Z}^{\textrm{\tiny{RBM}}}(\boldsymbol{W})=\sum_{\{\sigma\}} \int \prod_{\mu=1}^K dz_{\mu} \exp\left(\sum_{i,\mu=1}^{N,K}W_{i\mu}\sigma_{i}z_{\mu}\right) \exp\left(-\frac12 \sum_{\mu=1}^K z_{\mu}^2\right), \label{linp}
\end{equation}
where the factor $\propto e^{-\left(\frac12 \sum_{\mu=1}^K z_{\mu}^2\right)}$ is the Gaussian prior for the $z$ neurons.
The joint probability density related to the partition function \eqref{linp}
assumes the following form
\begin{equation}
\mathcal{P}(\boldsymbol{\sigma},\boldsymbol{z}|\boldsymbol{W})=\frac{1}{\mathcal{Z}^{\textrm{\tiny{RBM}}}(\boldsymbol{W})}\exp\left(\sum_{i,\mu=1}^{N,K}W_{i\mu}\sigma_{i}z_{\mu} - \sum_{\mu=1}^K \frac{z_{\mu}^2}{2}\right).\label{prob}
\end{equation}
\end{defn}
%
%%%%%%%%%%%%%%%%%%%%%%%%%%%%%%%%%%
%
\begin{figure}[tb]
\begin{centering}
\includegraphics[width=0.45\textwidth]{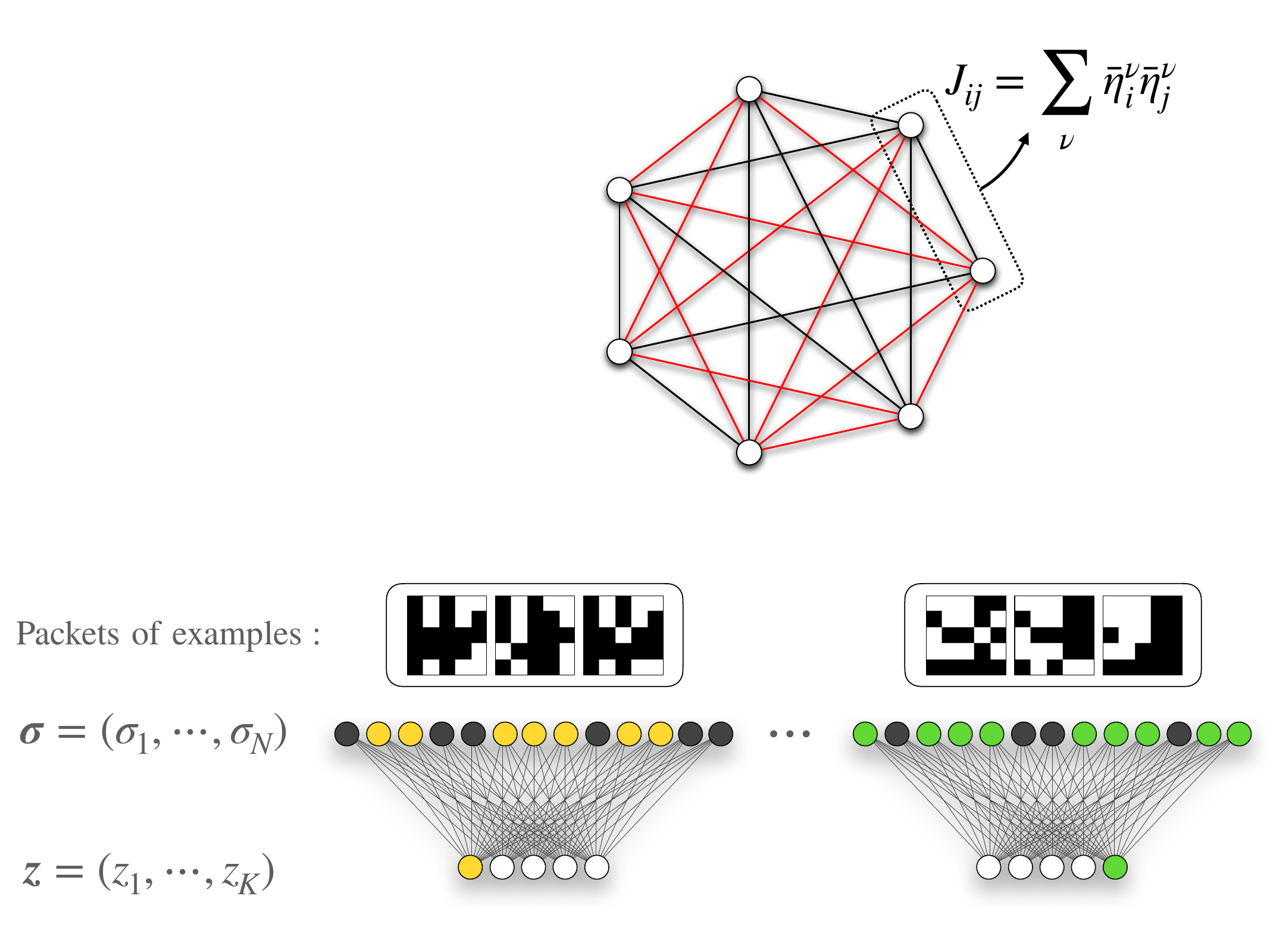}
\par\end{centering}
\caption{Schematic representation of the models investigated in this work. In the upper part of this figure we show the HN implementing the supervised Hebbian learning discussed in Secs.~\ref{Signal2Noise} and \ref{sec:SM}. Notice that the coupling between each neuron pair can be written as the scalar product between the empirical averages of the blocks $\boldsymbol \eta_i^{\mu} = (\eta_i^{\mu 1}, ..., \eta_i^{\mu M})$. In the lower part of the figure we show a RBM which is trained according to the grandmother-cell protocol: while the examples pertaining to the first archetype are presented and the visible neurons are clamped as $\boldsymbol \sigma = \boldsymbol \eta^{1 a}$ for $a=1,...,M$, the first hidden neuron is clamped as active, while the remaining hidden neurons are set quiescent. The same operation is repeated for the other group pf examples.}\label{fig:cartoon}
\end{figure}
%
%%%%%%%%%%%%%%%%%%%%
As a learning criterion, we minimize the Kullback-Leibler cross-entropy $D_{KL}(\mathcal{Q} \Vert \mathcal{P})$ between the model distribution $\mathcal{P}$ and the target distribution $\mathcal Q$, that is the empirical distribution of the experimental dataset.
For the moment we retain an arbitrary expression for $\mathcal Q$, depending on both $\boldsymbol \sigma$ and $\boldsymbol z$, and then we shall discuss specific choices for unsupervised and supervised protocols.
The KL cross-entropy is defined as
\begin{equation}
D_{KL}(\mathcal Q \Vert \mathcal{P})=-\sum_{\{\boldsymbol \sigma, \boldsymbol z\}} \mathcal Q(\boldsymbol \sigma, \boldsymbol z)\log\frac{\mathcal{P}(\boldsymbol{\sigma},\boldsymbol{z}|\boldsymbol{W})}{ \mathcal Q(\boldsymbol{\sigma},\boldsymbol{z})}\label{dist}
\end{equation}
and its minimization takes place by properly tuning the weights $\boldsymbol{W}$ in order for $\mathcal{P}(\boldsymbol{\sigma},\boldsymbol{z}|\boldsymbol{W})$ to be as ``close'' as possible to $\mathcal Q(\boldsymbol{\sigma},\boldsymbol{z})$. This can be accomplished by the gradient descent method that yields the following iterative rule
for $\boldsymbol{W}$
\begin{equation}
W_{i \mu}^{(n+1)}=W_{i \mu}^{(n)}-\epsilon\frac{dD_{KL}(\mathcal Q \Vert \mathcal{P})}{dW_{i \mu}},\label{iteration}
\end{equation}
where $\epsilon>0$ is the {\em learning rate}.
\begin{defn}
Given an observable $O(\boldsymbol{\sigma},\boldsymbol{z})$ depending
on the variables $\{\boldsymbol{\sigma},\boldsymbol{z}\}$, we indicate
with $\langle\cdot\rangle_{\textrm{free}}$ the standard average with respect
to the model probability density \eqref{prob}
\begin{equation}
\langle O\rangle_{\textrm{free}}:=\sum_{\{\boldsymbol \sigma, \boldsymbol z\}}\mathcal{P}(\boldsymbol{\sigma},\boldsymbol{z}|\boldsymbol{W})O(\boldsymbol{\sigma},\boldsymbol{z}),
\end{equation}
and with $\langle\cdot\rangle_{\textrm{clamped}}$ the average with respect
to the empirical probability density $\mathcal Q(\boldsymbol{\sigma},\boldsymbol{z})$
\begin{equation}
\langle O\rangle_{\textrm{clamped}}:=\sum_{\{\boldsymbol \sigma, \boldsymbol z \}} \mathcal Q(\boldsymbol{\sigma},\boldsymbol{z})O(\boldsymbol{\sigma},\boldsymbol{z}).
\end{equation}
\end{defn}
\begin{rem}
By using the fact that
\begin{equation}
\frac{dD_{KL}(\mathcal Q \Vert \mathcal{P})}{dW_{i \mu}}=\langle z_{\mu}\sigma_{i}\rangle_{\textrm{free}}-\langle z_{\mu}\sigma_{i}\rangle_{\textrm{clamped}},
\end{equation}
iteration \eqref{iteration} can be written as
\begin{equation}\label{sost1}
W_{i \mu}^{(n+1)}=W_{i \mu}^{(n)}+\epsilon\left(\langle z_{\mu}\sigma_{i}\rangle_{\textrm{clamped}}-\langle z_{\mu}\sigma_{i}\rangle_{\textrm{free}}\right).
\end{equation}
\end{rem}
As for the clamped average, we outline the following training modes:
\begin{defn}
\textbf{Supervised setting (grandmother-cell ansatz)}. 
%All computational experiments with the RBM \eqref{eq:RBM_model} have been carried out by generating a synthetic dataset of Rademacher archetypes $\{\boldsymbol \xi^{\mu}\}_{\mu=1,...,K}$ (see eq.~\ref{eq:Rade}), whence a set of corrupted examples 
Given a training set $\mathcal S := \{\boldsymbol \eta^{\mu a} \}_{\mu=1,...,K}^{a=1,...,M}$ generated according to Definition \eqref{def:dataset}, where for each item we are aware of the label $\mu$,
%\begin{equation}
%\eta_{i\mu}^{a}=\xi_{i}^{\mu}\chi_{i}^{\mu a}\quad\chi_{i}^{\mu a}\sim\frac{1+r}{2}\delta(\chi_{i}^{\mu a}-1)+\frac{1-r}{2}\delta(\chi_{i}^{\mu a}+1)
%\end{equation}
%where $i=1,\cdots,N\quad\mu=1,\cdots,K, \quad a=1,\cdots,M$ and $r$  tunes the quality of the dataset (see eqs.~\ref{eq:bern0}-\ref{eq:bern}). 
we envisage the following supervised-learning protocol to clamp the $\sigma$'s and the $z$'s variables: for every new example presented to the network, say $\boldsymbol \eta^{\nu a}$, we set $\boldsymbol \sigma = \boldsymbol \eta^{\nu a}$ and $\boldsymbol z = \boldsymbol z^{(\nu)}$, where $\boldsymbol z^{(\nu)}$ has entries $z^{(\nu)}_{\mu} = \delta_{\mu \nu}$. 
The target distribution therefore reads as
\begin{equation}
\mathcal Q^{\text{sup}}(\boldsymbol \sigma , \boldsymbol z) = \sum_{\mu,a} \delta (\boldsymbol \eta^{\mu a} - \boldsymbol \sigma) \delta (\boldsymbol z^{(\mu)} - \boldsymbol z)
\end{equation}
and the (batch) clamped-average appearing in the learning rule \eqref{sost1} reads as
\begin{equation}
\left\langle \sigma_{i}z_{\mu}\right\rangle _{\textrm{clamped}} \to \langle \sigma_i z_{\mu} \rangle_{\boldsymbol \sigma \& \boldsymbol z} = \frac{1}{M}\sum_{a=1}^{M}\eta_{i}^{\mu a} = \bar{\eta}_i^{\mu},
\end{equation}
where in the bracket subscript we highlighted that both kinds of degrees of freedom are clamped.
Thus, the iterative scheme \eqref{iteration} turns out to be
\begin{equation}\label{RBMsup}
W_{i,\mu}^{n+1} = W_{i,\mu}^{n} + \epsilon \left( \langle \sigma_i z_{\mu} \rangle_{\boldsymbol \sigma \& \boldsymbol z} -  \langle \sigma_i z_{\mu} \rangle_{\textrm{free}} \right).
\end{equation}
\end{defn}
\begin{defn}
\textbf{Unsupervised setting}. Given a training set $\mathcal S := \{\boldsymbol \eta^{\mu a} \}_{\mu=1,...,K}^{a=1,...,M}$ generated according to Definition \eqref{def:dataset}, where for each item the label $\mu$ is not disclosed, we envisage the following unsupervised-learning protocol to clamp the $\sigma$'s variables: for every new example presented to the network, say $\boldsymbol \eta^{\nu a}$, we set $\boldsymbol \sigma = \boldsymbol \eta^{\nu a}$, while the machine is left free to arrange its hidden degrees of freedom, namely it simply learns the statistical properties of the set of examples.
The target distribution therefore reads as
\begin{equation}
\mathcal Q^{\text{unsup}}(\boldsymbol \sigma) = \sum_{\mu,a} \delta (\boldsymbol \eta^{\mu a} - \boldsymbol \sigma), 
\end{equation}
and the (batch) clamped-average appearing in the learning rule \eqref{iteration} reads as
\begin{equation}
\left\langle \sigma_{i}z_{\mu}\right\rangle _{\textrm{clamped}} \to \langle \sigma_i z_{\mu} \rangle_{\boldsymbol \sigma}.
\end{equation}
Thus, the iterative scheme \eqref{iteration} turns out to be
\begin{equation}\label{RBMunsup}
W_{i,\mu}^{n+1} = W_{i,\mu}^{n} + \epsilon \left( \langle \sigma_i z_{\mu} \rangle_{\boldsymbol \sigma} -  \langle \sigma_i z_{\mu} \rangle_{\textrm{free}} \right).
\end{equation}
\end{defn}

As for the free average, it can be estimated by MC simulations or similar computational routes; here we followed the ``Persistent Contrastive Divergence'' criterion, see e.g., \cite{pers}. 
In particular, at the beginning of the simulation (in the first epoch)
the set of variables $(\boldsymbol \sigma, \boldsymbol z)$ is sampled randomly (i.e., $\sigma_i \in \{-1, +1\}$ is sampled by a symmetric Bernoulli distribution with parameter $1/2$ and $z_{\mu} \in \mathbb R$ is sampled by a standard Gaussian distribution, i.i.d. for any $i =1,...,N$ and $\mu =1,...,K$), then
they are updated for 10 steps by applying an alternating Gibbs sampling  \cite{Monasson}: 
$\boldsymbol{\sigma}$ is drawn from the following probability density
\begin{equation}
\mathcal{P}(\boldsymbol \sigma | \boldsymbol z, \boldsymbol W)=\frac{\mathcal{P}(\boldsymbol \sigma, \boldsymbol z| \boldsymbol W)}{\sum_{\{\boldsymbol \sigma\}}\mathcal{P}(\boldsymbol \sigma, \boldsymbol z| \boldsymbol W)}=\prod_{i=1}^{N}\frac{1}{2}\left[1+\tanh\left(\sum_{\mu=1}^{K}z_{\mu}W_{i \mu}\sigma_{i}\right)\right],
\end{equation}
and, analogously, $\boldsymbol{z}$ is drawn from
\begin{equation}
\mathcal{P}(\boldsymbol z| \boldsymbol \sigma, \boldsymbol W)=\frac{\mathcal{P}(\boldsymbol \sigma, \boldsymbol z| \boldsymbol W)}{\int d \boldsymbol z ~ \mathcal{P}(\boldsymbol \sigma, \boldsymbol z | \boldsymbol W)}=\prod_{\mu=1}^{K}\frac{1}{\sqrt{2\pi}}\exp\left[-\frac12 \left(z_{\mu}-\sum_{i=1}^{N}W_{i \mu}\sigma_{i}\right)^2 \right].
\end{equation}
Computationally, it is convenient to realize the sampling according
to the simplified rules:
\begin{eqnarray}
\sigma_{i}^{(t+1)} & = & \mathrm{sign}\left [ \tanh\left(\sum_{\mu=1}^{K}W_{i \mu}z_{\mu}^{(t)}\right)+\zeta_{i}^{t+1}\right] \quad\zeta_{i}^{t+1}\sim \mathcal U(-1,+1),\label{eq:mciterationBM}\\
z_{\mu}^{(t+1)} & = & \sum_{i=1}^{N}W_{i \mu}\sigma_{i}^{(t+1)} +\zeta_{\mu}^{t+1} \quad\quad\quad\quad\quad\quad\quad\zeta_{\mu}^{t+1}\sim \mathcal N(0,1).\label{eq:mciterationBM2}
\end{eqnarray}
At the end of each epoch averages are performed and the weights are updated. 

The learning rate $\epsilon$ is chosen according to
the Robbins-Monro scheme \cite{Monro}, namely $\epsilon_{n}=\frac{50}{100+n}$ that is
$\sum_{n=1}^{\infty}\epsilon_{n}=\infty,\sum_{n=1}^{\infty}\epsilon_{n}^{2}<\infty$.
Finally, the magnetization of the archetype $\langle m\rangle$ is cheaply
evaluated via Hinton's one-step MC approximation
\begin{equation}
\label{eq:1-one}
\langle m\rangle=\frac{1}{K}\sum_{\mu=1}^{K}\left|\frac{1}{N}\sum_{i=1}^{N}\xi_{i}^{\mu}\tanh\left(W_{i \mu}\right)\right|.
\end{equation}
Even if this is only an approximation of the magnetization, it has
very low variance and computational cost and still captures the essential
behavior of the model.

Results for supervised learning of structureless datasets are presented in Fig.~$2$ in the main text, where we showed that, remarkably, the values of $\alpha, \rho$ that ensure a successful information-processing are the same for both biological learning (i.e., by the HN via Hebb's rule) and artificial learning 
(i.e., by the RBM via contrastive divergence).

As evidenced in \cite{Leonelli,ALM-AMC2022,FAAB-IEEE} for structureless datasets, by training the RBM following the protocols described above one ends up with weight distributions that are sharply peaked and, in particular, the expected value for $W_{i \mu}$ corresponds to $\bar{\eta}_i^{\mu}$, therefore, in this case, the setting $W_{i \mu} = \bar{\eta}_i^{\mu}$ corresponds to a trained machine that is able to reconstruct archetypes and to generate new examples. 
To see this, one can look at the conditional probabilities: $\mathcal P_{\beta} (\boldsymbol z | \boldsymbol \sigma = \boldsymbol \eta^{\nu a}, \boldsymbol W) \propto \prod_{\mu} e^{-\beta (z_{\mu} - (\boldsymbol W \cdot \boldsymbol \eta^{\nu a})_{\mu})^2/2}$ should be peaked at $ \boldsymbol z^{(\nu)}$, while $\mathcal P_{\beta} (\boldsymbol \sigma | \boldsymbol z = \boldsymbol z^{(\nu)}, \boldsymbol W) \propto \prod_i e^{\beta \sigma_i W_{i \nu}}$ should be peaked at $\boldsymbol \xi^{\nu}$, therefore $\boldsymbol W_{\mu}$ should be simultaneously orthogonal to $\boldsymbol \eta^{\nu a}$ for any $\nu \neq \mu$ and parallel to $\boldsymbol \xi^{\mu}$ for any $\mu$, so the optimal solution is given by $\boldsymbol W = \bar{\boldsymbol \eta}$.
Notably, in the structured case, the setting $W_{i \mu} = \bar{\eta}_i^{\mu}$ does not correspond to an accomplished training but still it provides an effective pre-training as highlighted in \cite{FAAB-IEEE,Zilman}.

%We conclude this section with a last remark on the duality between HNs and RBMs \cite{Leonelli,BarraEquivalenceRBMeAHN,Marullo}. In fact, we understood that, as long as we identify the weights in the latter with the patterns in the former, the RBM results to have learnt the same patterns that the HN can retrieve and vice versa, also, the two models share the same phase diagrams, hence where the HN is able to retrieve, the RBM is able to correctly infer the archetype and perform generalizations. However, one could argue that the HN must deal with binary patterns (being these archetypes or examples) while, at the end of the RBM training, we have no guarantee that the entries of the weights aree digital, i.e. $W_{i \mu} \neq \pm 1$; this is not a real problem as we have shown in the past that the duality is robust w.r.t. this kind of perturbation, namely the HN keeps retrieving even if its entries are not sharply $\pm 1$ \cite{Peter1,Peter2}.

\section{Maximum Entropy Approach for the supervised Hebbian learning} \label{sec:MAE}
In this section we aim to reach the expression for the probability distribution of the trained RBM, that is $\mathcal{P}^{\textrm{\tiny{RBM}}}(\boldsymbol \sigma, \boldsymbol z | \boldsymbol W = \bar{\boldsymbol \eta})$, which, in turns, coincides with the probability distribution stemming from the linearized partition function $\mathcal{Z}^{\textrm{\tiny{HN}}}_{\beta=1}(\mathcal S)$ found for the supervised HN in Proposition \ref{prop:int_rep}, from another perspective. Let us recall the learning rule for the supervised RBM,
\begin{equation}
W_{i \mu}^{(n+1)}=W_{i \mu}^{(n)}+\epsilon\left(\langle z_{\mu}\sigma_{i}\rangle_{\textrm{clamped}}-\langle z_{\mu}\sigma_{i}\rangle_{\textrm{free}}\right),
\end{equation}
which makes use of the correlation function $\langle z_{\mu}\sigma_{i}\rangle_{\textrm{clamped}}$, thus this is the only information that the RBM requires from a dataset. In this section we will show how it is possible to retrieve both the RBM and the HN by looking for the least structured density function that reproduces the clamped correlation function.
This construction requires the minimal number of constraints in order to recover the key models of this work.
Let us spell out the crucial assumptions of this derivation:
\begin{itemize}
	\item the density function depends on the set of binary variables $\boldsymbol \sigma \in \{-1 ,+1\}^N$ and on the set of real variables $\boldsymbol z \in \mathbb{R}^K$,
	\item we require that the expected value of the $L_2-$norm of the $\boldsymbol{z}$ variables is bounded;
	\item we require that the density function correctly reproduces the correlations between the $\boldsymbol z's$ and the $\boldsymbol \sigma$'s variables.
\end{itemize}

\begin{defn}\label{def:maxscorr}
\label{cor}We introduce the following two-point empirical correlation functions
\begin{eqnarray}\label{ElleDue}
C_{z^{2}}&=&\sum_{\mu=1}^{K}\overline{ z_{\mu}^{2}},\\ \label{ElleUno}
C_{\sigma z}^{i,\mu}&=&\overline{\sigma_{i}z_{\mu}}.
\end{eqnarray}
where the bar denotes the empirical average evaluated over the examples making up the dataset $\mathcal S$.
\end{defn}
\begin{defn}
Given a probability density function $\mathcal{P}(\boldsymbol{\sigma},\boldsymbol{z})$, the expectation
of an observable $O(\boldsymbol{\sigma},\boldsymbol{z})$ is defined
as
\begin{equation}
\langle O\rangle_{\mathcal{P}}=\sum_{\sigma}\int\prod_{\mu=1}^{K}\left(\frac{dz_{\mu}}{\sqrt{2\pi}}\right)O(\boldsymbol{\sigma},\boldsymbol{z})\mathcal{P}(\boldsymbol{\sigma},\boldsymbol{z}),
\end{equation}
and, to lighten the notation, we introduce the trace operator $\text{Tr}$
\begin{equation}
\text{Tr}[O(\boldsymbol{\sigma},\boldsymbol{z})]:=\sum_{\boldsymbol \sigma}\int\left[\prod_{\mu=1}^{K}\frac{dz_{\mu}}{\sqrt{2\pi}}\right]O(\boldsymbol{\sigma},\boldsymbol{z})
\end{equation}
such that 
\begin{equation}
\langle O\rangle_{\mathcal{P}}=\text{Tr}[O(\boldsymbol{\sigma},\boldsymbol{z})\mathcal{P}(\boldsymbol{\sigma},\boldsymbol{z})].
\end{equation}
\end{defn}
\begin{thm}
The least structured probability distribution reproducing the sets of correlation functions $C_{z^{2}}$ and $C_{\sigma z}^{i,\mu}$ is
\begin{equation}
\mathcal{P}(\boldsymbol{\sigma},\boldsymbol{z} )= \frac{1}{\mathcal Z} \exp\left[ -\frac{\lambda_1 }{2}\sum_{\mu=1}^{K}z_{\mu}^{2}+\sum_{i,\mu=1}^{N,K}\Lambda_{i \mu}\sigma_{i}z_{\mu}\right]\label{prob-1}
\end{equation}
%where $\mathcal Z$ is the normalizing factor, the parameters $\boldsymbol \Lambda = \{ \Lambda_{i \mu} \}_{i=1,...,N}^{\mu=1,...,K}$ are Lagrangian multipliers to be set as $\Lambda_{i \mu} \propto \bar{\eta}_i^{\mu}$ in order constrain the correlation functions as prescribed in Definition \eqref{def:maxscorr}; this probability distribution corresponds to the density function of a trained RBM and, by marginalising w.r.t. the $z_\mu$ variables, we recover the density function of a HN with supervised Hebbian learning
where $\mathcal Z$ is the normalizing factor, the parameters $\lambda_1,  \{ \Lambda_{i \mu} \}_{i=1,...,N}^{\mu=1,...,K}$ are Lagrangian multipliers to be set in order constrain the correlation functions as prescribed in Definition \eqref{def:maxscorr} and therefore their value is dataset dependent; this probability distribution corresponds to the density function of a RBM and, by marginalising w.r.t. the $z_\mu$ variables, we recover the density function of a HN
\begin{equation}
\mathcal{P}(\boldsymbol{\sigma} )=\frac{1}{\mathcal{Z}}\exp\left[\frac{1}{2\lambda_{1}}\sum_{\mu=1}^{K}\sum_{i,j=1}^{N,N}\Lambda_{i \mu}\sigma_{i}\Lambda_{j \mu}\sigma_{j}\right].\label{comp1}
\end{equation}
\begin{proof}
We introduce the following Lagrangian 
\begin{equation}
%\begin{split} & S [\mathcal{P}]=-\text{Tr}(\mathcal{P}\log\mathcal{P})+\lambda_{0}[\text{Tr}(\mathcal{P})-1]+\frac{\lambda_{1}}{2}\left[\text{Tr}\left(\sum_{\mu=1}^{K}z_{\mu}^{2}\mathcal{P}\right)-C_{z^{2}}\right]+\\
% & +\sum_{i,\mu=1}^{N,K}\Lambda_{i \mu}\left[\text{Tr}(\sigma_{i}z_{\mu}\mathcal{P})-C_{\sigma z}^{i,\mu}\right]
%\end{split}
S [\mathcal{P}]=-\text{Tr}(\mathcal{P}\log\mathcal{P})+\lambda_{0}[\text{Tr}(\mathcal{P})-1]+\frac{\lambda_{1}}{2}\left[\text{Tr}\left(\sum_{\mu=1}^{K}z_{\mu}^{2}\mathcal{P}\right)-C_{z^{2}}\right]  +\sum_{i,\mu=1}^{N,K}\Lambda_{i \mu}\left[\text{Tr}(\sigma_{i}z_{\mu}\mathcal{P})-C_{\sigma z}^{i,\mu}\right],
\label{act}
\end{equation}
that accounts for maximum entropy and the above-mentioned constraints. By  extremizing $S[\mathcal{P}]$ with respect to $\mathcal{P}$ and $\lambda_{0,1}$, $\Lambda_{i \mu}$, playing as the Lagrangian multipliers, we obtain %\eqref{prob-1}.
%
%such that the probability function \eqref{prob-1} can be rewritten
%as
\begin{equation}
\mathcal{P}(\boldsymbol{\sigma},\boldsymbol{z})= \exp\left[1 - \lambda_0 -\frac{\lambda_{1}}{2}\sum_{\mu=1}^{K}z_{\mu}^{2}+\sum_{i,\mu=1}^{N,K}\Lambda_{i,\mu}\sigma_{i}z_{\mu}\right],\label{probc}
\end{equation}
and, posing $\mathcal Z = \exp(\lambda_0- 1)$, we recover \eqref{prob-1}.
By integrating \eqref{probc} with respect to $z_{\mu}$ we get the marginal, that can be recast in the probability distribution for the HN
\begin{equation}
\mathcal{P}(\boldsymbol{\sigma}) \propto \exp\left[\frac{1}{2\lambda_{1}}\sum_{\mu=1}^{K}\sum_{i,j=1}^{N,N}\Lambda_{i,\mu}\sigma_{i}\Lambda_{j,\mu}\sigma_{j}\right].%\label{comp1}
\end{equation}
\end{proof}
\end{thm}
\begin{rem}
The dataset-dependent expression for the Lagrangian multipliers $\lambda_{1}$, $\Lambda_{i \mu}$ that reproduce the specific correlation functions of our trained RBM model are
 %\begin{equation}
 %\exp(\lambda_{0}-1) = \mathcal{Z},
 %\end{equation}
\begin{equation}
\lambda_{1}=1,
\end{equation}
\begin{equation}
\Lambda_{i \mu}=\sqrt{\frac{\beta}{N\mathcal{R}}}\frac{1}{M}\sum_{a=1}^{M}\eta_{i}^{\mu a}.
\end{equation}
In fact, with this choice, we get 
\begin{equation}
\mathcal{P}(\boldsymbol {\sigma}, \boldsymbol{z} | \mathcal S) \propto \exp\left[-  \sum_{\mu=1}^K \frac{z_{\mu}^2}{2}  + \sqrt{ \frac{\beta}{N\mathcal{R}}} \frac{1}{M}\sum_{\mu=1}^{K}\sum_{i=1}^{N}\left(\sum_{a=1}^{M}\eta_{i}^{\mu a}\right)\sigma_{i} z_{\mu} \right].\label{pkan_bo}
\end{equation}
and, by marginalizing over the $\boldsymbol z$'s variable, we recover the probability distribution for the supervised HN defined in Eq.~\eqref{model}%$\mathcal{P}(\boldsymbol{\sigma})$ gets
\begin{equation}
\mathcal{P}(\boldsymbol{\sigma} | \mathcal S) \propto \exp\left[\frac{1}{2}\frac{\beta}{N\mathcal{R}}\sum_{\mu=1}^{K}\sum_{i,j=1}^{N,N}\left(\frac{1}{M}\sum_{a=1}^{M}\eta_{i}^{\mu a}\right)\left(\frac{1}{M}\sum_{a=1}^{M}\eta_{j}^{\mu a}\right)\sigma_{i}\sigma_{j}\right].\label{pkan}
\end{equation}
%namely, sharply the Boltzmann-Gibbs probability distribution for the supervised HN defined in Eq.~\eqref{model}.
\end{rem}
%
%
%%%%%%%%%%%%%%%%%%%%%%%%%%%%%%%%%%%%%%%%%

 \section{RBM for structured datasets} \label{sec:structured}
In this section we show how to generalize the previous grandmother-cell setting, worked out for RBMs dealing with structureless datasets, to the case of structured datasets. To this goal we must inspect the structure of the information content in the dataset: the underlying idea is that the machine inner-representation should be a ``carbon copy'' of the external information available, hence understanding its organization is pivotal to suitably generalize the processing scheme.

Before proceeding in that direction we introduce the notation. We consider a structured dataset $\{\boldsymbol {\breve\zeta}^{\mu a}\}_{\mu=1,...,K}^{a=1,...,M}$, with $\boldsymbol {\breve\zeta}^{\mu a} \in \mathbb{R}^N$, where $\mu$ labels the class each item belongs to and $a$ labels the items pertaining to the same class\footnote{Here we assume that classes display the same size to simplify the presentation, but this constraint can be relaxed.}. In the following we will consider as examples of structured datasets the MNIST \cite{MNIST} and the fashion-MNIST \cite{FMNIST}, each counting $K=10$ classes and $M=6000$ items per class, also, items are made of $N=28 \times 28 = 784$ pixels each.
\newline
As mentioned in the main text, before designing the network and setting the related weights, data must be pre-treated as we are going to explain.

First we binarize data, namely for each pixel $i$ we evaluate the average $\bar{\bar{\zeta_i}}$ over $\mu$ and over $a$, that is 
\begin{equation}
\bar{\bar{\zeta_i}} : = \frac{1}{K M} \sum_{\mu=1}^K \sum_{a=1}^M \zeta_i^{\mu a} 
\end{equation}
and set 
\begin{equation}
\zeta_i^{\mu a} = 
\begin{cases}
+1 ~~\textrm{if}~~ \breve{\zeta}_i^{\mu a} > \bar{\bar{\zeta_i}}\\ 
-1 ~~\textrm{if}~~ \breve{\zeta}_i^{\mu a} <\bar{\bar{\zeta_i}}
\end{cases}.
\end{equation}
The original dataset $\{\boldsymbol {\breve\zeta}^{\mu a}\}_{\mu=1,...,K}^{a=1,...,M}$ has therefore been mapped into the binary dataset $\{\boldsymbol \zeta^{\mu a}\}_{\mu=1,...,K}^{a=1,...,M}$.

\subsection{Ultrametricity in the structured datasets}\label{subsec:GG}
Given the dataset $\{\boldsymbol \zeta^{\mu a}\}_{\mu=1,...,K}^{a=1,...,M}$, we start inspecting its structure by tools inspired by the statistical mechanics of glassy systems. 
More precisely, looking at the items pertaining to the same class as different \emph{replicas}, we aim to evidence any signature of replica-symmetry or replica-symmetry-breaking.
A straightforward approach consists in measuring the distribution of the replica overlaps that is estimated by a MC routine as explained hereafter.
\newline
In each MC step, we draw a class $\mu$ uniformly in the range $\{1, ..., K\}$ and in that class we draw two labels $(a, b)$ uniformly and without repetition in the range $\{1, ..., M\}$; to simplify the notation let us take $(a, b) = (1, 2)$ without loss of generality.
With these extracted items we evaluate the replica overlap
\begin{equation}
q_{12} = \frac{1}{N} \sum_{i=1}^N \zeta_i^{\mu 1} \zeta_i^{\mu 2} .
\end{equation}
We repeat the operation $1.5 \times 10^4$ times and we build the histogram represented in panels $a,b,c$ of Fig.~\ref{fig:GG}, which provides our empirical representation of the overlap distribution $\mathcal P(q)$. 
To check the robustness of results versus the item size, we mimicked a finite-size scaling by progressively drop out a subset of the pixels making up the items selected.
More precisely, denoting with $\mathbb{1}$ the characteristic function, for any pixel $i$ we introduce $N^+_i = \sum_{\mu, a} \mathbb{1} [\zeta_i^{\mu a} =1]$ and $N^-_i = \sum_{\mu, a} \mathbb{1} [\zeta_i^{\mu a} =-1]$,
whence we evaluate $r_i = \frac{|N_i^+ - N_i^- |}{N_i^+ + N_i^-}$, representing the pixel quality.
Then, pixels are ranked based on their quality and we select the first quartile to make up the low-size sample, the second quartile to make up the medium-size sample and finally the fourth quartile (namely all pixels are retained) to make the large-size sample. 
Remarkably, unlike the random dataset \footnote{which has been generated in order to satisfy this two lowest order moments $\langle\frac{1}{N}\sum_{i=1}^N\zeta^{\mu a}_i\rangle = 0$ , $\langle\frac{1}{N}\sum_{i=1}^N\zeta^{\mu a}_i\zeta^{\nu b}_i\rangle = \delta_{\mu\nu}[\delta_{ab} + (1-\delta_{ab})r^2]$} that gives rise to a replica-symmetric overlap distribution (see Figure \ref{fig:GG}$a$), the histograms obtained for MNIST and fashion-MNIST (see Figure \ref{fig:GG}$b,c$) both exhibit a bimodal shape whose broadness does not shrink as the size of the dataset is enlarged and, for the latter, a central plateau also emerges as the item size $N$ gets larger.

\begin{figure}[tb]
\noindent \begin{centering}
\includegraphics[width=0.65\textwidth]{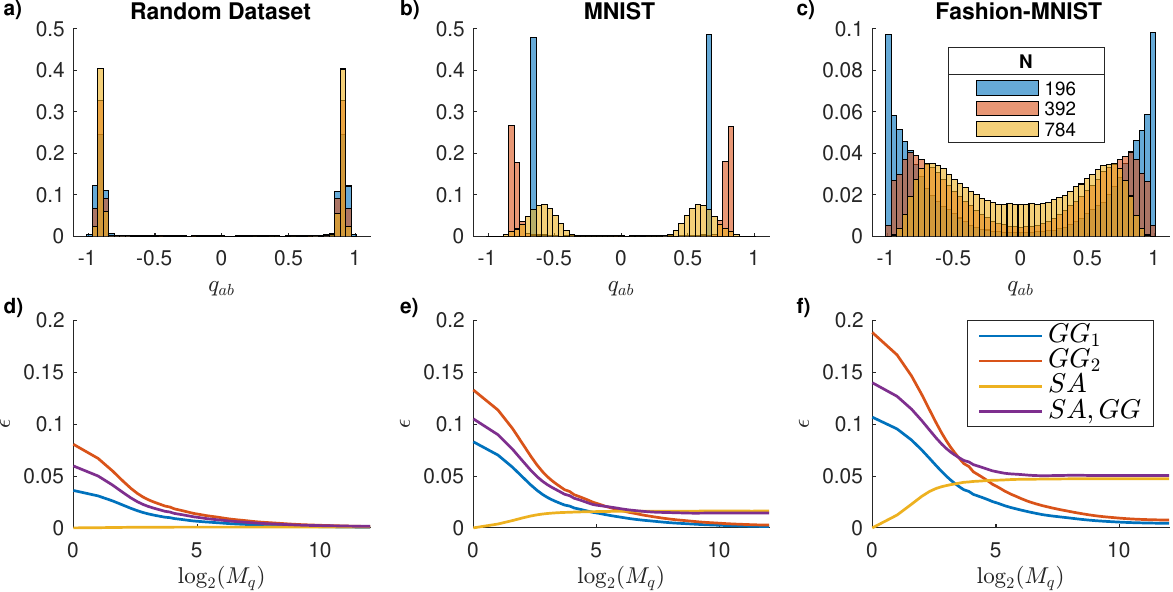}
\par\end{centering}
\caption{\label{fig:GG}Evidence of replica-symmetry breaking in datasets. In the first row we compare the empirical overlap distribution $\mathcal P(q)$ obtained for the random (panel $a$), the MNIST dataset (panel $b$), and the fashion-MNIST (panel $c$) datasets; three different item sizes are also considered, as explained by the legend. From left to right, we move from a replica-symmetry scenario where $\mathcal P(q)$ exhibits two peaks that get sharper and shaper as the item size is larger and larger, to a replica-symmetry-breaking scenario where $\mathcal P(q)$ is bimodal but with increasing broadness mirroring Parisi's overlap distribution \cite{MPV}.
In the second row we compare the value of the quantities introduced in eqs.~(\ref{a})-(\ref{b}), obtained for the random (panel $d$), the MNIST (panel $e$), and the fashion-MNIST (panel $f$) datasets. Again, from left to right, we move from a replica-symmetry scenario where the self-averaging relations hold and the Ghirlanda-Guerra relations (corresponding to trivial identities) are fast vanishing, to a replica-symmetry-breaking scenario where the self-averaging relations do not hold and the Ghirlanda-Guerra relations (this time not trivial) are also satisfied. We point out that, although these plots are built by aggregating results pertaining to all classes (for each $\mu = 1,\cdots,K$), we checked that they are in agreement with single-class results,  in fact, the results have been aggregated in order to increase the statistics at our disposal.}
\end{figure}

To further inspect this point, we analyze another feature characterizing ultrametricity induced by replica symmetry breaking, namely the validity of relations  known as ``Ghirlanda-Guerra identities'' \cite{GG,AC,P,cicio,Silvio} and that read as follows:
\begin{eqnarray}
\label{eq:GG1}
\lim_{N \to \infty} \left( \langle q_{12}^{4}\rangle-2\langle(q_{12}q_{13})^{2}\rangle+\langle q_{12}^{2}\rangle^{2} \right) &=&0\\
\label{eq:GG2}
\lim_{N \to \infty} \left( \langle q_{12}^{4}\rangle-3\langle(q_{12}q_{34})^{2}\rangle+2\langle q_{12}^{2}\rangle^{2}  \right) &=&0,
\end{eqnarray}
where $q_{ab}$ is the standard replica between two spin configurations labelled as, respectively, $a$ and $b$, sampled from the same Boltzmann-Gibbs distribution (see Definition \ref{def:q}), and the brackets $\langle \cdot \rangle$ denote the Boltzmann-Gibbs expectation. 
\newline
On the other hand, for configurations exhibiting replica symmetry we expect that
\begin{equation}  \label{eq:SA}
\lim_{N \to \infty}\left(\langle q_{12}^{2}\rangle-\langle q_{12}\rangle^{2}\right)=0, \ \ \lim_{N\to\infty} \left( \langle q_{12}^{4}\rangle-\langle q_{12}^{2}\rangle^{2}\right)=0.
\end{equation}
We now evaluate the validity of the previous equations for overlaps built over the dataset items. Again, we rely on MC estimates as described hereafter. 
In each MC step, we draw a class $\mu$ uniformly in the range $\{1, ..., K\}$ and in that class we draw four labels $(a, b, c, d)$ uniformly and without repetition in the range $\{1, ..., M\}$; to simplify notation let us take $(a, b, c, d) = (1, 2, 3, 4)$ without loss of generality.
From the extracted sample $\{ \boldsymbol \zeta^{\mu 1}, \boldsymbol\zeta^{\mu 2}, \boldsymbol\zeta^{\mu 3},\boldsymbol \zeta^{\mu 4}\}$ we evaluate the replica overlaps
\begin{equation}
q_{12} = \frac{1}{N} \sum_{i=1}^N \zeta_i^{\mu 1} \zeta_i^{\mu 2} 
\end{equation}
and analogously for $q_{13}, q_{14}, q_{24}$.
We repeat the procedure $M_q$ times, collect these overlaps and finally obtain averages of overlaps and overlaps correlations that are denote by the brackets $\langle \cdot \rangle$.
% over these trials to obtain $\langle q_{12} \rangle, \langle q_{13} \rangle, \langle q_{14} \rangle, \langle q_{24} \rangle$.
These quantities are then combined to build the following quantities
\begin{eqnarray}
\label{a}
\epsilon_{GG1} & = & \langle q_{12}^{4}\rangle-2\langle(q_{12}q_{13})^{2}\rangle+\langle q_{12}^{2}\rangle^{2}\\
\epsilon_{GG2} & = & \langle q_{12}^{4}\rangle-3\langle(q_{12}q_{34})^{2}\rangle+2\langle q_{12}^{2}\rangle^{2}\\
\epsilon_{SA} & = & \langle q_{12}^{4}\rangle-\langle q_{12}^{2}\rangle^{2}\\
\label{b}
\epsilon_{SA,GG} & = & 2\langle(q_{12}q_{13})^{2}\rangle-3\langle(q_{12}q_{34})^{2}\rangle+\langle q_{12}^{4}\rangle
\end{eqnarray}
The first three equations stem directly from the Ghirlanda-Guerra relations \eqref{eq:GG1}-\eqref{eq:GG2} and from the self-averaging property \eqref{eq:SA}; the forth relation was obtained by combining
\begin{equation}
\epsilon_{GG1}-\epsilon_{GG2}=-2\langle(q_{12}q_{13})^{2}\rangle+3\langle(q_{12}q_{34})^{2}\rangle-\langle q_{12}^{2}\rangle^{2}
\end{equation}
and reconstructing $\epsilon_{SA}$ as
\begin{eqnarray}
\nonumber
\epsilon_{SA}&=&\epsilon_{GG1}-\epsilon_{GG2}+2\langle(q_{12}q_{13})^{2}\rangle-3\langle(q_{12}q_{34})^{2}\rangle+\langle q_{12}^{4}\rangle\\
&=&\epsilon_{GG1}-\epsilon_{GG2}+\epsilon_{SA,GG}.
\end{eqnarray}
Therefore, any evidence of non-vanishing $\epsilon_{SA}$ and $\epsilon_{SA,GG}$ could be interpreted as a signature of replica-symmetry-breaking and, in that case, non-null values for $\epsilon_{GG1}$ and $\epsilon_{GG2}$ would suggest a non-trivial information organization.
\newline
Our empirical estimates for $\epsilon_{GG1}, \epsilon_{GG2}, \epsilon_{SA}, \epsilon_{SA,GG}$ are shown in Fig.~\ref{fig:GG}$d$-$f$ versus the number of MC steps, namely versus the number of extracted items on a logarithmic scale.
First, we notice that, as expected, for $M_q=1$ the SA ansatz is exact.
Moreover, for the random dataset the replica-symmetry scenario holds, while for structured datasets we see that, as $M_q$ gets larger, and therefore as our estimates get more and more reliable since based on larger and larger samples, $\epsilon_{SA}$ and $\epsilon_{SA,GG}$ tend to settle on non-null values that approximately coincide, while $\epsilon_{GG1}$ and $\epsilon_{GG2}$ tend to vanish.

Summarizing these empirical findings, we can speculate that:
\newline
- the replica-symmetry characterizing the structureless datasets guarantees that a unique hidden layer suffices to classify (note that classification is intrinsically a  {\em replica symmetric concept}), hence the grandmother-cell setting stands alone in that simple limit.
\newline
- the replica-symmetry-breaking characterizing the structured datasets -- that we infer by the lack of self-averaging in the overlap distribution and by the validity of ultrametric identities -- suggests that the grandmother-cell alone is no longer enough for structured datasets and hints at a generalization. 
\newline
How generalizing? Let us focus on the MNIST case as a practical example (the same holds also for other datasets): by construction, the last layer should be composed of $10$ neurons (one per archetype, namely one per digit) and that layer performs the replica-symmetric classification. Therefore, we must introduce in the machine architecture further internal layers, between the input layer and the final replica-symmetric layer, where the assumption that any example pertaining to the same digit (archetype) lights a unique hidden neuron does not hold any longer. In other words, we must introduce at least a 1-RSB extra hidden layer, where there are a few 1-RSB-archetypes for the same digit (e.g., a vertical line, a right-oriented oblique line and a left-oriented oblique line, that are all corresponding to the number $1$ in the last layer where replica symmetric classification takes place, etc.), see Fig.~\ref{fig:architettura} for a sketch. We discuss the consequences of this idea in the next subsection.   

\begin{figure}
\begin{center}
\includegraphics[scale=0.13]{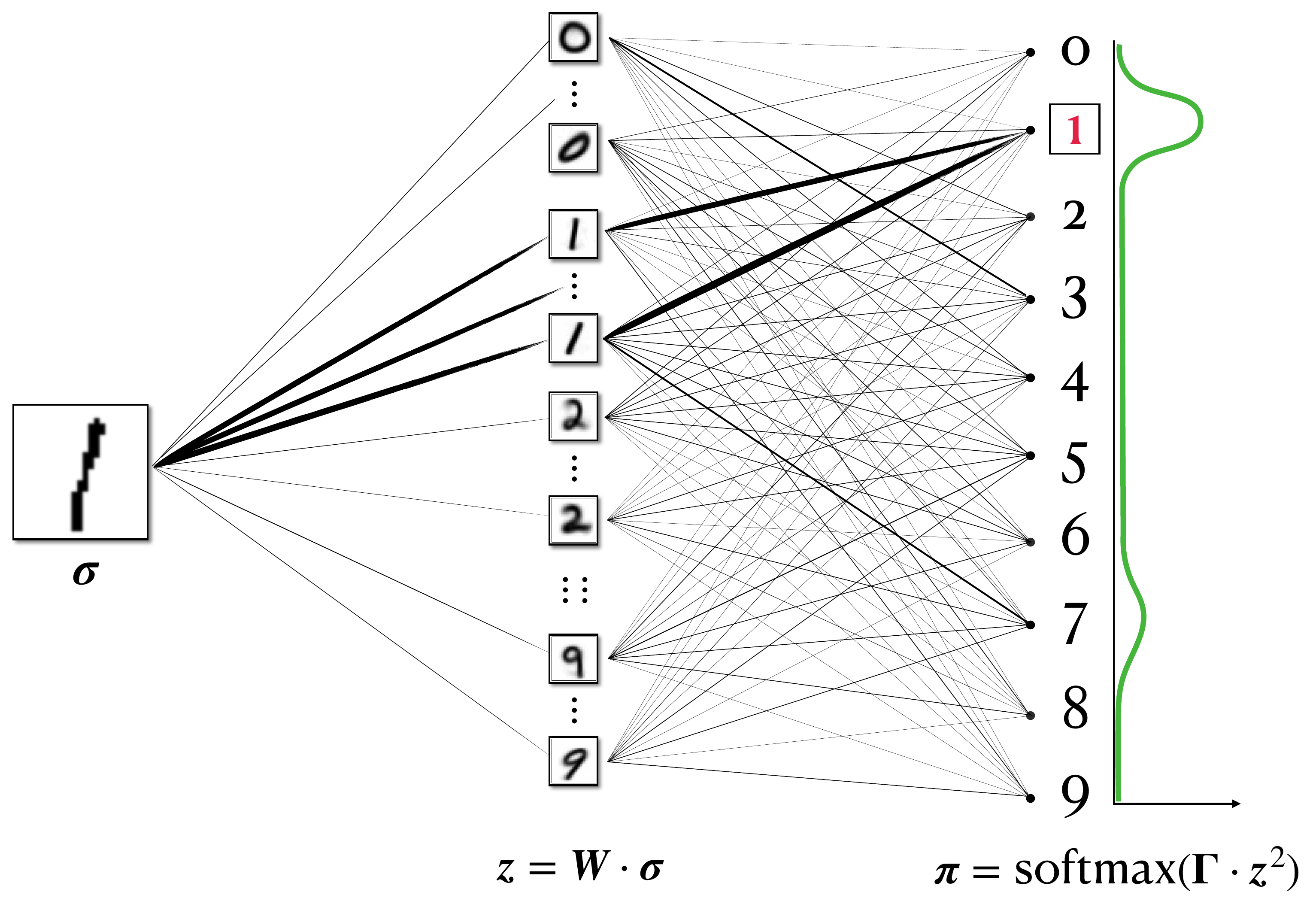}
\caption{Schematic representation of a three-layer RBM for the MNIST dataset. From left to right: visible layer $\boldsymbol \sigma \in \{-1, +1\}^N$ receiving digits to be classified; hidden layer  $\boldsymbol z \in \mathbb R^{\hat K}$ where each node corresponds to a pseudo archetype as sketched; softmax layer $\boldsymbol \pi \in [0, +1]^K$ for classification. 
}\label{fig:architettura}
\end{center}
\end{figure}

\subsection{Ultrametric generalized grandmother-cell ansatz}
We now proceed with the pre-treatment of data to reach a generalized setting.

First, for each class we assess the related quality denoted as $r_{\mu}$ and obtained as follows: 
given $i$ and $\mu$, we count the number of positive $N^+_{i \mu}$ and negative $N^-_{i \mu}$ pixels over the class sample $\{\zeta_i^{\mu a}\}_{a=1,...,M}$ and set 
\begin{equation}
r_{\mu} = \frac{1}{N} \sum_{i=1}^N \frac{|N^+_{i \mu} - N^-_{i \mu} |}{N^+_{i \mu} + N^-_{i \mu}}.
\end{equation}
Exploiting the scaling \eqref{eq:scalscal}, for an arbitrary $\rho \in \mathbb R$, we can also introduce 
\begin{equation}
M_{\mu} = \frac{1 - r_{\mu}^2}{\rho r_{\mu}^2}
\end{equation} 
that represents the number of items needed to correctly process the $\mu$-th class.
\newline
Further, we define the centroids $\hat{\boldsymbol \zeta}^{\mu \ell} \in \mathbb R^N$, for $\ell=1, ..., M_{\mu}$, obtained by applying the $k$-means clustering algorithm within the $\mu$-th sample and setting $k=M_{\mu}$.

\begin{figure}[bt]
\noindent \begin{centering}
\includegraphics[width=0.6\textwidth]{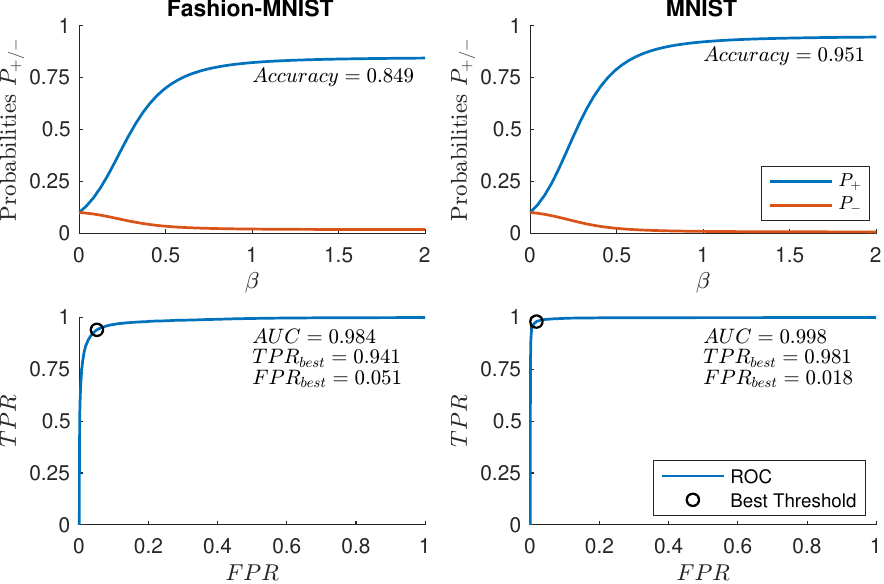}
\par\end{centering}
\caption{The ultrametric RBM (URBM) has been tested on two structured datasets, namely Fashion-MNIST (left) and MNIST (right). In the upper panels the classification probabilities are summarised as $P_+$, the probability of correctly classifying a given example ($\pi_{\nu}(\boldsymbol \sigma=\tilde{ \boldsymbol \zeta}^{\nu})$), and $P_-$, the class-specific probability of misclassifying a given example ($\pi_{\mu}(\boldsymbol \sigma=\tilde{ \boldsymbol \zeta}^{\nu})$), and these are shown by varying the parameter $\beta$. These probabilities are normalised as $P_+ + 9 P_- = 1$ since in total there are $10$ different classes for both datasets. Remarkably, the URBM achieves an accuracy of $95\%$ on the MNIST dataset and, as expected, for the more-challenging Fashion-MNIST dataset (which exhibits stronger RSB-like effects), the accuracy is reduced to $85\%$. In the lower panels the performance of these classifiers is more thoroughly analysed by the classical ROC approach; both datasets at the optimal threshold show excellent TPR/FRP ratios and AUC values which are very close to the ideal value of $1$.}\label{fig:MNIST}
\end{figure}

We are now ready to design the RBM to classify this dataset, that shall be made of $i.$ an input visible layer of size $N = 784$, corresponding to the number of pixels of each item and whose neurons are denoted as $\boldsymbol \sigma = (\sigma_1, ..., \sigma_N)$, $ii.$ a hidden layer of size $\hat{K} = \sum_{\mu} M_{\mu}$, corresponding to the overall number of centroids and whose neurons are denoted as  $\boldsymbol z = (z_1, ..., z_{\hat K})$, $iii.$ an output visible layer of size $K$ and whose neurons are denoted as $\boldsymbol \pi =(\pi_1, ..., \pi_K)$, see Fig.~4 in the main text.
As explained below, the latter is a softmax layer
%, namely $\boldsymbol \pi = \textrm{softmax} [\boldsymbol  \Gamma \cdot \boldsymbol z^2] $, where $\boldsymbol \Gamma  \in \mathbb{R}^{\hat K \times K}$ has to be determined by handling the training set (vide infra), 
in such a way that $\pi_{\mu}$ represents the probability that the input supplied to the network belongs to the $\mu$-th class. \\
Let us now specify the activation functions of the hidden and of the output neurons. 
Inputting a generic item ${\boldsymbol \zeta}^{\mu}$ pertaining to the $\mu$-th class, we initialize the visible layer as $\boldsymbol \sigma = {\boldsymbol \zeta}^{\mu}$, and, accordingly, the hidden layer is set as 
$\boldsymbol z = \boldsymbol W \cdot \boldsymbol \sigma = \boldsymbol W \cdot {\boldsymbol \zeta}^{\mu} $, where the weight matrix $\boldsymbol W \in \mathbb{R}^{\hat K \times N}$ is built analogously to the random case described in the main text, but here we refer to the $\hat K$ centroids instead of the empirical averages, thus the matrix has rows $\boldsymbol W_{j} = \hat{ \boldsymbol \zeta}^{j}/N$ for $j=1,...,\hat K$ (here with a slight abuse of notation, the couple $(\mu,\ell)$ indexing the centroids has been recast in the simple label $j$).
Finally, as anticipated, the output layer is determined as $\boldsymbol \pi = \textrm{softmax} (\boldsymbol  \Gamma \cdot \boldsymbol z^2)$, where $\boldsymbol \Gamma  \in \mathbb{R}^{  K\times \hat K}$, and returns the classification probability for the input; the square of the hidden neuron amplitude $\boldsymbol z^2 = \boldsymbol z \odot \boldsymbol z$ is an Hadamard product (i.e. it is carried out element-wise) and it is meant to preserve the gauge invariance characterizing the model.
The matrix $\boldsymbol \Gamma$ has to be determined by handling the sample $\{\boldsymbol \zeta^{\mu a}\}_{\mu=1,...,K}^{a=1,...,M}$: the procedure can indeed be looked at as a training and it can be accomplished by directly accounting for the whole sample in a simple algebraic passage as detailed hereafter.
%
%As for the matrix $\Gamma \in \mathbb{R}^{\hat K \times K}$, it must be properly trained as described hereafter:
\newline
Basically, our goal is to obtain a matrix that, applied to $\boldsymbol z^{2} = (\boldsymbol W \cdot {\boldsymbol \zeta}^{\mu})^{2}$, returns a one-hot vector $\boldsymbol c =  \boldsymbol \Gamma \cdot  \boldsymbol z^2 = \boldsymbol \Gamma \cdot  (\boldsymbol  W  \cdot {\boldsymbol \zeta}^{\mu } )^2 \in \{0, 1\}^K$ whose unique non-null entry is the one corresponding to the wanted class $\mu$.\\
In order to account for the whole training set, we introduce a capital notation, namely\\ $\boldsymbol \Sigma = (\zeta^{11}, \zeta^{12},...,\zeta^{1M}; \zeta^{21}, \zeta^{22},...,\zeta^{2M};$ $...; \zeta^{K1},\zeta^{K2},..., \zeta^{KM}) \in \{-1, +1\}^{N \times ( K M)}$, $\boldsymbol Z = \boldsymbol W \cdot \boldsymbol \Sigma \in \mathbb {R}^{\hat K \times (KM)} $ and $\boldsymbol C = \boldsymbol \Gamma \cdot \boldsymbol Z^2 \in \{0, 1\}^{K \times (K M)}$ defined in such a way that $C_{\mu, (\nu,\ell)} = \delta_{\mu,\nu}$, then we solve 
 \begin{equation}
 \boldsymbol \Gamma \cdot \boldsymbol Z^2 =  \boldsymbol \Gamma \cdot (\boldsymbol W \cdot \boldsymbol \Sigma)^2 = \boldsymbol C.
 \end{equation}
%
%, where $Z^2 = (W \Sigma)^2$ and $C \in \mathbb {0, 1}^{K \times M'}$, where $M'$ is the number of examples that we decide to use for training the network. (60 000). 
%where the square is applied element-wise. \textcolor{blue}{anche sopra?}
%$\Sigma \in {-1, +1}^{N \times M'}$
%$C$ is an indicator function that describes the belonging of a certain example to a certain class, that is
%$C_{ij} = 1$ if example $j$ belong the the class $i$ and viceversa $C_{ij} = 0$.
%This is a one-hot vector representation of the examples. 

for $\boldsymbol \Gamma$ by applying the pseudo-inverse rule as
\begin{equation}
\boldsymbol \Gamma  = \boldsymbol C (\boldsymbol Z^2)^T \cdot [\boldsymbol Z^2 \cdot (\boldsymbol Z^2)^T]^{-1}.
\end{equation}
In this way $\boldsymbol \Gamma$ is tailored in such a way that  $\boldsymbol c = \boldsymbol \Gamma \cdot \boldsymbol z^2$ is approximately binary and has entry approximately equal to $1$ for the correct class and approximately equal to zero for the remaining, hence generalizing the grandmother cell setting described in Sec.~\ref{sec:RBM}. 
Notice that the obtained matrix $\boldsymbol \Gamma$ has size $K \times \hat K$ and can therefore be directly implemented in the RBM receiving single items as inputs.

\begin{figure} [tb]
\begin{center}
\includegraphics[scale=0.6]{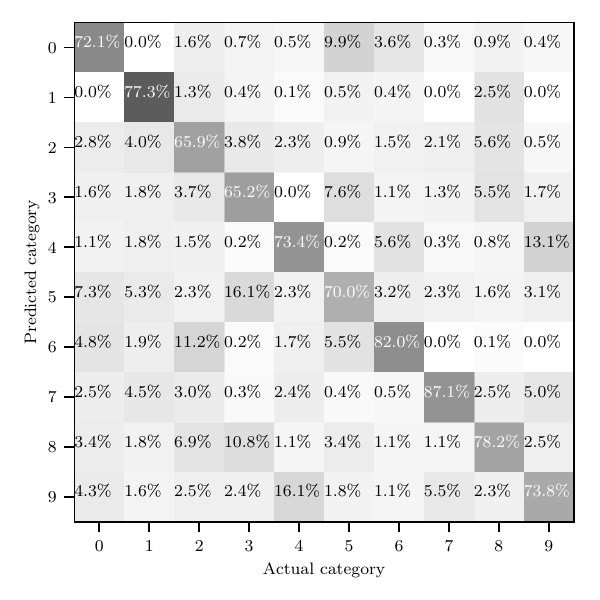}
\includegraphics[scale=0.6]{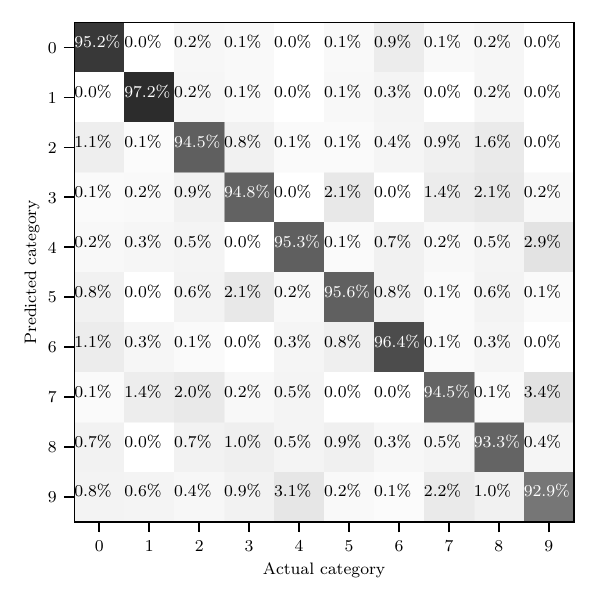}
\includegraphics[scale=0.6]{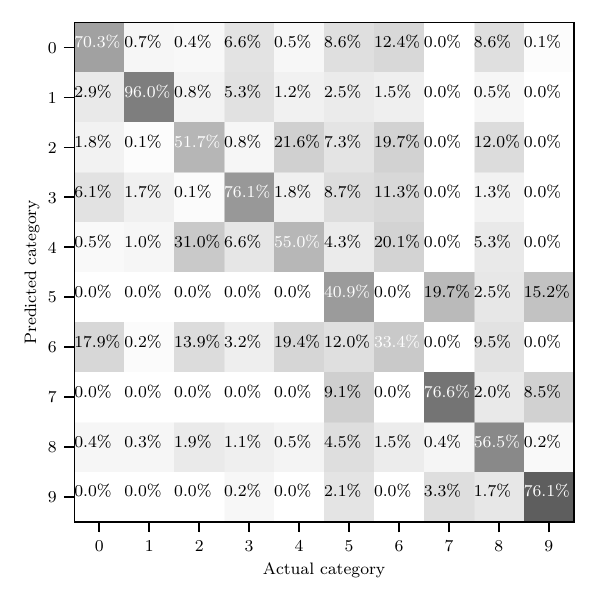}
\includegraphics[scale=0.6]{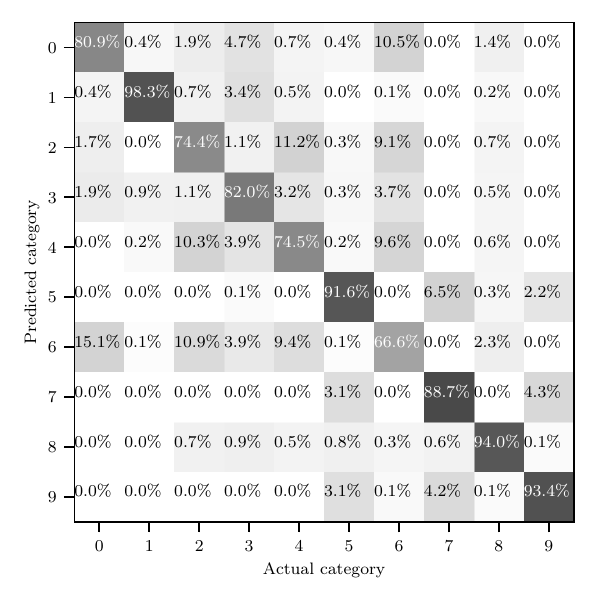}
\caption{Confusion matrices for the MNIST (first row) and the fashion-MNIST datasets (second row) obtained by a RBM trained under the grandmother-cell protocol and made up of one hidden layer of size $K=10$ corresponding to the number of classes (first column) and made up of one hidden layer of size $\hat{K}=300$ corresponding to the overall number of pseudo archetypes identified (second column).}\label{fig:confucio}
\end{center}
\end{figure}

Before proceeding we also recall that in the softmax layer is applied with a free parameter denoted with $\beta$, that is
\begin{equation}
\pi_{\nu} = \frac{e^{\beta c_{\nu}}}{\sum_{\mu=1}^Ke^{\beta c_{\mu}} },
\end{equation}
in such a way that $\beta$ tunes the broadness of the distribution (playing a role similar to the temperature in the statistical mechanics framework). 
Notice that if we let $\beta \to \infty$ the softmax collapses to a delta function peaked at $\textrm{argmax}_{\mu=1,...,K}(\boldsymbol \Gamma \boldsymbol z^2)$, while if we let $\beta \to 0$ the softmax collapses to a uniform distribution.

We trained the RBM as specified above for the MNIST and the fashion-MNIST dataset and then evaluated their accuracy over the test sample; results are shown in Fig.~\ref{fig:MNIST}.
In particular, in panels $a$ and $b$, we present the probability $\pi_{\nu}(\boldsymbol \sigma=\tilde{ \boldsymbol \zeta}^{\nu})$ that the network classifies correctly a test example and the probability $\pi_{\mu}(\boldsymbol \sigma=\tilde{ \boldsymbol \zeta}^{\nu})$ that the network misclassifies; in both cases probabilities are averaged over the test sample. Notice that, as $\beta$ gets larger, the former grows and the latter decreases, anyhow the former is always larger than the latter -- the two trivially coincide when $\beta=0$.
The accuracy of this machine, defined as the ratio of correct answers versus the number of trials, coincides with the test-sample average of $\pi_{\nu}(\boldsymbol \sigma=\tilde{ \boldsymbol \zeta}^{\nu})$ as $\beta \to \infty$ and this quantity is as well reported in Fig.~\ref{fig:MNIST}.
Further, we built the ROC curve and obtained an area under the curve which is approximately $1$ for MNIST and approximately $0.98$ for Fashion-MNIST.

Finally, it is instructive to compare the confusion matrices obtained for the same structured datasets exploiting a simple (RS) RBM and an ultrametric (RSB) RBM, see Fig.~\ref{fig:confucio}. Again, we notice that the performance for the fashion-MNIST is slightly lower than that for the MNIST and this may be related to a larger extent of internal structure as discussed in Sec.~\ref{subsec:GG}. However, in both cases, the accuracy resulting from the URBM is significantly larger than that resulting from the bare RBM.
indeed, for the former we reach an accuracy of about $95 \%$ for MNIST and $84 \%$ for fashion-MNIST, to be compared with, respectively,  $75 \%$ and $63 \%$ obtained for the simple RBM.

\end{document}